\pdfoutput=1

\documentclass[acmsmall]{acmart}

\usepackage[T1]{fontenc}
\usepackage[utf8]{inputenc} 
\usepackage{upgreek}
\usepackage{paralist}
\usepackage{enumitem}
\usepackage{url}
\usepackage{stmaryrd}
\usepackage{graphicx}
\usepackage{todonotes}
\usepackage{xspace}
\usepackage{xfrac}
\usepackage{booktabs}

\usepackage{amsmath}
\usepackage{amssymb}
\usepackage{amsthm}
\usepackage{stmaryrd}
\usepackage{mathtools}
\usepackage{color}
\usepackage{mathpartir}

\usepackage{footnote}
\usepackage{subcaption}
\usepackage[capitalise]{cleveref}
\makesavenoteenv{tabular}

\newtheorem*{theorem*}{Theorem}
\newtheorem{theorem}{Theorem}
\newtheorem{lem}{Lemma}
\newtheorem*{remark*}{Remark}
\newtheorem{example}{Example}
\newtheorem{definition}{Definition}

\newcommand{\RR}{\mathbb{R}}
\newcommand{\NN}{\mathbb{N}}

\newcommand{\EE}{\mathbb{E}}
\newcommand{\RRnn}{\mathbb{R}_{\geq 0}}
\newcommand{\RRe}{\mathbb{R}_{\geq 0}^\infty}
\newcommand{\cI}{\mathcal{I}}

\newcommand{\pgcl}{\textsc{pGCL}}
\newcommand{\pwhile}{\textsc{pWhile}}
\newcommand{\prhl}{\textsc{pRHL}}
\newcommand{\eprhl}{$\mathbb{E}$\textsc{pRHL}}
\newcommand{\Kant}[1]{#1^\#}

\newcommand{\cA}{\mathcal{A}}

\newcommand{\cP}{\mathcal{P}}
\newcommand{\cR}{\mathcal{R}}
\newcommand{\cS}{\mathcal{S}}

\newcommand{\Abort}{\mathbf{abort}}
\newcommand{\Skip}{\mathbf{skip}}

\newcommand{\Assn}[2]{#1 \leftarrow #2}
\newcommand{\Rand}[2]{#1 \stackrel{\raisebox{-.25ex}[.25ex]%
 {\tiny $\mathdollar$}}{\raisebox{-.2ex}[.2ex]{$\leftarrow$}} #2}

\newcommand{\Cond}[3]{\mathbf{if}\ #1\ \mathbf{then}\ #2\ \mathbf{else}\ #3}
\newcommand{\Condt}[2]{\mathbf{if}\ #1\ \mathbf{then}\ #2}
\newcommand{\WWhile}[2]{\mathbf{while}\ #1\ \mathbf{do}\ #2}

\newcommand{\Exp}{\mathcal{E}}
\newcommand{\Inv}{\mathcal{I}}
\newcommand{\Pred}{\mathcal{P}}
\newcommand{\Dist}{\mathbf{Dist}}
\newcommand{\Mem}{\mathbf{State}}
\newcommand{\EXP}{\mathbf{Exp}}
\newcommand{\denot}[1]{\llbracket #1 \rrbracket}

\newcommand{\ind}[1]{[#1]}
\newcommand{\gcpsymbol}{\widetilde{\mathit{rpe}}}
\newcommand{\gcp}[2]{\gcpsymbol(#1 , #2)}
\newcommand{\sgcp}[2]{\mathit{rpe}(#1 ,\, #2)}
\newcommand{\wpe}[2]{\mathit{wpe}(#1 ,\, #2)}
\newcommand{\sidel}{\langle 1\rangle}
\newcommand{\sider}{\langle 2\rangle}
\newcommand{\wpel}[2]{\mathit{wpe}\sidel(#1 ,\, #2)}
\newcommand{\wper}[2]{\mathit{wpe}\sider(#1 ,\, #2)}
\newcommand{\dunit}[1]{\delta(#1)}
\newcommand{\dbind}[2]{\EE_{#1}[#2]}
\newcommand{\supp}{\text{supp}}
\newcommand{\wpc}{\mathit{wpe}}
\newcommand{\unif}[1]{U(#1)}
\newcommand{\Val}{\mathbf{Val}}
\newcommand{\Var}{\mathbf{Var}}

\newcommand{\ehl}[5]{\{ #3 \}\ #1 \sim_{#5} #2\ \{ #4\}}

\newif\ifcomments
\commentsfalse

\newcommand{\gb}[1]{\ifcomments\textit{\color{green}[GB] : #1}\fi}

\begin{document}

\title{A Pre-Expectation Calculus for Probabilistic Sensitivity}


\author{Alejandro Aguirre}
\affiliation{%
   \institution{IMDEA Software Institute/Universidad Polit\'ecnica de Madrid}}

\author{Gilles Barthe}
\affiliation{%
 \institution{Max Planck Institute for Security and Privacy/IMDEA Software Institute}}

\author{Justin Hsu}
\affiliation{%
 \institution{University of Wisconsin--Madison}}

\author{Benjamin Lucien Kaminski}
\affiliation{%
 \institution{RWTH Aachen University}}

\author{Joost-Pieter Katoen}
\affiliation{%
 \institution{RWTH Aachen University}}

\author{Christoph Matheja}
\affiliation{%
 \institution{RWTH Aachen University/ETH Zurich}}


\begin{abstract}
Sensitivity properties describe how changes to the input of a
program affect the output, typically by upper bounding the distance
between the outputs of two runs by a monotone function of the
distance between the corresponding inputs. When programs are
probabilistic, the distance between outputs is a distance between
distributions. The Kantorovich lifting provides a general way of
defining a distance between distributions by lifting the distance of
the underlying sample space; by choosing an appropriate distance on
the base space, one can recover other usual probabilistic distances,
such as the Total Variation distance. We develop a relational
pre-expectation calculus to upper bound the Kantorovich distance
between two executions of a probabilistic program. We illustrate our
methods by proving algorithmic stability of a machine learning algorithm,
convergence of a reinforcement learning algorithm, and fast mixing for card
shuffling algorithms. We also consider some extensions: proving lower bounds
on the Total Variation distance and convergence to the uniform distribution.
Finally, we describe an asynchronous extension of our calculus to reason about
pairs of program executions with different control flow.
\end{abstract}




\maketitle


\section{Introduction}
\label{sec:intro}

\emph{Sensitivity properties} describe how changes in program inputs
affect program outputs, with respect to particular distances on
program inputs and program outputs. By varying these distances,
sensitivity properties are relevant in many application areas,
including:
\begin{inparaenum}[(i)]
\item numerical computations, where distances are taken between real numbers,
\item numerical queries, where program inputs are databases, and the
  distance between them is the number of differing entries, and
\item learning algorithms, where the distance between two training
  sets is the number of differing examples, and the distance between
  outputs measures the difference in errors labeling unseen examples.
\end{inparaenum}
This paper is concerned with sensitivity properties of probabilistic
programs. As such programs return distributions over their output
space, the corresponding notions of sensitivity use distances over
distributions. The \emph{Total Variation} (TV) distance (a.k.a.\,
statistical distance), for example, is a widely used notion of distance that
measures the maximal difference of probabilities for two
distributions. One key benefit of the TV distance is that it is
defined for distributions over arbitrary spaces. However, it is
sometimes desirable to consider distances inherited from the
underlying space. It is common to consider classes of distances on
distributions that are obtained by lifting a distance on an
underlying space. This lifting is defined by the so-called
\emph{Kantorovich metric}, which yields a family of probabilistic
metrics obtained by lifting a distance $\Exp$ on a ground set $X$ to a
distance $\Kant{\Exp}$ on distributions over $X$. The class of
Kantorovich metrics cover many notions of distance, including the TV
distance which can be obtained by applying the Kantorovich lifting to
the discrete distance.

\subsubsection*{Approach}
We develop a \emph{relational pre-expectation calculus} for reasoning
about sensitivity of probabilistic computations under the Kantorovich
metric. Relational pre-expectations are mappings expressing a
quantitative relation (e.g., a distance or metric) between states, and
are modelled as maps of the form $\Mem \times \Mem \to [0, \infty]$.  Our
calculus takes as input a probabilistic program $c$ written in a core
imperative language and a pre-expectation $\Exp$ between output states
and determines a pre-expectation $\sgcp{c}{\Exp}$ between input states.
The calculus is a sound approximation of sensitivity, in the sense
that running the program $c$ on inputs at distance smaller than
$\sgcp{c}{\Exp}$ yields output distributions at distance smaller than
$\Kant{\Exp}$.

Technically, our calculus is inspired by early work on probabilistic dynamic logic due to \mbox{\citet{Kozen85}}
in which maps $\Exp\colon \Mem \to [0, \infty]$ serve as quantitative counterparts of Boolean predicates $P\colon \Mem \to \{ 0, 1 \}$.
\citet{DBLP:series/mcs/McIverM05} later coined the term \emph{expectation}---not to be confused with expected values---for such maps $\Exp$.
Moreover, they developed a weakest pre-expectation calculus for the probabilistic imperative language \pgcl. Their calculus was designed as a generalization of Dijkstra's weakest pre-conditions supporting both probabilistic and non-deterministic choice.  The basic idea is to define an operator
$\wpc(c, \Exp)$ that transforms an expectation $\Exp$ averaged over
the \emph{output} distribution of a program $c$ into an expectation
evaluated over the \emph{input} state. In this way, the expectation is
transformed by the effects of the probabilistic program in a backwards
fashion, much like how predicates are transformed through Dijkstra's
weakest pre-conditions.

Our pre-expectation calculus operates similarly, but---as it aims to measure distances between distributions of outputs in terms of inputs---manipulates \emph{relational} expectations instead. We next motivate
the need for relational expectations, and explain why they are challenging.

\paragraph*{Why do we need relational pre-expectations?}
The classical weakest pre-expectation calculus enjoys strong theoretical
properties: in particular, it is both sound and complete (in an extensional sense)
w.r.t.\, common program semantics (cf.\ \citet{GretzKM14}). 
Therefore, weakest pre-expectations can---in principle---be applied to reason about bounds on the Total Variation distance: Given a program $c$, 
\begin{inparaenum}[(i)]
\item take a copy $c'$ over a fresh set of program variables---e.g. if variable $x$ appears in $c$, substitute it by $x'$ in $c'$---and 
\item determine the weakest pre-expectation $\wpc(c;c', \Exp)$, where the expectation $\Exp$ measures the distance between variables in $c$ and their counterparts in $c'$.
\end{inparaenum}

However, this na\"ive approach is not practical for analyzing sensitivity:
the TV distance, for example, is defined as a
maximum of a difference of probabilities over all events of the output
space.
While the output space---and thus potentially the TV distance---is often \emph{unbounded}, the calculus of \citet{DBLP:series/mcs/McIverM05} is restricted to bounded expectations.
Moreover, the above approach pushes the difficulty of reasoning about sensitivity properties into the task of finding suitable invariants for probabilistic programs---a highly challenging task on its own.
In particular, finding invariants may involve reasoning about probabilistic independence, which is not readily available when using weakest pre-expectations.
%
%
In fact, mathematicians have long observed that reasoning about the TV
distance or the Kantorovich metric directly from their definition is
inappropriate. Rather, they rely on \emph{probabilistic
  couplings}~\citep{Villani08}, a mathematical tool for relating two
different distributions. Relational pre-expectations naturally connect
with probabilistic couplings, and capture well-established proof principles used by
mathematicians for reasoning about the TV distance.

\paragraph*{Challenges of relational pre-expectations}
Relational pre-expectations pose a number of specific challenges
compared to their unary counterpart. First, the Kantorovich distance
cannot be defined inductively on the structure of programs. More
specifically, the Kantorovich distance between two runs of $c; c'$ is
not a simple combination of the Kantorovich distances between two runs
of $c$ and two runs of $c'$ (we provide a counterexample in
Section~\ref{sec:couplings}). Instead, we define a pre-expectation
calculus $\gcp{c}{\Exp}$ that can compute a compositional
\emph{upper-bound} of the Kantorovich distance---this is sufficient
for proving sensitivity properties.

Second, proofs of soundness and continuity for our relational
pre-expectation calculus are significantly more involved than for the
usual weakest pre-expectation calculus, and use non-elementary results
from optimal transport theory. In particular, we are only able to prove
continuity for finitely supported distributions, and soundness for
discrete distributions.

Third, relational calculi are naturally better suited to reason about
two executions that follow the same control-flow. We offer useful
support for reasoning about executions with different control-flow,
through a careful generalization of the rules for conditionals and
loops. While our rules do not suffice for arbitrary examples (it
remains an open problem to develop relatively complete verification
approaches for relational properties of probabilistic programs), they
suffice for non-trivial examples that exhibit asynchronous behavior.



\paragraph*{Applications}
We demonstrate our technique on several applications. In our first application,
we formalize an \emph{algorithmic stability} property of machine-learning
algorithms. Informally, algorithmic stability describes how much the output
parameters from a learning algorithm are affected when one input training
example is changed; this notion of probabilistic sensitivity is known to imply
generalization and prevent overfitting~\citep{BousquetE02}. We use our calculus
for proving algorithmic stability of a commonly-used learning algorithm:
stochastic gradient descent (SGD). We use these examples to constrast our
approach with prior work.

Then, we consider a pair of applications showing convergence properties.  We
first formalize convergence of a reinforcement learning
algorithm~\citep{DBLP:journals/ml/Sutton88}, following a recent analysis
by~\citet{DBLP:conf/aistats/AmortilaPPB20}. 
Then, we show convergence and rapid mixing of several
card shuffling algorithms~\cite{aldous1983random}.  
We show that the TV distance between the outputs of two
probabilistic loops decreases to 0 as the number of loop iterations
increases---that is, the output distributions from any two inputs
\emph{converge} to the same distribution. Moreover, our technique is
precise enough to describe the rate of this convergence. Upper bounds
on convergence speed are key properties in algorithms that generate
samples form complex distributions, such as Markov Chain Monte
Carlo.

\paragraph*{Extensions: uniformity and lower bounds}
Next, we show how to formalize other properties complementing our
bounds on convergence rate. First, we prove with our system that some
card shuffling examples converge to the uniform distribution. Second,
we study lower bounds--a task already challenging in the non-relational $\wpc$ calculus~\cite{DBLP:journals/pacmpl/HarkKGK20}. 
While upper bounds on convergence speed are
often the main focus of formal analyses of probabilistic processes,
lower bounds are also useful to understand how far apart the output
distributions must be. The Monge-Kantorovich theorem provides a
general method for proving lower bounds on the Kantorovich metric by
using the unary $\wpc$ transformer from
\citet{DBLP:series/mcs/McIverM05}. However, proving lower bounds poses
some challenges~\cite{DBLP:journals/pacmpl/HarkKGK20}: we have to find a separating event and establish both
upper and lower bounds on its probability. We show how to solve these
challenges for card shuffling examples.

\paragraph*{Extensions: asynchronous reasoning}
Finally, we describe extensions to our calculus for asynchronous reasoning. We
show how to prove relational properties when pairs of program executions have
different control flow. We demonstrate our asynchronous extensions to reason
about a program generating a binomial distribution.

\subsubsection*{Contributions and outline.}
After introducing preliminaries on probability theory and the Kantorovich
distance (\S~\ref{sec:prelims}), we present our main contributions:
\begin{itemize}
  \item We define a sound, compositional, relational pre-expectation calculus
    for computing upper-bounds on the Kantorovich distance. We introduce
    convenient proof rules for sampling commands and loops, and we show that the
    core fragment of probabilistic relational Hoare logics, namely
    \prhl~\cite{BartheGZ09} and \eprhl~\cite{BartheEGHS18}, can
    be embedded into our calculus (\S~\ref{sec:couplings}).
  \item We apply our calculus to three case studies.
    As a warmup example, we use our calculus to provide a clean proof of
    algorithmic stability of stochastic gradient descent~\citep{HardtRS16}
    (\S~\ref{sec:ex-stability}).
    Second, we formalize convergence of TD(0), an algorithm from the Reinforcement
    Learning literature~\citep{DBLP:journals/ml/Sutton88} (\S~\ref{sec:ex-rl}).
    Third, we apply our calculus to show rapid convergence of random walks and card
    shuffling algorithms~\citep{aldous1983random}
    (\S~\ref{sec:ex-shuffle}).
  \item We show two complementary extensions to the previous examples: we use
    the weakest pre-expectation transformer from McIver and Morgan to compute
    lower bounds for the distance between distributions, and we use our calculus
    to show that the limiting distribution is uniform (\S~\ref{sec:extensions}).
  \item We present proof rules for reasoning about programs with asynchronous
    control flow (\S~\ref{sec:async}).
\end{itemize}
Finally, we survey related work (\S~\ref{sec:rw}) and conclude
(\S~\ref{sec:conc}).

\section{Mathematical Preliminaries} \label{sec:prelims}
We briefly recap the foundations required for 
relational reasoning about sensitivity properties:
\begin{inparaenum}[(1)]
  \item probability theory,
  \item probabilistic programming languages, and
  \item distances on probability distributions.
\end{inparaenum}
A comprehensive treatment of these topics is found, e.g., in the
textbooks~\cite{ash2000probability,DBLP:series/mcs/McIverM05,Villani08}.
\subsection{Basic probability concepts}
We will use sub-distributions to model probabilistic behavior.
A \emph{sub-distribution}
over a countable set $A$ is a function $\mu \colon A \to [0, 1]$ assigning a
probability to each element of $A$.  Probabilistic \emph{events} are
subsets $B \subseteq A$; the probability of $B$ is denoted
$\mu(B)$ and defined by $\mu(B) = \sum_{b \in B} \mu(b)$. 
The \emph{support} of $\mu$ is the set of all events $a \in A$ with $\mu(a) > 0$.
Moreover, we let
$|\mu|= \mu(A)$. As usual, the probabilities in any sub-distribution
must sum up to at most $1$: $|\mu| \leq 1$. We call $\mu$ a distribution if $|\mu| = 1$.
We let $\Dist(A)$ denote the set of \emph{sub-distributions} over $A$.

Given a sub-distribution $\mu \in \Dist(A_1 \times A_2)$ over a product, 
its left and right \emph{marginals}, $\pi_1(\mu)$ and $\pi_2(\mu)$, are sub-distributions 
over $A_1$ and $A_2$, respectively, which are given by
$
  \pi_1(\mu)(x_1) = \sum_{x_2 \in X} \mu(x_1, x_2),
$
and
$
  \pi_2(\mu)(x_2) = \sum_{x_1 \in X} \mu(x_1, x_2)\,.
$

The \emph{Dirac distribution} $\dunit{a} \in
\Dist(A)$ is the point distribution at $a \in A$, $\dunit{a}(a') = [a = a']$,
where the right-hand-side is an \emph{Iverson-bracket} which evaluates to $1$ if
the formula inside (in this case, $a =a'$) evaluates to true, and to $0$ otherwise. 
If $f \colon A \to \RRe$ is a function mapping into the extended reals, we can take its
\emph{expected value} $\EE_{\mu}[f]$  with respect to some sub-distribution $\mu \in \Dist(A)$:
$
  \EE_{\mu}[f] = \sum_{a \in A} f(a) \cdot \mu(a) .
$
If the sum diverges, the expected value is $\infty$.
We assume that addition and multiplication are extended in the natural way, with
the convention $0 \cdot \infty = \infty \cdot 0 = 0$.

\subsection{Programming language and semantics}
We work with a standard probabilistic imperative language
\pwhile. This language has commands defined by the following grammar:
\begin{align*}
	c &\coloneqq \Skip
         \mid \Assn{x}{e}
         \mid \Rand{x}{d}
         \mid c ; c
	       \mid \Cond{e}{c}{c}
         \mid \WWhile{e}{c}~.
\end{align*}
Variables $x$ are drawn from an arbitrary but finite set $\Var$ of variable
names. Expressions $e$ are largely standard, formed from variables and basic
operations (e.g., integer addition, boolean conjunction). To handle programs
with (static) arrays, we assume expressions include basic array operations for
accessing and updating. For instance, when $a$ is an array variable we have
syntactic sugar:
\[
	a[e] \triangleq \mathbf{Lookup}(a, e) \quad\text{\textit{(expression)}} \qquad\textnormal{and}\qquad
	\Assn{a[e]}{e'} \triangleq \Assn{a}{\mathbf{Update}(a, e, e')} \quad\text{\textit{(command)}}
\]
The random sampling command $\Rand{x}{d}$ takes a sample from some primitive
distribution $d$ and stores it in $x$. For simplicity, we assume that primitive
distributions do not have free program variables, and we interpret them as full
distributions $\denot{d} \colon \Dist(D)$ over some countable set $D$, possibly
different for different distributions. We will often use the uniform
distribution $\unif{S}$ when $S$ is a finite, non-empty set; for instance, for a
positive integer $N$ we will write $[N]$ for the set of integers $\{ 0, \dots, N
  - 1 \}$, so that $\Rand{x}{\unif{[N]}}$ samples each number in $[N]$ with
probability $1/N$ and stores it in $x$. The distributions can also be
parameterized by some more complex expression, for instance in $\Rand{x}{[y]}$
for a program variable $y$.

\pwhile\ programs transform \emph{states}, which are finite maps $s \colon \Var
\to D$; we write $\Mem$ for the set of all states.  The semantics of a program
$c$ is a map $\denot{c} \colon \Mem \to \Dist(\Mem)$ assigning a
sub-distribution over possible outputs to each input. For example, for the
random sampling command, we define 
\[
  (\denot{\Rand{x}{d}} s)(s') \triangleq 
  \begin{cases}
    s(d)(s'(x)) &: s(y) = s'(y) \text{ for all } y \neq x \\
    0 &: \text{otherwise}
  \end{cases}
\]%
The semantics of the remaining language constructs is standard and deferred to
the appendix. As we only work with discrete primitive distributions and states
have finitely many variables, output distributions programs always have
countable support. 

To express properties about pairs of states we use \emph{relational expectations},
which are maps of type $\Mem \times \Mem \to \RRe$; we write $\EXP$ for the set
of all relational expectations. 
This set is equipped with the pointwise order
inherited from the order on $\RRe$, i.e., $\Exp \leq \Exp'$ if and only if
$\Exp(s_1, s_2) \leq \Exp'(s_1, s_2)$ for all pairs $(s_1,s_2)$ of states. Since $\RRe$ is a
complete lattice and $\EXP$ has the pointwise order, $\EXP$ is also a complete
lattice; the top and bottom elements are the constant relational expectations
$\infty$ and $0$, which send all pairs of states to $\infty$ and $0$ respectively.

For denoting specific relational expectations, we borrow notation from
relational Hoare logic~\citep{Benton04}: We tag variables with $\sidel$ or
$\sider$ to refer to their value in the first or the second state, respectively.
For instance, $[x\sidel=x\sider]$ is a relational expectation encoding the
predicate $\lambda \langle s_1,s_2\rangle.~ [s_1(x)=s_2(x)]$.

\subsection{Distances between probability distributions}
Various notions of distances between distributions allow us to specify sensitivity
properties of probabilistic programs. A popular example is the following: 
\begin{definition}[Total Variation distance]
The Total Variation (TV) distance between $\mu_1, \mu_2 \in \Dist(X)$ is defined as:
  $
    \mathit{TV}(\mu_1,\, \mu_2) \triangleq \frac{1}{2} \sum_{x \in X} \bigl| \mu_1(x) - \mu_2(x) \bigr|~.
  $
\end{definition}
\noindent{}%
The term distance (or \emph{metric}) is justified as 
$\mathit{TV}(\mu_1,\mu_2)$ is symmetric, satisfies the triangle inequality, 
and maps to zero if and only if $\mu_1 = \mu_2$.
The normalization factor of $\smash{\tfrac{1}{2}}$ ensures that the
TV distance is within $[0, 1]$. Roughly speaking, the TV distance measures the largest
difference in probabilities of any event between two given distributions. 
%

Note that the TV distance does not require a metric space, i.e., the underlying set $X$ is not 
necessarily equipped with any metric.
If $X$ is a metric space, we can define:
\begin{definition}[Kantorovich distance]\label{def:kantorovich}
Let $X$ be a (extended) metric space with a distance
$\Exp \colon X\times X \to \RRe$. The \emph{Kantorovich distance} is a canonical
lifting of $\Exp$ to a function $\Kant{\Exp} \colon \Dist(X) \times \Dist(X) \to
\RRe$ that defines a metric on $\Dist(X)$. This distance is defined as
\[
  \Kant{\Exp}(\mu_1, \mu_2) = \inf_{\mu \in \Gamma(\mu_1, \mu_2)} \EE_{\mu}[ \Exp ], 
\]
where $\Gamma(\mu_1, \mu_2)$ is the set of \emph{probabilistic couplings} of $\mu_1, \mu_2$, given by
\[
  \Gamma(\mu_1, \mu_2) = \{ \mu \in \Dist(X \times X) \mid \pi_i(\mu) = \mu_i,
  \text{ for } i = 1, 2 \} .
\]
The set $\Gamma(\mu_1, \mu_2)$ is non-empty provided $|\mu_1|=|\mu_2|$. Otherwise, $\Gamma(\mu_1, \mu_2) = \emptyset$
and $\Kant{\Exp}(\mu_1, \mu_2) = \infty$. 
\end{definition}%
\noindent{}%
The coupling-based definition of the Kantorovich distance is more abstract than
other distances between distributions, but its generality turns out to be a
strength. First, we can recover the TV distance as a lifting of the
discrete metric:
\begin{theorem}[Total variation and Kantorovich distance] \label{thm:tv}
  Let $\mu_1, \mu_2 \in \Dist(X)$ such that $|\mu_1|=|\mu_2|=1$.  
  If the discrete metric $\Exp \colon X \times X \to \{ 0, 1 \}$ is given by
  $\Exp(x_1, x_2) = [x_1 \neq x_2]$, then
  $
  \mathit{TV}\bigl(\mu_1,\, \mu_2\bigr) = \Kant{\Exp} \bigl( \mu_1,\, \mu_2 \bigr) .
  $
\end{theorem}%
%
%
%
%
\noindent{}Another advantage of the Kantorovich distance is that it is defined as an infimum.
For our goal of proving continuity, it suffices to compute an upper bound of the
distance, which corresponds to determining $\EE_{\mu}[ \Exp ]$ for some particular coupling $\mu$.

Traditionally, the definition of $\Kant{\Exp}$ is restricted to functions $\Exp$
defining a metric on $X$. However, the definition of $\Kant{\Exp}$ extends \emph{mutatis mutandis} to arbitrary functions $\Exp$. We abuse
terminology and use the term Kantorovich distance also in the more general case.
For instance, we can use this more general notion to bound the difference
between the expected values of two functions on the outputs of two program runs:%
\begin{theorem}[Absolute expected difference] \label{thm:abs-diff}
	Let $\mu_1,\mu_2\in\Dist(X)$ such that $|\mu_1|=|\mu_2|=1$,
        and let $f_1, f_2 \colon X \to \RRe$. Let $\Exp\colon
        X\times X \to \RRe$ be defined by
	$	\Exp(x_1, x_2) = |f_1(x_1) - f_2(x_2)|~.$
	Then
	$
		\bigl| \EE_{\mu_1}[ f_1 ] - \EE_{\mu_2}[ f_2 ] \bigr|
		\leq \Kant{\Exp}\bigl(\mu_1, \mu_2\bigr)~.
	$
\end{theorem}%
%
%
%
%
%
%
\noindent{}We can also obtain bounds on the TV distance when lifting other base
distances that assign a minimum, non-zero distance to all pairs of distinct
elements.%
\begin{theorem}[Scaled TV distance] \label{thm:scaled-tv}
  Let $\mu_1,\mu_2\in\Dist(X)$ with $|\mu_1|=|\mu_2|=1$, let $\Exp_\rho \colon X
  \times X \to [0, 1]$, and let $\rho \in \RR_{>0}$ be a strictly positive
  constant with $\Exp_\rho(x_1, x_2) \geq \rho \cdot [x_1 \neq x_2]$. Then,
  $
    \mathit{TV} \bigl( \mu_1,\, \mu_2 \bigr) \leq \frac{1}{\rho} \cdot \Kant{\Exp_\rho} \bigl( \mu_1,\, \mu_2 \bigr)~.
  $
\end{theorem}%
%
%
%

\section{Bounding expected sensitivity with relational pre-expectations} \label{sec:couplings}

As we have seen, the Kantorovich distance encompasses many specific distances on
distributions. To reason about probabilistic and expected sensitivity, we would
like to bound the Kantorovich distance between two output distributions in terms
of the distance between two program inputs. In this section, we develop a
relational pre-expectation operation to prove these bounds.

\subsection{A first unsuccessful attempt: a relational pre-expectation for exact bounds}
Since we want to reason about the Kantorovich distance lifting of a
relational expectation $\Exp \colon \Mem \times \Mem \to \RRe$ between output distributions of
a program $c$, an initial idea is to define a relational pre-expectation
operator $\sgcp{c}{\Exp}$ coinciding exactly with the Kantorovich distance:
\[
  \sgcp{c}{\Exp}(s_1,\, s_2) = \Kant{\Exp} \bigl( \denot{c}s_1,\, \denot{c}s_2 \bigr)~,
\]
and then prove bounds of the form $\sgcp{c}{\Exp_{out}} \leq \Exp_{in}$ in order
to bound the Kantorovich distance between outputs by some distance between
inputs. While this definition is appealing, it turns out to be inconvenient for
formal reasoning because it does not behave well under sequential composition:
the expected sequence rule
$
\sgcp{c;c'}{\Exp} = \sgcp{c}{\sgcp{c'}{\Exp}}
$
does \emph{not} hold. Roughly, this is because choosing local infima on each step does
not necessarily amount to a global infimum. In fact, in some cases \emph{no
local choice} amounts to a \emph{global infimum}.
%
%
\begin{example} \label{ex:non-comp}
  The \emph{Bernoulli} distribution $B(p)$ with bias $p$ returns $1$ with
  probability $p$ and $0$ with probability $1{-}p$. Consider the following
  programs:
  \begin{align*}
    c &= \Cond{b}{\Rand{x}{B(\sfrac{1}{2})}}{\Rand{y}{B(\sfrac{1}{2})}} \\
    c' &= \Cond{b}{\Rand{y}{B(\sfrac{1}{2})}}{\Rand{x}{B(\sfrac{1}{2})}}~.
  \end{align*}
  Moreover, consider the relational expectation
  $
  \Exp = [x\sidel \neq x\sider \vee y\sidel \neq y\sider].
  $
  If we fix $b\sidel = \mathsf{true}$ and $b\sider = \mathsf{false}$ throughout, then
  \[
	  \sgcp{c'}{\Exp}(s_1', s_2') = \inf_{\Gamma(\denot{\Rand{y}{B(\sfrac{1}{2})}} s_1', \denot{\Rand{x}{B(\sfrac{1}{2})}} s_2') } \EE[\Exp]~.
  \]
  To compute the above relational pre-expectation, we first need to understand the possible couplings.
  Hence, we compute the marginals of the involved distributions:
  %
  \begin{align*}
	  \mu_1 \triangleq \denot{\Rand{y}{B(\sfrac{1}{2})}}s_1' &= 
    \begin{cases}
	\tfrac{1}{2} \colon x \mapsto s'_1(x), y \mapsto 0 \\
      \tfrac{1}{2} \colon x \mapsto s'_1(x), y \mapsto 1
    \end{cases}
    \\
    \mu_2 \triangleq \denot{\Rand{x}{B(\sfrac{1}{2})}}s_2' &=
    \begin{cases}
      \tfrac{1}{2} \colon x \mapsto 0, y \mapsto s'_2(y) \\
      \tfrac{1}{2} \colon x \mapsto 1, y \mapsto s'_2(y)~.
    \end{cases}
  \end{align*}
  The marginal conditions for couplings (Def.~\ref{def:kantorovich}) then yield that any
  coupling in $\Gamma(\mu_1, \mu_2)$ is of the form
  \begin{align*}
          \mu_{\rho}(s_1,s_2) &= \rho \cdot [s_1(x) = s'_1(x) \wedge s_1(y) = 1] \cdot [s_2(x) = 1 \wedge s_2(y) = s'_2(y)] \\
          &+ \left( \tfrac{1}{2} - \rho \right) \cdot [s_1(x) = s'_1(x) \wedge s_1(y) = 1] \cdot [s_2(x) = 0 \wedge s_2(y) = s'_2(y)] \\
          &+ \left( \tfrac{1}{2} - \rho \right) \cdot [s_1(x) = s'_1(x) \wedge s_1(y) = 0] \cdot [s_2(x) = 1 \wedge s_2(y) = s'_2(y)] \\
          &+ \rho \cdot [s_1(x) = s'_1(x) \wedge s_1(y) = 0] \cdot [s_2(x) = 0 \wedge s_2(y) = s'_2(y)]~.
  \end{align*}
  %
  for some $0 \leq \rho \leq \tfrac{1}{2}$ and the previously fixed states $s'_1$ and $s'_2$. Hence, 
  \begin{align*}
    \EE_{\mu_\rho}[\Exp] &= \rho \cdot \ind{s'_1(x) \neq 1 \vee s'_2(y) \neq 1}
		      + \left( \tfrac{1}{2} - \rho \right)\ind{s'_1(x) \neq 0 \vee s'_2(y) \neq 1} \\
		      &+ \left( \tfrac{1}{2} - \rho \right)\ind{s'_1(x) \neq 1 \vee s'_2(y) \neq 0}
		      + \rho \cdot \ind{s'_1(x) \neq 0 \vee s'_2(y) \neq 0}~.
  \end{align*}
  Since $\sgcp{c'}{\Exp}$ takes the minimum over all couplings, i.e., the
  minimum over all $\rho \in [0, \tfrac{1}{2}]$, by simple computation we get
  that $\sgcp{c'}{\Exp}(s_1', s_2') = 1/2$, setting $\rho=1/2$ if $s_1'(x) =
  s_2'(y)$ and $\rho=0$ otherwise.
  Since $s_1'(x), s_2'(y)$ are sampled from $\denot{c} s_1$ and $\denot{c} s_2$, for any way to couple them
  $\sgcp{c}{\sgcp{c'}{\Exp}}(s_1, s_2) = \tfrac{1}{2} > 0$.
  However, $\denot{c; c'}s_1$ and $\denot{c;
  c'}s_2$ have the same marginal distributions for $(x, y)$ and thus distance $0$. Therefore, 
  \[
      0 = \sgcp{c; c'}{\Exp}(s_1, s_2) < \sgcp{c}{\sgcp{c'}{\Exp}}(s_1, s_2) = \tfrac{1}{2} ~.
  \]
\end{example}%
\noindent{}Fortunately, we generally do not need to compute the exact
Kantorovich distance to prove sensitivity properties: an upper bound suffices.
Since the Kantorovich distance is an infimum over \emph{all} couplings, we can
establish upper bounds by exhibiting a \emph{specific} coupling---of course, the
tightness of these upper bounds will depend on the particular coupling we chose.
Crucially, couplings \emph{can} be constructed compositionally: a coupling for a
sequential composition $c; c'$ can be obtained by combining a coupling for $c$
with a coupling for $c'$. We leverage this observation into our compositional
relational pre-expectation calculus, which provides upper bounds on the
Kantorovich distance.

\subsection{Compositional upper bounds by relational pre-expectation}
To facilitate compositional reasoning, we define an upper bound $\gcp{c}{\Exp}$
of the Kantorovich distance~$\Exp$ with respect to program $c$.  Technically,
$\gcp{c}{\Exp}$ is a relational pre-expectation calculus defined by induction on
the structure of $c$, similarly to the calculus by McIver and
Morgan~\cite{Morgan96}. The rules of our calculus are shown in
Figure~\ref{fig:wp2-rules}. We take the indicator expectation $\ind{\Pred}$ to
be $1$ if $\Pred$ is true, otherwise $0$, and we define addition and
multiplication on expectations pointwise.
The cases of skipping, assignments and sequential composition are
straightforward and apply the backwards semantics of commands.
The relational pre-expectation of sampling is expressed directly in terms of the
Kantorovich distance, i.e., an infimum is taken over the set of all couplings,
which is not always possible in practice. We give more details on this problem
in Section~\ref{sec:gcp-rules}.
The relational pre-expectation for conditionals assumes the two runs are
synchronized. If not, $\ind{e\sidel \neq e\sider} = 1$ and the distance is
(trivially) upper bounded by $\infty$, since the branches may not terminate with
the same probability, so the set of couplings may be empty.
Finally, in the case of while loops, we take the least fixed point of the
characteristic functional $\Phi_{\Exp, c}$ of the loop.
It is not hard to show that $\Phi_{\Exp, c}(-) \colon \EXP \to \EXP$ is monotonic
(see Lemma~\ref{lem:monotonic} in the Appendix), so by the Knaster-Tarski
theorem the least fixed point is well-defined.
As in the previous case, the relational pre-expectation returns $\infty$ when
runs are not synchronized, i.e., only one loop guard is true.
Computing the least fixed point is usually not possible. We present an
invariant-based rule in Section~\ref{sec:gcp-rules}.
\begin{figure*}[t]
  \scalebox{0.85}{\parbox{\linewidth}{%
\begin{align*}
  \gcp{\Skip}{\Exp} &\triangleq \Exp \\[.5em]
  \gcp{\Assn{x}{e}}{\Exp} &\triangleq \Exp \{ e\sidel, e\sider / x\sidel, x\sider \} \\
                          & \triangleq \lambda s_1 s_2.\Exp(s_1[x \mapsto e\langle1\rangle], s_2[x \mapsto e\langle2\rangle]) \\[.5em]
  \gcp{\Rand{x}{d}}{\Exp} &\triangleq \lambda s_1 s_2.\, \Kant{\Exp}(\denot{\Rand{x}{d}}s_1, \denot{\Rand{x}{d}}s_2)~,
                          \text{where}\ \Kant{\Exp}(\mu_1, \mu_2) \triangleq \smash{\inf_{\mu \in \Gamma(\mu_1, \mu_2)} \EE_{\mu} [ \Exp ]} \\[.5em]
  \gcp{c; c'}{\Exp} &\triangleq \gcp{c}{\gcp{c'}{\Exp}} \\[.5em]
  \gcp{\Cond{e}{c}{c'}}{\Exp} &\triangleq 
  	\ind{e\sidel \!\land \!e\sider} \cdot \gcp{c}{\Exp}
  	+ \ind{\neg e\sidel\! \land \!\neg e\sider} \cdot \gcp{c'}{\Exp}
    + \ind{e\sidel \!\neq\! e\sider} \cdot \infty \\[.5em]
    \gcp{\WWhile{e}{c}}{\Exp} & \triangleq {\sf lfp} X. \Phi_{\Exp, c}(X), \\
    \text{where}\ \Phi_{\Exp, c}(X) &\triangleq
 		 \ind{e\sidel \!\land\! e\sider} \cdot \gcp{c}{X}
     + \ind{\neg e\sidel \!\land\! \neg e\sider} \cdot \Exp
     + \ind{e\sidel \!\neq\! e\sider} \cdot \infty
\end{align*}
}}
\caption{Definition of the relational pre-expectation operator $\gcp{c}{\Exp}$.}

\label{fig:wp2-rules}
\end{figure*}%
\paragraph*{Remark (Synchronous vs.~asynchronous control flow).}
In contrast to the Kantorovich distance operator $\sgcp{c}{\Exp}$, cour
compositional relational pre-expectation operator $\gcp{c}{\Exp}$ only gives
useful (i.e., finite) bounds when the control flows in the two executions of $c$
can be \emph{synchronized}.  For deterministic guards, this means that pairs of
related executions always take the same branches; for randomized guards, this
means that we can relate the random samplings so that pairs of related
executions always take the same branches. In \cref{sec:async}, we describe
extensions of our calculus that can give more useful bounds when reasoning
asynchronously.

\paragraph*{Remark (Tightness of bounds).}
It is also complicated to estimate the exact loss between $\gcp{c}{\Exp}$ and
$\sgcp{c}{\Exp}$, since lower bounds on $\sgcp{c}{\Exp}$ are not given by a
witness coupling.  Nonetheless, in our setting this limitation is not exclusive
to our technique---in the statistical literature, lower bounds for stochastic
processes such as the ones we analyze in Section~\ref{sec:ex-shuffle} are in
general hard to compute and so the exact distance is often not known. We will
return to this topic in Section~\ref{sec:extensions}.

We now study the metatheory of our calculus.  Our first result is that our
calculus is sound: it correctly upper bounds the Kantorovich distance.%
\begin{theorem}[Soundness of $\gcpsymbol$] \label{thm:sound} 
Let $c$ be a \pwhile\ program and $\Exp \in \EXP$ be a relational
  expectation. Then $\sgcp{c}{\Exp} \leq \gcp{c}{\Exp}$, i.e., if
  $\gcp{c}{\Exp}(s_1, s_2) < \infty$ for $s_1, s_2 \in \Mem$ then
  \[
    \EE_{\mu_{s_1, s_2}} [ \Exp] \leq \gcp{c}{\Exp}(s_1, s_2) \quad 
    \mbox{ for some coupling } \quad \mu_{s_1, s_2} \in \Gamma(\denot{c}s_1, \denot{c}s_2)~.
  \]
\end{theorem}
\begin{proof}[Proof Sketch]
  By induction on $c$. The most challenging cases are for sampling and loops.
  The case for sampling requires first showing that there exists a coupling
  realizing the infimum defining the Kantorovich distance; such existence
  results belong to the theory of optimal transport~\cite{Villani08}.
  
  The case for loops is challenging for another reason: it is not clear how
  to show that the pre-expectation operator is continuous in its second
  argument (but see Thm.~\ref{thm:fin-continuity}). Instead, our
  proof relies on extracting a convergent sequence of couplings. We defer the
  details to Appendix D.
\end{proof}%
\noindent{}While it is not clear whether our relational pre-expectation operator is continuous
for \emph{all} programs, continuity does hold for programs that sample from
\emph{finite} distributions. Note that such programs can still produce
distributions with infinite support by sampling in a loop.%
\begin{theorem}[Continuity of $\gcpsymbol$] \label{thm:fin-continuity}
  Let $c$ be a \pwhile\ program where all primitive distributions have \emph{finite}
  support, and let $\Exp_n \in \EXP$ for $n \in \mathbb{N}$ be a monotonically
  increasing chain of relational expectations converging pointwise to $\Exp
  \in \EXP$. Then,
  \[
	  \gcp{c}{\Exp} = \sup_{n\in\NN} \gcp{c}{\Exp_n} .
  \]
\end{theorem}
\begin{proof}[Proof Sketch]
  By induction on the structure of $c$. The most challenging case is for
  sampling instructions, where the proof depends on a continuity property for
  the Kantorovich distance. We establish this property for distributions with
  finite support, and complete the proof of continuity for relational
  pre-expectations. We defer details to Appendix D.
\end{proof}%

\begin{figure*}
  \scalebox{0.85}{
\begin{mathpar}
  \inferrule*[right=Mono]
  { \Exp \leq \Exp' }
  { \gcp{c}{\Exp} \leq \gcp{c}{\Exp'} }
  \and
  \inferrule*[right=Const]
  { FV(\Exp') \cap MV(c) = \emptyset }
  { \gcp{c}{\Exp + \Exp'} \leq \gcp{c}{\Exp} + \Exp' }
  \\
  \inferrule*[right=SupAdd]
  { }
  { \gcp{c}{\Exp} + \gcp{c}{\Exp'} \leq \gcp{c}{\Exp + \Exp'} }\and
  \inferrule*[right=Scale]
  { f : \RRnn \to \RRnn \text{ linear, with } f(\infty) \triangleq \infty }
  { \gcp{c}{f \circ \Exp} = f \circ \gcp{c}{\Exp} }
  \\
  \and
  \inferrule*[right=Samp]
  { M\colon \Mem \times \Mem \to \Gamma(\denot{d}, \denot{d}) }
  { \gcp{\Rand{x}{d}}{\Exp} \leq \EE_{(v_1, v_2) \sim M(-, -)} [ \Exp \{ v_1, v_2 / x\sidel, x\sider \} ] }
  \and
  \inferrule*[right=Unif]
  { f\colon \Mem \times \Mem \to (D \to D) \text{ bijection} }
  { \gcp{\Rand{x}{\unif{D}}}{\Exp} \leq \frac{1}{|D|} \sum_{v \in D} \Exp \{ v, f(-,-)(v)/ x\sidel, x\sider \} }
  \and
  \inferrule*[right=Inv]
  { \ind{e\sidel \land e\sider} \cdot \gcp{c}{\Inv}
  + \ind{\neg e\sidel \land \neg e\sider} \cdot \Exp
  + \ind{e\sidel \neq e\sider} \cdot \infty \leq \Inv }
  { \gcp{\WWhile{e}{c}}{\Exp} \leq \Inv }
\end{mathpar}
}
\caption{Properties of relational pre-expectation operator $\gcp{c}{\Exp}$.}
  \label{fig:aux-rules}
\end{figure*}

\subsection{Reasoning with relational pre-expectations}
\label{sec:gcp-rules}

The definition of $\gcpsymbol$ in \cref{fig:wp2-rules} is sufficient to prove
relational properties of probabilistic programs in theory, but there are some
practical obstacles:
\begin{itemize}
  \item Comparing different relational pre-expectations for the same program is
    difficult---using the definition to compute each relational pre-expectation
    separately is tedious.
  \item Computing the relational pre-expectation for random sampling is
    difficult: it requires computing a minimum over all couplings.
  \item Computing the relational pre-expectation for loops is also difficult: in
    general, it is not possible to compute the least fixed point in closed form.
\end{itemize}
To make our operator easier to use, we introduce a collection of auxiliary
properties in \cref{fig:aux-rules}. We briefly describe the rules below.

\paragraph*{Basic properties.} 
The first four rules are basic properties of relational pre-expectations.  Rule
\textsc{Mono} states that the $\gcpsymbol$ transformer is monotone, and
\textsc{Const} intuitively states that the relational pre-expectation of $\Exp$
is $\Exp$ if $c$ doesn't modify $\Exp$; the rule is carefully stated to behave
correctly when $\gcp{c}{\Exp}$ is infinite.

The next two rules encode linearity-like properties of relational
pre-expectations. \textsc{SupAdd} states that the property is super-additive:
the relational pre-expectation of a sum can be greater than the sum of the
relational pre-expectations. Intuitively, this is because $\gcp{c}{\Exp}$
involves an infimum for random sampling, and the infimum of a sum is always less
than the sum of the infima.  \textsc{Scale} states that the relational
pre-expectation is preserved by scaling. The requirement that the scaling
function satisfies $f(\infty) = \infty$ is needed for correctly handle scaling
by $0$: $\gcp{c}{\Exp}$ may be infinite, even if $\Exp$ is identically zero.

\paragraph*{Bounding the pre-expectation for sampling.} 
Using the Kantorovich distance for defining the relational pre-expectation of a
sampling command $\Rand{x}{d}$ is theoretically clean, but inconvenient in
practice for two reasons. First, the set of couplings
$\Gamma(\denot{\Rand{x}{d}}s_1, \denot{\Rand{x}{d}}s_2)$ over which the infimum
is computed is a set of distributions over pairs of states. Given denotations of
primitive distributions $\denot{d} \in \Dist(D)$, it would be more convenient to
reason about the set $\Gamma(\denot{d},\denot{d})$---this is a set of
distributions over pairs of sampled values $D \times D$, rather than pairs of
memories. Second, computing the infimum is often difficult, and moreover
unnecessary for establishing upper bounds.

Corresponding to the \textsc{Samp} rule, the following result states that we can actually upper bound this Kantorovich distance by picking \emph{any} coupling of the
primitive distribution with itself; we call such a function $M \colon \Mem
\times \Mem \to \Gamma(\denot{d}, \denot{d})$ a \emph{coupling function} (on $d$).%
%
%
\begin{proposition} \label{prop:sampling-sound}
  Let $d$ be a primitive distribution, and let $M$ be a coupling function on $d$. For any relational
  expectation $\Exp \in \EXP$, we have:
  \[
    \gcp{\Rand{x}{d}}{\Exp} \leq \EE_{(v_1, v_2) \sim M(-, -)} [ \Exp \{ v_1, v_2 / x\sidel, x\sider \} ] ~.
  \]
\end{proposition}
%
%
\noindent{}We can reuse common couplings of primitive distributions across
different proofs. For example, let $D$ be a finite, non-empty set and let $f
\colon \Mem \times \Mem \to (D \to D)$ map pairs of program states to bijections
on~$D$. Then the \emph{bijection coupling} $M_f$, the coupling function on  $\unif{D}$ is defined by
\[
  f(s_1, s_2)(x_1, x_2) = \begin{cases}
    1/|D| &: f(s_1, s_2)(x_1) = x_2 \\
    0 &: \text{otherwise}
  \end{cases}~,
\]
where $x_1$ and $x_2$ are elements in $D$. Specialized to this case,
Proposition~\ref{prop:sampling-sound} gives \textsc{Unif}:
\begin{align*}
  \gcp{\Rand{x}{\unif{D}}}{\Exp}
  &{}\leq \gcp{\Rand{x}{d}}{\Exp} \leq \EE_{(v_1, v_2) \sim M_f(-, -)} [ \Exp \{ v_1, v_2 / x\sidel, x\sider \} ] \\
  &{}\leq \EE_{v \sim \denot{\unif{D}}} [ \Exp \{ v, f(-,-)(v)/ x\sidel, x\sider \} ] \\
                          &{} = \frac{1}{|D|} \sum_{v \in D} \Exp \{ v, f(-,-)(v)/ x\sidel, x\sider \} ~.
\end{align*}
Different coupling functions can give upper bounds of different
strengths---selecting appropriate couplings to show the target property is the
key part of reasoning by couplings. This technique is well-known to probability
theory, where it is called the \emph{coupling method}~\cite{aldous1983random}.

\paragraph*{Bounding the pre-expectation for loops.} 
As in the case of sampling, it may not always be desirable or possible to
compute the fixed point for loops. Instead, we can upper bound the relational
pre-expectation by a relational expectation $\Inv$, called an
\emph{invariant}---intuitively, if the relational pre-expectation of $\Inv$ with
respect to the loop body is at most $\Inv$, then the relational pre-expectation
of the loop is also at most $\Inv$. Formally, this reasoning is captured
by \textsc{Inv} and the following theorem:%
\begin{theorem} \label{thm:loop-sound}
Let $\Inv \in \EXP$ be a relational expectation. If
\[
  \ind{e\sidel \land e\sider} \cdot \gcp{c}{\Inv}
  + \ind{\neg e\sidel \land \neg e\sider} \cdot \Exp
  + \ind{e\sidel \neq e\sider} \cdot \infty
  \leq \Inv ,
\]
then $\gcp{\WWhile{e}{c}}{\Exp} \leq \Inv$.
\end{theorem}
\begin{proof}
  Let $\Phi$ be the characteristic functional of the loop, as defined for the
  relational pre-expectation. The hypothesis implies $\Phi(\Inv) \leq \Inv$, so
  $\Inv$ is a prefixed point of $\Phi$. By Park induction~\cite{park1969fixpoint}, the least
  fixed point $\gcp{\WWhile{e}{c}}{\Exp}$ is less than $\Inv$.
\end{proof}

\subsection{Embedding \eprhl}
\label{sec:embedding_eprhl}

Expectation Probabilistic Relational Hoare Logic (\eprhl) is a
quantitative extension of \textsc{pRHL}~\cite{BartheEGHS18}.
Judgments of \eprhl\ are of the form:
$\{ P; \Exp \}\ c_1 \sim_f c_2\ \{ Q; \Exp' \}$
where $P,Q$ are boolean-valued assertions, $\Exp,\Exp'$ are relational
expectations, $f$ is an affine function of the form $ax+b$, where
$a,b\in\mathbb{R}_{\geq 0}$, and $c_1$ and $c_2$ are \pwhile\ programs.
This judgment states that for every pair of input states $s_1, s_2$ satisfying
the pre-condition $P$, there is a coupling $\mu$ of $\denot{c_1}(s_1),
\denot{c_2}(s_2)$ whose support lies within the post-condition
$Q$, and moreover $\EE_{\mu}[\Exp']\leq f(\Exp(s_1,s_2))$. We can embed the
core inference rules of \eprhl\ in our proof system (see \cref{sec:embeddings} for details).
\begin{theorem}[Embedding \eprhl]\label{thm:eprhl}
  Let $\vdash \ehl{c}{c}{P; \Exp}{Q; \Exp'}{f}$ be a valid \textnormal{\eprhl} 
  judgment derived using the rules of \textnormal{Figure~\ref{fig:eprhl}} in
  Appendix E, with finite $\Exp$ and $\Exp'$. Then:
  \[ \gcp{c}{\Exp' + [\neg Q] \cdot \infty} \leq f(\Exp) + [\neg P] \cdot \infty . \]
  Furthermore, this inequality can be derived using just the definition of
  $\gcp{c}{\Exp}$ for skip, assignment, sequence, and conditionals in
  \textnormal{Figure~\ref{fig:wp2-rules}}, and the auxiliary proof rules in
  \textnormal{Figure~\ref{fig:aux-rules}}.
\end{theorem} 

Intuitively, the bound on the relational pre-expectation captures the validity of
the original \eprhl\ judgment. For any pair of states $(s_1, s_2)$, if $(s_1,
s_2)$ does not satisfy $P$, then the right-hand side is infinite and the bound 
trivially holds. If $(s_1, s_2)$ satisfies $P$, then the right-hand side is finite (since
$\Exp$ is finite) and the relational pre-expectation is finite. This implies
that $Q$ must be satisfied almost surely in the coupling and
$\gcp{c}{\Exp'} \leq f(\Exp)$. This last inequality recovers the \eprhl\
judgment's bound on the output distance in terms of the input distance.
Furthermore, the embedding shows that the bound is derivable in our calculus
without computing infimums over couplings for sampling, or computing least fixed
points for loops.

\section{Warmup Example: Stability of SGD}
\label{sec:ex-stability}

To demonstrate our relational pre-expectation operator, we analyze the stability
of Stochastic Gradient Descent (SGD) as our warmup example. SGD is a core tool
in modern machine learning; SGD is the most common learning algorithm used in
practice for training neural networks. Its stability was first established in
\citet{HardtRS16}, and it was later formalized in a relational program logic
\eprhl~\cite{BartheEGHS18}. The corresponding proof in \eprhl{} involves complex
proof rules---our calculus can establish the same property with significantly
cleaner reasoning.

\subsection{Background}
Let $Z$ be a space of labeled \emph{examples}, e.g., images annotated with the
main subject. A \emph{learning algorithm} $A : S \to \RR^d$ takes a set $S \in
Z^n$ of examples as input and produces (``learns'') \emph{parameters} $w \in
\RR^d$. The algorithm is tailored to a given \emph{loss function} $\ell : Z \to
\RR^d \to [0, 1]$, which describes how well an example is labeled by some
parameters. The goal is to find parameters that have low loss on examples.

In machine learning, \emph{uniform stability} is a useful property for learning
algorithms. In a nutshell, a randomized learning algorithm $A$ is
$\epsilon$-\emph{uniformly stable} if for all pairs $S, S'$ of training sets 
differing in exactly one example, and for all examples $z \in Z$, the expected
losses of $z$ are close:
\[
  | \EE_{A(S)} [ \ell(z) ] - \EE_{A(S')} [ \ell(z) ]| \leq \epsilon~.
\]
Stable learning algorithms \emph{generalize}: their performance on new,
unseen examples is similar to their performance on the training
set~\cite{BousquetE02}. In particular, stability controls how much a
learning algorithm can \emph{overfit} the training set.

\begin{figure}[t]
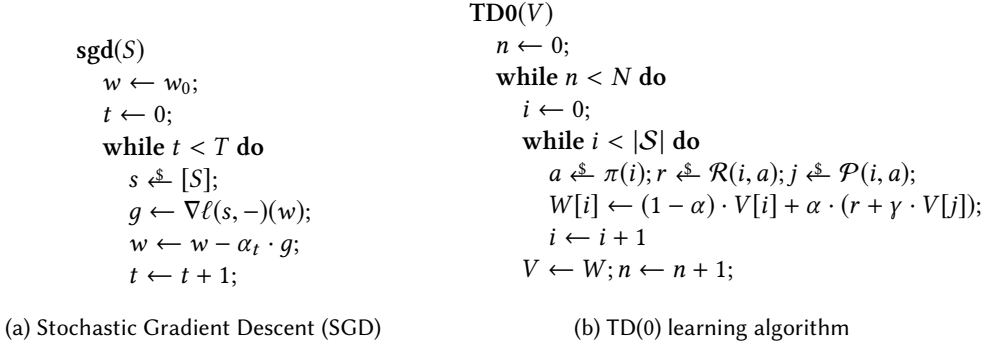

  \begin{subfigure}[b]{0.49\textwidth}
\[
  \begin{array}{l}
    \mathbf{sgd}(S) \\
    \quad \Assn{w}{w_0}; \\
    \quad \Assn{t}{0}; \\
    \quad \WWhile{t < T}{} \\
    \quad \quad \Rand{s}{[S]}; \\
    \quad \quad \Assn{g}{\nabla \ell(s,-) (w)}; \\
    \quad \quad \Assn{w}{w - \alpha_t \cdot g}; \\
    \quad \quad \Assn{t}{t + 1};
  \end{array}
\]
\caption{Stochastic Gradient Descent (SGD)}
\label{fig:sgd}
\end{subfigure}
\begin{subfigure}[b]{0.49\textwidth}
  \[
    \begin{array}{l}
      \mathbf{TD0}(V) \\
      \quad \Assn{n}{0}; \\
      \quad \WWhile{n < N}{} \\
      \quad \quad \Assn{i}{0}; \\
      \quad \quad \WWhile{i < |\cS|}{} \\
      \quad \quad \quad \Rand{a}{\pi(i)};
      \Rand{r}{\cR(i, a)};
      \Rand{j}{\cP(i, a)}; \\
      \quad \quad \quad \Assn{W[i]}{(1 - \alpha) \cdot V[i] + \alpha \cdot (r + \gamma \cdot V[j])}; \\
      \quad \quad \quad \Assn{i}{i+1} \\
      \quad \quad \Assn{V}{W};
      \Assn{n}{n + 1};
    \end{array}
  \]
  \caption{TD(0) learning algorithm}
\label{fig:RL}
\end{subfigure}
\caption{Example programs: Stability and convergence}
\label{fig:ex:ml}
\end{figure}

\subsection{Verifying stability for stochastic gradient descent}

We consider the program \textbf{sgd} in Figure~\ref{fig:sgd}.  The gradient $\nabla$ is a higher-order function\footnote{This makes our states non-discrete, but the distributions over them will still have discrete support, since they are generated by a composition of discrete samplings.} with type $\nabla : (\RR^d \to [0, 1]) \to (\RR^d \to \RR^d)$; we assume that it is well-defined and given. In SGD, the true gradient of a function is approximated by a gradient $g$ at a single sample $s$. The step sizes $\alpha_t$ (with $t \in \mathbb{N}$) are a sequence of real numbers that control (together with the local gradient $g$) how to adjust the parameters in each iteration of SGD. Following \citet{HardtRS16}, we make the following assumptions:
\begin{enumerate}
	\item The loss function $\ell$ is convex and $L$-Lipschitz in its second argument, i.e., 
	$|\ell(z, w) - \ell(z, w')| \leq L \cdot \| w - w' \|$ for all parameters 
	$w, w' \in \RR^d$.
	\item The gradient $\nabla \ell(z, -) : \RR^d \to \RR^d$ is 
	$\beta$-Lipschitz for every $z \in Z$.
	\item The step sizes satisfy $0 \leq \alpha_t \leq 2/\beta$.
\end{enumerate}
To show uniform stability, for any two training sets $S\sidel, S\sider$
differing in one element and every example $z \in Z$, our proof obligation is
\[
  | \EE_{\mathbf{sgd}(S\sidel)} [ \ell(z) ]
  - \EE_{\mathbf{sgd}(S\sider)} [ \ell(z) ] | \leq \gamma L
\qquad
  \text{where}\
  \gamma \triangleq \frac{2 L}{n} \sum_{t = 0}^{T - 1} \alpha_t\,.
\]
Rather than working with the loss function directly, we will first bound the
pre-expectation of the distance $\| w\sidel - w\sider \|$ and then use the
$L$-Lipschitz property of $\ell$ to conclude uniform stability.
As usual, the main part of the proof is bounding the
pre-expectation of the loop. We use the following loop invariant:
\begin{equation*}
  \Inv \triangleq{} \ind{t\sidel \neq t\sider} \cdot \infty
                   + \ind{t\sidel = t\sider} \cdot \left( \| w\sidel - w\sider \| + \frac{2L}{n} \sum_{j = t\sidel}^{T - 1} \alpha_j \right)~.
\end{equation*}
By the loop rule (Theorem~\ref{thm:loop-sound}), it suffices to show the
following invariant condition:
\begin{equation}
  \label{eq:sgd-inv}
  \ind{e\sidel \land e\sider} \cdot \gcp{\mathit{bd}}{\Inv}
  + \ind{\neg e\sidel \land \neg e\sider} \cdot \| w\sidel - w\sider \|
  + \ind{e\sidel \neq e\sider} \cdot \infty
  \quad \leq \quad \Inv~.
\end{equation}
The main case corresponds to the first term, where both loop guards $e\sidel$ and $e\sider$ are true.
To bound the pre-expectation $\gcp{\mathit{bd}}{\Inv}$, we consider
$\gcp{\mathit{bd}}{\Inv} = \gcp{\Rand{s}{\unif{S}}}{\Inv'}$ where
\begin{align*}
  \Inv' &\triangleq
  \ind{t\sidel{+}1 \neq t\sider{+}1} \cdot \infty 
        + \ind{t\sidel{+}1 = t\sider{+}1} \cdot P, \mbox{ with} \\
  P &\triangleq \frac{2L}{n} \sum_{j = t\sidel + 1}^{T - 1} \!\!\alpha_j 
    + \left\|\!\!\begin{array}{c} (w\sidel - \alpha_{t\sidel} \nabla \ell(s\sidel, -)(w\sidel) ) \\
        - (w\sider - \alpha_{t\sider} \nabla \ell(s\sider, -)(w\sider) )
	 \end{array}\!\!\right\|~.
\end{align*}
To handle the random sampling command, we apply the sampling rule
(Proposition~\ref{prop:sampling-sound}) with the coupling function $M$ for the two
uniform distributions $[S\sidel]$ and $[S\sider]$ induced by the bijection $f :
S\sidel \to S\sider$ mapping the differing example in $S\sidel$ to its
counterpart in $S\sider$, and fixing all other examples. We then have
$\gcp{\Rand{s}{\unif{S}}}{\Inv'} \leq \Inv''$, where
\begin{align*}
  \Inv'' &\triangleq
  \ind{t\sidel{+}1 \neq t\sider{+}1} \cdot \infty + \ind{t\sidel{+}1 = t\sider{+}1} \cdot P', \mbox{ with} \\
  P' &=  \frac{2L}{n} \sum_{j = t\sidel + 1}^{T - 1} \alpha_j
     + \frac{1}{n} \sum_{s \in S\sidel}^{n - 1}
  \left\| \begin{array}{c} (w\sidel - \alpha_{t\sidel} \nabla \ell(s, -)(w\sidel) )\\
    - (w\sider - \alpha_{t\sider} \nabla \ell(f(s), -)(w\sider) ) 
 \end{array}\right\|
\end{align*}
We focus on the terms of the last sum. Using the $L$-Lipschitz property of
$\ell$, when $s$ is the differing example, we can bound the absolute difference
by $\|w\sidel - w\sider \| + 2 \alpha_{t\sidel} L$.
When $s$ is not the differing example, we have $s\sidel = s\sider$. By the
$\beta$-Lipschitz property of $\nabla \ell$, convexity, and $0 \leq \alpha_t
\leq 2/\beta$, we can bound each of the terms by $\| w\sidel - w\sider \|$.
Combining the two cases gives
\[
  \gcp{\mathit{bd}}{\Inv} \leq \left( \| w\sidel - w\sider \| + \frac{2L}{n} \sum_{j = t\sidel}^{T - 1} \alpha_j \right)
\]
for all input states with $t\sidel = t\sider$ and $e\sidel \land e\sider$.
This establishes \eqref{eq:sgd-inv}. Theorem~\ref{thm:loop-sound} gives
\[
  \gcp{\WWhile{e}{\mathit{bd}}}{\| w\sidel - w\sider \|} \leq \Inv .
\]
Finally, taking the pre-expectations of both sides with respect to the initial
assignments yields
\[
  \gcp{\mathbf{sgd}(S)}{\| w\sidel - w\sider \|}
  \leq \frac{2L}{n} \sum_{j = 0}^{T - 1} \alpha_j = \gamma ,
\]
when $S\sidel$ and $S\sider$ differ in exactly one training example.  Since
$\ell$ is $L$-Lipschitz, we conclude
\[
  \gcp{\mathbf{sgd}(S)}{|\ell(z, w)\sidel - \ell(z, w)\sider |}
  \leq \gamma L~,
\]
for any example $z \in Z$. By Theorem~\ref{thm:abs-diff}, the expected losses
are at most $\gamma L$ apart:
\[
  | \EE_{\mathbf{sgd}(S\sidel)} [ \ell(z) ] - \EE_{\mathbf{sgd}(S\sider)} [ \ell(z) ] | \leq \gamma L~,
\]
and so SGD satisfies $\gamma L$-uniform stability.

\begin{remark*}
  This stability bound for SGD was previously verified in the program logic
  \eprhl~\citep{BartheEGHS18}, using a complex rule for sequential composition
  (\textsc{SeqCase}) that required bounding the probability of selecting two
  differing examples. Our proof using $\gcpsymbol$ is much simpler, involving
  just compositional reasoning for sequencing and a loop invariant.
\end{remark*}

\begin{remark*}
While our calculus was designed for probabilistic programs, it is also a useful
tool for proving relational properties of deterministic programs. In the
\cref{sec:pgd}, we show how to prove a sensitivity bound for \emph{projected
gradient descent}, a deterministic version of SGD.
\end{remark*}

\section{Example: Convergence of Reinforcement Learning algorithms}
\label{sec:ex-rl}

In the previous section, the stability guarantee weakens as the program
progresses: starting from two initially-equal parameter settings, the learned
parameters may drift apart as SGD runs for more iterations. In the following two
sections, we will apply our technique to prove a different style of guarantee:
probabilistic convergence of two outputs, starting from two different inputs.
Our first example shows convergence for a classical algorithm from Reinforcement
Learning (RL)~\citep{DBLP:journals/ml/Sutton88}, guided by a novel analysis
by~\citet{DBLP:conf/aistats/AmortilaPPB20}.

\subsection{Background}
In the standard reinforcement learning setting, an agent (the learning
algorithm) repeatedly interacts with the environment, a Markov Decision Process
(MDP) with \emph{state} space $\cS$. At each step, the agent chooses an
\emph{action} from a set $\cA$. The MDP draws a numeric \emph{reward} according
to a function $\cR : \cS \times \cA \to \Dist([0, R])$, and transitions to a new
random state drawn from a \emph{transition} function $\cP : \cS \times \cA \to
\Dist(\cS)$. The current state of the process is known to the learner---imagine
the current position of a chessboard---but the exact reward and transition
functions ($\cR, \cP$) are not known. Given black-box access to $\cR$ and $\cS$,
the goal of the learner is to find a policy map $\pi : \cS \to \cA$ that
maximizes the learner's expected reward when interacting with the unknown MDP
over an infinite time horizon; estimated rewards in the future are reduced by a
discount factor $\gamma \in [0, 1)$ for each step into the future.

For many approaches to learning the optimal policy, an important requirement is
estimating the \emph{value function} $V : \cS \to [0, R]$ of the MDP, i.e., the
expected reward at each state if the agent were to repeatedly act according to
some given policy $\pi$. \emph{Temporal difference (TD)} learning is one
approach to estimating the value function~\citep{DBLP:journals/ml/Sutton88}. In
brief, a TD learner maintains an estimate of $V$ and loops through states in
$\cS$. At each state $s$, the learner selects an action $a \sim \pi(s)$, draws a
reward $r \sim \cR(s, a)$, and draws a transition $s' \sim \cR(s, a)$. Then, the
estimate $V(s)$ is updated by incorporating the observed reward $r$ and the
estimated value $V(s')$ of the new state.

\Cref{fig:RL} shows one simple approach, known as TD(0). We assume that the
program takes only one argument $V$, the initial estimate of the value
function. All other parameters are assumed to be fixed: the current policy
$\pi$, the reward and transition functions $\cR$ and $\cP$, the discount factor
$\gamma$, the step size $\alpha \in (0, 1)$---higher $\alpha$ allows $V$ to
evolve faster--and the number of iterations $N$.

\subsection{Verifying convergence for $\mathbf{TD0}$}

Since the true value function is not known, the initial estimate $V$ chosen with
little information. A natural question is: does the algorithm converge to the
same distribution no matter how $V$ is initialized? If so, how fast does
convergence happen, as a function of the number of iterations $N$? To answer
these questions, we will verify that $\mathbf{TD0}$ is contractive on $V$. More
specifically, we will show the bound
\begin{equation} \label{eq:td0:guarantee}
  \gcp{\mathbf{TD0}(V)}{\| V\sidel - V\sider \|_\infty}
  \leq k^N \cdot \| V\sidel - V\sider \|_\infty ,
\end{equation}
where $k \triangleq (1 - \alpha + \alpha \gamma) < 1$. Before we describe the
verification, we unpack the guarantee. First, the $\infty$-norms are defined by
$\| V\sidel - V\sider \|_\infty \triangleq \max_{i < |\cS|} | V\sidel[i] -
V\sider[i] |$. By \cref{thm:sound}, \cref{eq:td0:guarantee} implies that for any
inputs $V_1$ and $V_2$, there exists a coupling $\mu$ of the output
distributions $\mu_1$ and $\mu_2$ from $\mathbf{TD0}(V\sidel)$ and
$\mathbf{TD0}(V\sider)$, such that:
\begin{align}
  k^N \cdot \| V_1 - V_2 \|_\infty
  &\geq \EE_{(s_1, s_2) \sim \mu} [ \| s_1(V) - s_2(V) \|_\infty ]
  \notag \\
  &\geq \max_{i < |\cS|} \EE_{(s_1, s_2) \sim \mu} [ | \, s_1(V[i]) - s_2(V[i]) \, | ]
  \notag \\
  &\geq \max_{i < |\cS|}~\bigl| \, \EE_{(s_1, s_2) \sim \mu} [ s_1(V[i]) - s_2(V[i])  ] \, \bigr|
  \notag \\
  &= \max_{i < |\cS|}~\bigl| \, \EE_{s_1 \sim \mu_1} [ s_1(V[i]) ] - \EE_{s_2 \sim \mu_2} [ s_2(V[i]) ] \, \bigr|
  \tag{by \cref{thm:abs-diff}}
\end{align}
In words, the right-hand side of the final line is the maximum difference
between the average estimates of $V[i]$ in the two outputs, taking the maximum
over all indices $i$. Since $k < 1$, both sides tend to zero exponentially
quickly from any pair of starting states $V_1$ and $V_2$.

\paragraph*{Inner loop.}
We start by analyzing the inner loop~$w_{in}$. We first show that
\[
    \gcp{w_{in}}{\| W\sidel - W\sider \|_\infty} \leq \cI_{in}
\]
for the invariant $\cI_{in}$:
\begin{multline*}
	\cI_{in} \triangleq \ind{i\sidel \neq i\sider} \cdot \infty \\
                      + \ind{i\sidel = i\sider} \cdot \max_{l < |S|} (\ind{l < i\sidel} \cdot | W\sidel[l] - W\sider[l] |
                      + \ind{i\sidel\leq l}\cdot k\cdot \| V\sidel - V\sider \|_\infty ).
\end{multline*}
Let $c_{in}$ be the body, and $c_{samp}$ be the three sampling statements.
Applying \textsc{Inv}, it suffices to show:
\[
  \ind{i\sidel < |\cS| \land i\sider < |\cS|} \cdot \gcp{c_{in}}{\Inv_{in}}
  + \ind{i\sidel \geq |\cS| \land i\sider \geq |\cS|} \cdot \| W\sidel - W\sider \|_\infty
  + \ind{i\sidel \neq i\sider} \cdot \infty \leq \Inv_{in}
\]
The main case is bounding $\gcp{c_{in}}{\Inv_{in}}$; the other cases are
simpler. We describe the overall idea here, deferring details to \cref{sec:td0}.
To bound the relational pre-expectation for the three sampling instructions, we
apply the sampling rule \textsc{Samp}. Since the relational pre-expectation is
computed in reverse order, we must choose a coupling for sampling $j$ first, then choose
a coupling for sampling $r$, and then finally choose a coupling for sampling
$a$. We aim to take the identity coupling in each case, ensuring $j\sidel =
j\sider$, $r\sidel = r\sider$, and $a\sidel = a\sider$, but there is a small
problem: we can only take the identity coupling when samples are taken from
the same distributions, e.g., $\cR(i\sidel, a\sidel) = \cR(i\sider, a\sider)$.
The invariant assumes $i\sidel = i\sider$, but we can only ensure $a\sidel =
a\sider$ after we have specified the couplings for $j$ and $r$. Accordingly, our
coupling functions for \textsc{Samp} will be of the following form: if $a\sidel
= a\sider$ then we take the identity coupling, otherwise we take the trivial
(independent) coupling.

\paragraph*{Outer loop.}
We now turn to the analysis of the outer loop. Consider the invariant:
\[
  \Inv_{out} \triangleq [n\sidel \neq n\sider] \cdot \infty
  + [n\sidel = n\sider] \cdot k^{(N\ominus n\sidel)} \| V\sidel - V\sider \|_\infty~,
\]
where $N \ominus n$ denotes $\max(N{-}n,0)$.  We compute:
\begin{align*}
  &\gcp{\Assn{i}{0}; w_{in}; \Assn{V}{W}; \Assn{n}{n+1}}{\Inv_{out}} \\
  &= \gcp{\Assn{i}{0}; w_{in}}{[n\sidel \neq n\sider] \cdot \infty + [n\sidel = n\sider]
  \cdot k^{(N\ominus (n\sidel{+}1))} \| W\sidel - W\sider \|_\infty} \\
  &\leq \gcp{\Assn{i}{0}}{[n\sidel \neq n\sider] \cdot \infty
  + [n\sidel = n\sider] \cdot k^{(N \ominus (n\sidel+1))}\cdot \Inv_{in}}\\
  &\leq [n\sidel \neq n\sider] \cdot \infty
  + [n\sidel = n\sider] \cdot k \cdot k^{(N \ominus (n\sidel+1))} \| V\sidel -
  V\sider \|_\infty
  = \Inv_{out}
\end{align*}
where the last step holds because $\Inv_{in} = k \cdot \| V\sidel - V\sider
\|_\infty$ when $i = 0$. This establishes the outer invariant. Computing the
pre-expectation of the first initialization, we conclude:
\[
  \gcp{\mathbf{TD0}(V)}{\| V\sidel - V\sider \|_\infty}
  \leq k^N \cdot \| V\sidel - V\sider \|_\infty~.
\]

\section{Example: Random Walks and Card Shuffles}
\label{sec:ex-shuffle}

In this section, we verify more challenging examples of
probabilistic convergence from the theory of Markov chains, formalizing
arguments by \citet{aldous1983random} in his seminal work introducing the
coupling method. Our use of relational pre-expectations is similar in spirit to
the previous section, but there are two key differences: (1) we aim to prove
convergence under Total Variation (TV) distance, which is the standard notion of
distance in this field, and (2) our arguments will require selecting more
complex couplings, instead of just the identity coupling.

\subsection{Preliminaries: Card shuffling and Markov chain mixing}
Distributions that are easy to describe can be surprisingly difficult to sample
from. For instance, consider the uniform distribution over all permutations of a
deck of playing cards. It is not clear how to sample from this
distribution---i.e., perform a \emph{perfect shuffle}---but we can implement a
card shuffle algorithm that executes a sequence of simple randomized steps (e.g.
swapping pairs of cards) and hope that after a small number steps, we will
produce a shuffle that is close to uniform.

Abstracting a bit, card shuffling algorithms are a representative example of
random walks for approximating complex distributions. This is a technique with a
long history, combining elements of probability theory with statistical physics;
and it is the basis of many heuristic algorithms used today, e.g., Markov Chain
Monte Carlo (MCMC). From a theoretical perspective, the central question is:
\emph{how fast do these processes converge to their target distribution}?  How
many steps do we need to get within $\epsilon$ distance of the uniform
distribution on shuffles?

Random walks and card shuffling algorithms are classical examples of
\emph{Markov chains}. A finite, discrete-time Markov chain is defined by a
finite state space $\Sigma$ and a transition function $P \colon \Sigma \to \Dist(\Sigma)$. Given an initial state $\sigma$, the associated Markov
process $\{ X^\sigma_k \}_{k \in \NN}$ is a sequence of distributions such that
$X^\sigma_0 = \delta(\sigma)$ and $X^\sigma_{k+1}(\tau') = \sum_\tau
X^\sigma_k(\tau) \cdot P(\tau,\tau')$. For example, the state space $\Sigma$
could be the set of all permutations of a deck of cards, and the transition
function $\tau$ could describe randomly splitting the deck and interleaving the
halves.

Consider the TV distance $v(k)$ between two state distributions
after running $k$ steps from two states $\sigma, \tau$, i.e.,
$
  v(k) \triangleq \max_{\sigma, \tau} TV(X^\sigma_k, X^\tau_k)~.
$
If $v(k)$ tends to 0, then there exists a unique \emph{stationary} distribution
$\eta$ such that $\eta(\sigma)\cdot P(\sigma,\sigma') = \eta(\sigma')$;
typically, $\eta$ will be the target distribution we are trying to sample from.
Furthermore, $v(k)$ provides an upper bound on the distance between the state
distribution after $k$ steps to the stationary distribution $\eta$:
\[
  \max_\sigma TV(X^\sigma_t, \eta) \leq v(k)~.
\]
While it is usually not possible to derive $v(k)$ exactly, we can upper-bound
$v(k)$ by constructing couplings of $(X^\sigma_t, X^\tau_t)$ and applying
Theorems~\ref{thm:tv} and~\ref{thm:scaled-tv}. In this way, we can prove bounds
on the number of steps needed to get within some distance of the target
distribution.

\subsection{Warmup: Hypercube walk}
\label{ex:hypercube}

We start off with a (rather naive) random walk for sampling $N$ uniformly random
bits, which serves as a toy version of the more complex random walks we will see
later.  Our \emph{position} is a string of $N$ bits (which can be regarded as a
vertex of an $N$-dimensional \emph{hypercube}). On every iteration of the walk
we uniformly sample from $\{ 0, \dots, N \}$. Note that there are $N + 1$
possible draws, but only $N$ coordinates: if we sample $0$, then we do not move,
otherwise we reverse the sampled coordinate $i$ in the current position. We will
show that starting from any two positions, the process \emph{mixes rapidly},
i.e.\ starting from any position we will quickly reach the uniform distribution
over positions.

Let $e(i) = (0, \dots, 1, \dots, 0) \in \{0,1\}^N$ be the position where all
coordinates are set to zero except for coordinate $i$, which is set to one. We
also write $\oplus$ for xor applied coordinate-wise. We can model $K$ steps of
the random walk with the following simple \pwhile\ program:
\[
  \begin{array}{l}
    \mathbf{hWalk}(\mathit{pos}, N, K) \\
    \quad \Assn{k}{0}; \\
    \quad\WWhile{k < K}{} \\
    \quad\quad \Rand{i}{\unif{[N {+} 1]}}; \\
    \quad\quad \Condt{i \neq 0}{}
    \Assn{\mathit{pos}}{\mathit{pos} \oplus e(i)}; \\
    \quad\quad \Assn{k}{k + 1} \\
  \end{array}
\]
%
Consider two program runs, started at $\mathit{pos}\sidel$ and
$\mathit{pos}\sider$ respectively. Let $d_H$ be normalized Hamming distance
between the two positions:
\[
	d_H \triangleq \frac{1}{N} \sum_{i = 1}^N [\mathit{pos}\sidel[i] \neq \mathit{pos}\sider[i]]~.
\]
That is, $d_H$ equals the fraction of coordinates where $\mathit{pos}\sidel$ and
$\mathit{pos}\sider$ differ. Let $C(\mathit{pos}\sidel, \mathit{pos}\sider)
\subseteq [N]$ be the set of differing coordinates. We specify a coupling on
$\unif{[N {+} 1]}$ by giving a bijection on $[N {+} 1]$. There are three cases:
\begin{enumerate}
  \item $d_H \geq 2/N$: Let $C(\mathit{pos}\sidel, \mathit{pos}\sider) = \{ i_0,
    \dots, i_{m - 1} \}$.  Take the bijection that behaves like the identity on $[N{+}1] \setminus
    C(\mathit{pos}\sidel, \mathit{pos}\sider)$ and that, for all $0 \leq n \leq m$, maps $i_n$ to $i_{n + 1}$,
    where we set $i_m = i_0$.
  \item $d_H = 1/N$: Take the bijection exchanging the differing coordinate
    and $0$.
  \item $d_H = 0$: Take the identity bijection.
\end{enumerate}
The coupling captures the following intuition.  When $d_H \geq 2/N$, the distance
decreases by $2/N$ if we select a differing coordinate; otherwise, it
remains unchanged. Likewise when $d_H = 1/N$, if we select the differing
coordinate or $0$, then the distance decreases by $1/N$ (to $0$); otherwise, the
distance remains unchanged.

We can analyze the program $\mathbf{hWalk}$ using our relational pre-expectation
calculus. Let the target relational expectation be $d_H$.  The main step in the
reasoning is to select a relational invariant for the loop. We define:
\[
  \Inv \triangleq \ind{k\sidel \neq k\sider} \cdot \infty
      + \ind{k\sidel = k\sider} \cdot d_H \cdot \left(\frac{N-1}{N + 1}\right)^{K\ominus k\sidel} .
\]
Then, we can verify for the loop $\WWhile{k < K}{bd}$ of program $\mathbf{hWalk}$ that
\[
\begin{array}{rl}
&\ind{(k\sidel < K\sidel) \land (k\sider < K\sider)} \cdot \gcp{\mathit{bd}}{\Inv} \\
  + &\ind{(k\sidel \geq K\sidel) \land (k\sider \geq K\sider)} \cdot d_H \\
  + &\ind{(k\sidel < K\sidel) \neq (k\sider < K\sider)} \cdot \infty \quad\qquad \leq \quad \Inv,
\end{array}
\]
and conclude by the loop rule (Theorem~\ref{thm:loop-sound}):
\[
  \gcp{\WWhile{k < K}{\mathit{bd}}}{d_H} \leq \Inv .
\]
The main step here is showing that 
\[\ind{(k\sidel < K\sidel) \land (k\sider < K\sider)} \cdot \gcp{\mathit{bd}}{\Inv}
\ \leq \ 
\ind{(k\sidel < K\sidel) \land (k\sider < K\sider)} \cdot \Inv~, \]
where we use the fact that the coupling described above makes $d_H$ decrease.

Pushing the invariant past the initialization instruction $\Assn{k}{0}$ yields:
\[
  \gcp{\mathbf{hWalk}(\mathit{pos}, N, K)}{d_H}
  \leq \gcp{\Assn{k}{0}}{\cI}
  = \left(\frac{N-1}{N + 1}\right)^K~.
\]
Since the distance $d_H$ takes distance at least $1/N$ on pairs of distinct
positions, by Theorem~\ref{thm:scaled-tv} the TV distance between the
distributions over positions satisfies
\begin{align*}
v(K,N) = & \max_{p_1,p_2 \in \{0,1\}^N } TV(\denot{\mathbf{hWalk}}(p_1,N,K), \denot{\mathbf{hWalk}}(p_2,N,K)) \\
\leq & \ N \left( 1 - \frac{2}{N{+}1} \right)^K .
\end{align*}
Plugging in specific values gives concrete bounds between the two output
distributions. Let $\rho > 1$. To achieve a bound of $O(1/\rho)$ on the
right hand side, we need to take $K \geq (1/2) N \log (N\rho)$. 
The inequality above also gives useful asymptotic information;
if we set $\rho = N$, and take $K \geq N \log N$, the right-hand side is
asymptotically bounded by $O(1/N)$ for large $N$.
We can show that this converges to the uniform distribution over vectors.
We provide more details in Section~\ref{sec:extensions}. In summary, we 
have shown the following:

\begin{theorem}
	Let $K = N \log N$.
	For any initial position $\mathit{pos}$,
	\[TV \left( \mathbf{hWalk}(\mathit{pos},N,K), U(\{0,1\})^N)\right) \in \mathcal{O}(1/N)~. \] 
\end{theorem}

\begin{figure}[t]
  \begin{minipage}[t]{.25\textwidth}
\[
  \begin{array}{l}
    \mathbf{rTop}(\mathit{deck}, N, K) \\
    \quad \Assn{k}{0}; \\
    \quad\WWhile{k < K}{} \\
    \quad\quad \Rand{p}{\unif{[N]}}; \\
    \quad\quad \Assn{\mathit{deck}}{\mathsf{shiftR}(\mathit{deck},p)}; \\
    \quad\quad \Assn{k}{k + 1}; \\
  \end{array}
\]
	\end{minipage}
  \begin{minipage}[t]{.4\textwidth}
\[
  \begin{array}{l}
    \quad\mathbf{rTrans}(\mathit{deck}, N, K) \\
    \quad\quad \Assn{k}{0}; \\
    \quad\quad\WWhile{k < K}{} \\
    \quad\quad\quad \Rand{p}{\unif{[N]}}; \Rand{p'}{\unif{[N]}}; \\
    \quad\quad\quad \Assn{c}{\mathit{deck}[p]}; \Assn{c'}{\mathit{deck}[p']}; \\
    \quad\quad\quad \Assn{\mathit{deck}[p]}{c'}; \Assn{\mathit{deck}[p']}{c}; \\
    \quad\quad\quad \Assn{k}{k + 1}; \\
  \end{array}
\]
	\end{minipage}
  \begin{minipage}[t]{.3\textwidth}
\[
  \begin{array}{l}
    \mathbf{riffle}(deck, N, K) \\
    \quad \Assn{k}{0}; \\
    \quad\WWhile{k < K}{} \\
    \quad\quad \Rand{b}{\unif{\{ 0,1 \}^N}}; \\
    \quad\quad \Assn{top}{deck(\bar{b})}; \\
    \quad\quad \Assn{bot}{deck(b)}; \\
    \quad\quad \Assn{deck}{\mathsf{cat}(top, bot)}; \\
    \quad\quad \Assn{k}{k + 1};
  \end{array}
\]
	\end{minipage}
\caption{Shuffling algorithms}
\label{fig:shuffling}
\end{figure}
\subsection{Random-to-top shuffle}
\label{ex:rand-to-top}

For our shuffling examples, we will need some notation. We view a permutation
$deck$ as a map from positions in $p \in [N]$ to names of cards in $c \in C$;
$deck[p]$ denotes the card at position $p$, while $deck^{-1}(c)$ denotes the
position corresponding to card $c$. Summation over an empty set of indices is
treated as zero, while the product over an empty set of indices is treated as one.
We outline the arguments here; further details are provided in Appendix F.

For our first card shuffling algorithm we consider the \emph{random-to-top}
shuffle. In each iteration, it selects a random position in the deck and moves
the card at that position to the top.\footnote{%
  This algorithm is the time-reversed version of the \emph{top-to-random}
  shuffle, where the top card is moved to a random position. It is known
  that a Markov chain's convergence behavior is equivalent to that of its
reversed process~\cite{aldous1983random}.}
We model this shuffle with program \textbf{rTop} in Figure~\ref{fig:shuffling}.
For a given input deck of size $N$, the program repeats $K$ times the process
of selecting a random card and moving it to the top. The operation
$\mathsf{shiftR}(\mathit{deck},j)$ takes the block
$\mathit{deck}[0],\ldots,\mathit{deck}[j]$ and cycles it one
position to the right (thus moving $\mathit{deck}[j]$ to the top), leaving the rest of the deck intact.

We are interested in bounding the distance between the stationary distribution---which in this case is the uniform distribution---and the output distribution
after $K$ iterations. We will start with two decks of
size $N$ that are both permutations of $[N]$. 
As in the hypercube example, we bound the
pre-expectation of the normalized Hamming distance:
\[
  d_H \triangleq \frac{1}{N} \sum_{i=0}^{N-1} \ind{\mathit{deck}\sidel[i] \neq \mathit{deck}\sider[i]}~.
\]
Note that $d_H$ takes distance at least $1/N$ on pairs of distinct permutations.
If we can show that the pre-expectation of $d_H$ is not too big, then we can apply
Theorem~\ref{thm:scaled-tv} to conclude that the final distributions over
permutations have a close TV distance. It will be convenient to work with
an auxiliary distance:
\[
        d_M \triangleq (1/N) \cdot \left(N - \max_{i} \left(\forall j<i. \mathit{deck}\sidel[j] = \mathit{deck}\sider[j] \right) \right)~.
\]
The idea is that the coupling chooses identical cards on both decks and moves
them to the top. This will form a block of matched cards on the top of both
decks. Intuitively, $d_M$ measures the fraction of the deck that is not part of
this top block. The target distance $d_H$ is upper-bounded by $d_M$, since $d_M$
counts all cards outside the first block as different. Bounds on $d_H$ follow
from bounds on $d_M$. 
To bound the pre-expectation of $d_M$, we take the invariant:
\[
  \Inv \triangleq \ind{k\sidel \neq k\sider} \cdot \infty
  + \ind{k\sidel = k\sider} \cdot d_M \cdot \left(\dfrac{N-1}{N}\right)^{K \ominus k\sidel}~.
\]
We can check that it satisfies the inequality
\[
  \ind{k\sidel < K \land k\sider < K} \cdot \gcp{\mathit{bd}}{\Inv}
  + \ind{k\sidel \geq K \land k\sider \geq K} \cdot d_H
  + \ind{(k\sidel < K) \neq (k\sider < K)} \cdot \infty \leq \Inv,
\]
where $\mathit{bd}$ is the loop body. The main case is to show the inequality
for the first term when both loop guards are true: we need to bound the
pre-expectation of $\Inv$ with respect to $\mathit{bd}$. We can bound
\[
  \gcp{\mathit{bd}}{\Inv} \leq d_M \cdot \left(\dfrac{N-1}{N}\right)^{K \ominus k\sidel}~,
\]
by applying the sampling rule (Proposition~\ref{prop:sampling-sound}) with the
coupling function $M$ that selects the same card in both decks:
\[
  M(s_1, s_2)(p_1, p_2) \triangleq \begin{cases}
    1/N &: \denot{\mathit{deck}}s_1 [p_1] = \denot{\mathit{deck}}s_2 [p_2] \\
    0 &: \text{otherwise} .
  \end{cases}
\]
The idea is that if we pick two cards in the first matched block, which happens
with probability $(1 - d_M)$, then the distance will remain the same. Otherwise,
we will create at least one new matched pair in the first block and the distance
will decrease by $1/N$.
Hence, we can apply the loop rule (Theorem~\ref{thm:loop-sound}) to conclude:
\[
  \gcp{\WWhile{k < K}{\mathit{bd}}}{d_H} \leq \Inv .
\]
Computing the pre-expectation of $\Inv$ with respect to the first
instruction, we have
\[
  \gcp{\mathbf{rTop}(\mathit{deck},N,K)}{d_H}
  \leq \left(\frac{N - 1}{N}\right)^K~,
\]
noting that the distance $d_M$ between the initial decks is at most $1$. Since
$d_H$ assigns pairs of distinct decks a distance at least $1/N$,
Theorem~\ref{thm:scaled-tv} implies that the TV distance between the
distributions over decks satisfies:
\[ v(K,N) = \max_{d_1,d_2 \in [N]} TV(\denot{\mathbf{rTop}}(d_1,N,K), \denot{\mathbf{rTop}}(d_2,N,K))
\ \leq N \ \left( \frac{N - 1}{N} \right)^K~. \]
For example, if we choose $K$ to be $N \log (N\rho)$, then the distance between
permutation distributions is bounded by $O(1/\rho)$ for large $N$ and $\rho > 1$.
By setting $\rho=N$, we have shown the following:

\begin{theorem}
	Let $K = 2 N \log N$, and $\mathit{Perm}([N])$ be the set of permutations over $N$. 
	For any initial permutation of $\mathit{deck}$,
	\[TV(\mathbf{rTop}\left(\mathit{deck},N,K), U(\mathit{Perm}([N])) \right) \in \mathcal{O}(1/N). \] 
\end{theorem}
%
%
%
\subsection{Random transpositions shuffle}
\label{ex:rand-trans}

Our next shuffle (\textbf{rTrans} in Figure~\ref{fig:shuffling}) repeatedly
selects two positions uniformly at random and swaps the cards, allowing for the
possibility of swapping a card with itself. 
As before, let $d_H$ be the normalized Hamming distance between the two decks.
We aim to bound $\gcp{\mathbf{rTrans}}{d_H}$. As before, the key of the proof is
finding an invariant for the loop. We take:
\[
  \Inv \triangleq \ind{k\sidel \neq k\sider} \cdot \infty
  + \ind{k\sidel = k\sider} \cdot d_H \cdot \left(1 - \frac{1}{N^2}\right)^{K \ominus k\sidel}
\]
%
%
There are two samplings in the loop body,
so we need to provide two couplings.  
For the first sampling $p$, we use the identity coupling. For
the second sampling $p'$, we couple using the bijection induced by the two decks
$\mathit{deck}\sidel$ and $\mathit{deck}\sider$, i.e., the coupling matches
every position $p'\sidel$ with the unique position $p'\sider$ such that
$\mathit{deck}[p']\sidel = \mathit{deck}[p']\sider$. There are three cases:
(1) if cards at $p\sidel, p\sider$ are already matched, $d_H$ remains unchanged;
(2) if positions $p'\sidel, p'\sider$ are equal, $d_H$ remains unchanged; otherwise
(3) $d_H$ decreases by $1$.
%
%
This is enough to show that the invariant decreases. We can conclude:
\[
\begin{array}{l}
  \gcp{\mathbf{rTrans}(\mathit{deck}, N, K)}{d_H}
  \leq \left(1 - \frac{1}{N^2}\right)^K
\end{array}
\]
using the fact that $d_H$ between the inputs is at most $1$. Since $d_H$ takes value of at least $1/N$ for pairs of
distinct decks, by Theorem~\ref{thm:scaled-tv}
\[
v(K,N) = \max_{d_1,d_2 \in [N]}TV(\denot{\mathbf{rTrans}}(d_1,\! N,\! K), \denot{\mathbf{rTrans}}(d_2,\! N,\! K))\! \ \leq \
N \! \left(1\! -\! \frac{1}{N^2}\right)^K~,
\]
so the distance between the deck distribution and the uniform distribution
decreases as $K$ increases. If we take $K \geq N^2 \log(N\rho)$, then the
right-hand side is bounded asymptotically by $O(1/\rho)$ for large $N$.
By setting $\rho = N$, we conclude: 


\begin{theorem}
	Let $K = 2 N^2\log N$, and $\mathit{Perm}([N])$ be the set of permutations over $N$. 
	For any initial permutation of $\mathit{deck}$,
	\[TV(\mathbf{rTrans}(\mathit{deck},N,K), U(\mathit{Perm}([N]))) \in \mathcal{O}(1/N). \] 
\end{theorem}

\begin{remark*}
  Aldous'~\cite{aldous1983random} bound is slightly sharper: the TV distance
  between output distributions is bounded by $O(1/N)$ asymptotically already for
  $K \geq CN^2$ for some constant $C$. This discrepancy appears because
  our proofs are carried out compositionally, while Aldous uses a global
  analysis. However, it is possible that a clever choice of coupling or
  loop invariant could let us match Aldous' bound. 
\end{remark*}

\subsection{Uniform riffle shuffle}
\label{ex:riffle}

In this example we will analyze the uniform riffle shuffle, which is a more
realistic model of how cards are shuffled by humans. The shuffle begins by 
dividing the deck in approximately two halves, and then
merges the two halves in an approximately alternating manner. 
The reversed process, program \textbf{riffle} on Figure~\ref{fig:shuffling}
which we analyze, takes a deck, samples a uniform random bit for
each card, and then places all cards labeled with 0 on top of the deck without
altering their relative order.  After repeating this process $k$ times, for
every card $i$ we have sampled a string of bits 
$(b_{i,0}, \dots, b_{i,k-1})$, and card $i$ is on top of card $j$ if, for some $m$,
$b_{i,k} = b_{j,k}, b_{i,k-1} = b_{j,k-1}, \dots, b_{i,m} = b_{j,m}$ and $b_{i,m-1} < b_{j,m-1}$.

The vector $b$ holds $N$ bits, indexed by position; $\bar{b}$ negates each
entry. We use shorthands for partitioning: $deck(b)$ and $deck(\bar{b})$
represent the sub-permutations from taking all positions where $b$ is $0$ and
$1$, respectively. Finally, $\mathsf{cat}$ concatenates two permutations.

We will take the coupling that always samples the same bit for the same card on
both sides: $b(deck^{-1}(c))\sidel = b(deck^{-1}(c))\sider$ for every $c \in C$.
It is not hard to see that this coupling will eventually make the decks match.
However, choosing an appropriate distance takes more care, since the Hamming
distance may not always decrease under this coupling.
For reasons of space, we leave the details of verification to Appendix G. We can show
the following:
\begin{theorem}
	Let $K = 3 \log N$, and $\mathit{Perm}([N])$ be the set of permutations over $N$. 
	For any initial permutation of $\mathit{deck}$,
	\[TV(\mathbf{riffle}(\mathit{deck},N,K), \mathbf{Unif}\{\mathit{Perm}([N])\}) \in \mathcal{O}(1/N). \] 
\end{theorem}

\section{Extensions: Proving Lower Bounds and Uniformity}
\label{sec:extensions}

In this section, we describe two extensions to our random walk examples from
\cref{sec:ex-shuffle}: proving that the limit distribution is uniform, and
proving lower bounds on the TV distance.

\subsection{Convergence to uniform distribution}

In \cref{sec:ex-shuffle}, we showed that the Markov chains correspond
to each example converge to a stationary distribution, but we did not shown that
this distribution is the uniform distribution over states---if we had made an
error in the implementation, the probabilistic program may converge to the wrong
distribution. We can use our relational pre-expectation calculus along with
Theorem~\ref{thm:abs-diff} to show that the limit distribution is indeed
uniform.

We illustrate the technique for the random-to-top shuffle, but the idea is
applicable to all our examples. Consider any two permutations of the deck
$R_1,R_2$, and the unary expectations
\[
  S_1 (\mathit{deck}) \triangleq [\mathit{deck} = R_1]
  \qquad\text{and}\quad
  S_2(\mathit{deck}) \triangleq [\mathit{deck} = R_2] .
\]
To show that the shuffle converges to uniform, we need to show that the expected
values of $S_1$ and $S_2$ converge to the same value. Recall that
Theorem~\ref{thm:abs-diff} states that for any initial states $s_1,s_2$,
\[\bigl| \EE_{\denot{\mathbf{rTop}} s_1}[ S_1 ] - \EE_{\denot{\mathbf{rTop}}
s_2}[ S_2 ] \bigr| \leq \Kant{|S_1 - S_2|}\bigl(\denot{\mathbf{rTop}} s_1,
\denot{\mathbf{rTop}} s_2\bigr)\] so it suffices to show that the right hand
side converges to zero.

Computing the weakest pre-expectation of $|S_1 - S_2|$ directly is difficult, so
we define an alternative distance.  We can see $R_1$ and $R_2$ as defining a
relation (actually, a permutation $\pi$ over $[N]$) of pairs $(R_1[i], R_2[i])$
of cards that are at the same positions. We let $d$ be the distance defined by:
\[
  d(\mathit{deck}\sidel,\mathit{deck}\sider)
  \triangleq \sum_{i=0}^{N-1} [(\mathit{deck}\sidel[i], \mathit{deck}\sider[i]) \not\in \pi]~.
\]
We can show that $|S_1(\mathit{deck}\sidel) - S_2(\mathit{deck}\sider)| \leq
d(\mathit{deck}\sidel,\mathit{deck}\sider)$, since $d$ takes non-negative
integer values, and whenever $d = 0$, then $S_1$ and $S_2$ can only be true
simultaneously. So it suffices to show that the right-hand side converges to
zero. This bound can also be established by our pre-expectation calculus in much
the same way as in our proof for the random-to-top shuffle, but we use a
different coupling. After sampling $p\sidel$ on the first execution we just need
to pick the $p\sider$ on the second such that $(\mathit{deck}\sidel[p\sidel],
\mathit{deck}\sider[p\sider]) \in \pi$. This makes $d$ decrease any time a new
match is formed, and once a match is formed and moved to the top, it is never
undone. By starting from the same permutation $\mathit{deck}\sidel =
\mathit{deck}\sider$, this analysis shows that the rate of convergence---this
time to the \emph{uniform} distribution---is the same as in our previous
analysis of random-to-top: $d$ converges to 0 at rate $(1 - 1/N)^K$.

\subsection{Proving lower bounds}

\gb{we should mention the kantorovich duality theorem}
Previously, we verified upper bounds of the Total Variation distance by using
the Kantorovich distance. It is also interesting to compute \emph{lower bounds}
on the TV distance, describing how far apart the distributions must be.  We
consider how to verify these bounds using the $\wpc$ calculus of McIver and
Morgan~\citep{Morgan96}, summarized in Figure~\ref{fig:wpe-rules}.  We will need
an alternative definition of the TV distance expressed in terms of expected
values rather than sets:%
\begin{proposition}
\label{prop:TV-expected-value-def}
	Let $\mu_1, \mu_2 \in \Dist(X)$. Then,
	\[ \sup_{f \colon X \to [0,1]} |\EE_{\mu_1}[f] - \EE_{\mu_2}[f]| = \sup_{S\subseteq X} |\mu_1(S) - \mu_2(S)| = TV(\mu_1, \mu_2)~.\]
\end{proposition}%
%
\noindent{}%
%
%
Thus, it suffices to pick \emph{any} $f \colon X \to [0,1]$ and compute its
expected values wr.t.\ $\mu_1$ and $\mu_2$ 
to get \emph{a} lower bound on $TV(\mu_1, \mu_2)$.  This
is performed with the $\wpc$ operator, which takes a program $c$
and an \emph{expectation} $\mathcal{F} \colon \Mem \to [0,\infty]$ and
computes:%
\[ \wpc(c, \mathcal{F}) = \lambda s. \EE_{x \sim \denot{c} s}[\mathcal{F}(c)]\]
\begin{figure*}[t]
  \scalebox{0.85}{\parbox{\linewidth}{%
	\begin{align*}
  \wpe{\Skip}{\Exp} &\triangleq \Exp \\
  \wpe{\Assn{x}{e}}{\Exp} &\triangleq \Exp \{ e / x \} \\
  \wpe{\Rand{x}{d}}{\Exp} &\triangleq \lambda s. \EE_{x\sim d}[\Exp \{e / x \} ] \\
  \wpe{c; c'}{\Exp} &\triangleq \wpe{c}{\wpe{c'}{\Exp}} \\
  \wpe{\Cond{e}{c}{c'}}{\Exp} &\triangleq 
  	\ind{e} \cdot \wpe{c}{\Exp}
  	+ \ind{\neg e} \cdot \wpe{c'}{\Exp} \\
  \wpe{\WWhile{e}{c}}{\Exp} &\triangleq {\sf lfp} X.
 		 \ind{e} \cdot \wpe{c}{X}
		  + \ind{\neg e} \cdot \Exp
\end{align*}
}}
\caption{Definition of the weakest pre-expectation operator $\wpe{c}{\Exp}$}
\label{fig:wpe-rules}
\end{figure*}%
Computing lower bounds using $\wpc$ poses some technical challenges. First, we
need to find some expectation $f$ that will achieve a lower bound that is as
large as possible. Second, as we are computing a difference of two expected
values, we need to be able to compute exact pre-expectations---standard
invariant rules for $\wpc$ only produce \emph{upper bounds} on the expectation.
Overcoming the first problem requires ingenuity, but the second problem can be
addressed by using a technical result of \citet{KaminskiKMO16}. We start by
defining upper and lower invariants.

\begin{definition}
  A family of unary expectations $\{I_n\}_{n\in\mathbb{N}}$ is an upper
  $\omega$-invariant of the loop $\WWhile{b}{c}$ with respect to the expectation
  $\mathcal{F}$ if
  \[
		[\neg b] \cdot \mathcal{F} \leq I_0
    \qquad\text{and}\qquad
		[b] \cdot \wpc(c, I_n) + [\neg b] \cdot \mathcal{F} \leq I_{n+1}~.
  \]
  The definition of lower $\omega$-invariant is analogous, reversing the two
  inequalities.
\end{definition}%
\noindent{}These invariants can be used to compute exact weakest
pre-expectations of loops.%
\begin{theorem}[{\citet[Theorem 5]{KaminskiKMO16}}]
	Let $I_n$ and $J_n$ be an upper and a lower $\omega$-invariant of
	$\WWhile{b}{c}$ with respect to $\mathcal{F}$, respectively. If the limits $\lim_{n\to\infty} I_n$
	and $\lim_{n\to\infty} J_n$ exist, then
	\[ \lim_{n\to\infty} J_n \leq \wpc(\WWhile{b}{c}, \mathcal{F}) \leq \lim_{n\to\infty} I_n~. \]%
When the limits coincide, the weakest pre-expectation is determined exactly.
\end{theorem}%
\noindent{}We illustrate our technique on the example from Section~\ref{ex:hypercube}.  Let $w_H$
be the expression denoting the normalized Hamming weight of a vector, i.e.  $w_H
\triangleq (1/N) \sum_{i=0}^N \mathit{pos}[i]$.
We can use the expected value of $w_H$ to compute a lower bound for the $TV$
distance between two runs of $\mathbf{hWalk}$.
We will compute this expected value by using the usual (unary) pre-expectation
calculus.
To compute \emph{exact} weakest preconditions, we need to find upper and lower
invariants.
We consider the following $\omega$-invariant $\cI_n \triangleq [K - n \leq k ](
\mathcal{B}_k + \mathcal{A}_k \cdot w_H)$, where
\[
\mathcal{A}_k \triangleq \left( \frac{N-1}{N+1} \right)^{K \ominus k}
\quad\text{and}\quad
\mathcal{B}_k \triangleq \dfrac{1}{N+1} \cdot \sum_{i = 0}^{K-k-1} \left (\dfrac{N-1}{N+1} \right)^i~.
\]
We can check that $\cI_n$ is both an upper and a lower $\omega$-invariant,
therefore the weakest pre-expectation for the loop is exactly $\lim_{n\to\infty}
\cI_n = \mathcal{B}_k + \mathcal{A}_k \cdot w_H$.
After computing its pre-expectation with respect to the first assignment, we get:
\[
  \wpc(\mathbf{hWalk}, w_H) = \mathcal{B}_0 + \mathcal{A}_0 \cdot w_H
\]
%
%
Now, let $W(p) \triangleq (1/N) \sum_{i=0}^N p[i]$, and let $s_1, s_2$ be any
two initial states such that $s_1(N) = s_2(N)$ and $s_1(K) = s_2(K)$. By
\cref{prop:TV-expected-value-def}, we can lower bound the TV distance as
follows:
\begin{align}
\label{eq:lower-bound}
TV(\denot{\mathbf{hWalk}}(\mathit{s_1(\mathit{pos})},\!N,\!K), \denot{\mathbf{hWalk}}(\mathit{s_2(\mathit{pos})},\!N,\!K))
  &\geq |\wpc(\mathbf{hWalk}, w_H)(s_1) - \wpc(\mathbf{hWalk}, w_H)(s_2)|
  \notag \\
  &= \left(\dfrac{N\!-\!1}{N\!+\!1}\right)^{K} \!\!\!\!\!\cdot |W(s_1(\mathit{pos})) \!-\! W(s_2(\mathit{pos}))|~.
  \notag
\end{align}
By selecting the initial positions appropriately---essentially, picking
worst-case inputs---we can derive useful lower bounds on the Total Variation
distance between output distributions. For instance, taking $s_1(\mathit{pos})$
to be the all-zeros vector and $s_2(\mathit{pos})$ to be the all-ones vector
gives:
\[
TV(\denot{\mathbf{hWalk}}(\mathit{s_1(\mathit{pos})},\!N,\!K), \denot{\mathbf{hWalk}}(\mathit{s_2(\mathit{pos})},\!N,\!K))
\ \geq \ \left(\dfrac{N\!-\!1}{N\!+\!1}\right)^{K}~.
\]
Using standard algebraic bounds, the right-hand side (and the TV distance
between the two output distributions) is at least $\rho > 0$ when $K < (N/2)
\log \rho$.


Verifying precise lower bounds is highly challenging---for many simple examples
of randomized processes, exact lower bounds are not known. Nonetheless, efforts
in this direction could provide useful, complementary information when analyzing
probabilistic programs.

\section{Extensions: Rules for Asynchronous Reasoning}
\label{sec:async}

Our relational pre-expectation operator $\gcp{c}{\Exp}$ can often derive useful
upper bounds on the Kantorovich distance $\sgcp{c}{\Exp}$, but it gives a
trivial bound of infinity when the program $c$ can take different branches on
the two inputs. In this section, we develop techniques to give more useful
bounds in the asynchronous case.

\subsection{Asynchronous rules for bounding the Kantorovich distance}

Our asynchronous bounds will use one-sided relational operators $\wpel{c}{\Exp}$
(resp.  $\wper{c}{\Exp}$) that transform relational expectations by holding the
left (right) state constant and then computing the unary weakest pre-expectation
$\wpe{c}{\Exp}$. We use the following soundness lemma for the left version of
the operator, the one for the right version being analogous.

\begin{lem} \label{lem:wpe:sound}
Let $c$ be a \pwhile\ program that is almost surely terminating, i.e.,
$\wpe{c}{1} = 1$. Then, for all $s_1,s_2$, $\EE_{s_1' \sim
\denot{c}s_1}[\Exp(s_1', s_2)] \leq \wpel{c}{\Exp}(s_1, s_2)$.
\end{lem}%
\noindent{}%
Now we can present our asynchronous rules.
\begin{theorem} \label{thm:async}
Let $c$ be a program that is almost surely terminating. Then:
\begin{align*}
  \sgcp{\Cond{e}{c}}{\Exp} &\leq \ind{e\sidel \land e\sider} \cdot \gcp{c}{\Exp}
  + \ind{e\sidel \land \neg e\sider}\cdot \wpel{c}{\Exp} \\
                            &+ \ind{\neg e\sidel \land e\sider}\cdot \wper{c}{\Exp} 
                            + \ind{\neg e\sidel \land \neg e\sider}\cdot \Exp
\end{align*}
Let $\WWhile{e}{c}$ be an almost surely terminating loop, $\rho_i(s)$ be the
probability that the loop does not terminate after executing the body at most
$i$ times starting from state $s$, and:
\[
  M_i(\Exp, s_1, s_2) = \max \{ \Exp(t_1, t_2) \mid t_1 \in
  \supp(\denot{c_i}s_1), t_2 \in \supp(\denot{c_i}s_2) \}
\]
where $c_i$ is the first $i$ iterations of the loop. If $\rho_i$ and $M_i$
satisfy:
\[
  \lim_{i \to \infty} (\rho_i(s_1) + \rho_i(s_2)) \cdot M_i(\Exp, s_1, s_2) = 0
\]
for any two states $(s_1, s_2)$, and if $\Inv$ is an invariant satisfying
\begin{align*}
  &\ind{e\sidel \land e\sider} \cdot \gcp{c}{\Inv}
  + \ind{e\sidel \land \neg e\sider}\cdot \wpel{c}{\Inv} \\
  &+ \ind{\neg e\sidel \land e\sider}\cdot \wper{c}{\Inv} 
  + \ind{\neg e\sidel \land \neg e\sider}\cdot \Exp \leq \Inv~,
\end{align*}
then $\sgcp{\WWhile{e}{c}}{\Exp} \leq \Inv$.
\end{theorem}
\begin{proof}[Proof Sketch]
  The soundness of the conditional rule follows a similar argument as soundness
  for the definition of $\gcpsymbol$ for conditionals, using
  \cref{lem:wpe:sound} for the asynchronous cases. The soundness of the loop
  rule is more intricate, but it follows the same strategy as
  in~\cref{thm:loop-sound}: we define a loop characteristic function based on
  the conditional rule (now asynchronous), show that the least fixed-point lies
  above $\mathit{rpe}$, and finally show that the invariant rule implies that
  $\Inv$ is a pre-fixed-point, so it must be above the fixed point.
\end{proof}

We detail the full proof in \cref{sec:async-extra}.

\subsection{Example: Bounding the distance between binomial distributions}

Consider the following program, which simulates a binomial distribution:
\[
  \begin{array}{l}
    \mathbf{binom}(N)\\
    \quad \Assn{n}{0}; \\
    \quad \Assn{k}{0}; \\
    \quad \WWhile{n<N} \\
    \qquad \Rand{b}{\mathbf{Bern}(p)}; \\
    \qquad \Condt{b}{\Assn{k}{k+1}}; \\
    \qquad \Assn{n}{n+1};
  \end{array}
\]
We treat $p \in [0, 1]$ as a fixed constant. We will compare the distribution
on the output $k$ starting from two inputs. Since the loops will run for
different numbers of iterations if $N\sidel \neq N\sider$, we will employ our
asynchronous rule. We take the following invariant:
\[
  \Inv \triangleq |\, k\sidel - k\sider + p\cdot(N\sidel \ominus n\sidel) - p\cdot(N\sider \ominus n\sider) \, |~,
\]
We will show the following invariant bound:
\begin{align*}
  &\ind{(n < )\sidel \land (n < N)\sider} \cdot \gcp{c}{\Inv}
  + \ind{(n < N)\sidel \land (n \geq N)\sider}\cdot \wpel{c}{\Inv} \\
  &+ \ind{(n \geq N)\sidel \land (n < N)\sider}\cdot \wper{c}{\Inv} 
  + \ind{(n \geq N)\sidel \land (n \geq N)\sider}\cdot \Exp \leq \Inv~.
\end{align*}
In the synchronous case, we can establish the invariant by applying
\textsc{Samp} with the identity coupling; the inner conditional can also be
analyzed synchronously. In the asynchronous case, computing the unary weakest
pre-expectation establishes the invariant. Thus, the asynchronous loop rule
(\cref{thm:async}) gives:
\[
  \sgcp{w}{| \, k\sidel - k\sider \, |} \leq \Inv
\]
where $w$ is the loop. Applying the assignment rule, we conclude:
\[
  \sgcp{\mathbf{binom}(N)}{| \, k\sidel - k\sider \, |}
  \leq p \cdot | \, N\sidel - N\sider \, | .
\]
By \cref{thm:abs-diff}, this bound implies that the expected values of the
output $k$ differ by at most $p \cdot |N\sidel - N\sider|$ across the two runs.

\section{Related work}
\label{sec:rw}

\textit{Proving expected sensitivity of probabilistic programs.}  We have shown
that the quantitative logic \eprhl~\citep{BartheEGHS18} can be embedded into the
framework of this paper (cf.\ Section~\ref{sec:embedding_eprhl}), so we focus on
other work.  \citet{wang2019proving} propose an alternative method based on
martingales for proving the expected sensitivity of probabilistic programs.
Their technique focuses on computing the expected sensitivity when the
(expected) number of iterations for a loop may be different across two related
executions (i.e., loops may be \emph{asynchronous}); this is similar to our
asynchronous rules from \cref{sec:async}.
However, \citet{wang2019proving} also frame their target property in a slightly weaker way, showing that programs are Lipschitz continuous for \emph{some} finite Lipschitz constant. 
In contrast, our method establishes bounds on this constant, which is an important aspect in many applications (e.g., it determines the rate of convergence for Markov chains). 
We are also able to handle the broader class of expected sensitivity properties arising from Kantorovich metrics, subsuming the notion considered by \citet{wang2019proving} where the output distance is the absolute difference between two expected values.
\\[0.5em]
\noindent
\textit{Formal reasoning for probabilistic programs.}
Logics for probabilistic programs has been an active research area since the
1980s.  Seminal work by \citet{Kozen85} defines a probabilistic propositional
dynamic logic for reasoning about probabilistic programs, using real-valued
functions rather than boolean assertions.  \citet{Morgan96} define a weakest
pre-expectation calculus for a programming language with (demonic)
non-determinism and probabilities. Extensions of this calculus with recursion
and conditioning have been
considered~\citep{DBLP:conf/lics/OlmedoKKM16,DBLP:journals/toplas/OlmedoGJKKM18}.
\citet{KaminskiKMO16} define a similar calculus for bounding expected run-times
of probabilistic programs. 
These works do not prove relational properties of programs, and are unsuitable
for verifying sensitivity.
%
%
\\[0.5em]
\textit{Continuity in programs and process calculi.}
\noindent
Formal reasoning about the continuity of deterministic programs has received some attention.
\citet{ChaudhuriGL10,DBLP:journals/cacm/ChaudhuriGL12} were the first to give a sound, compositional framework for verifying that a program is continuous.
\citet{ReedP10} gave a type system that can verify Lipschitz continuity of functional programs (see also \citep{GHHNP13,AGGH14,AmorimGHKC17,adafuzz}).
Recently, \citet{DBLP:conf/atva/HuangWM18} proposed the tool \texttt{PSense} which can perform sensitivity analysis of probabilistic programs. 
Their technique relies on symbolic computation using the symbolic verifier PSI and Mathematica, and supports, e.g., the Total Variation distance
and the expectation distance. 
\texttt{PSense} cannot reason, however, about general Kantorovich distances, or unbounded loops.

Finally, in the process-algebra setting,
compositional reasoning about metrics has received 
some attention. \citet{DBLP:journals/corr/GeblerLT16} used
uniform continuity 
to reason about the distance
between recursive processes in a compositional way, while
\citet{DBLP:journals/jcss/GeblerT18} recently defined
specification formats that can check uniform continuity syntactically. 
A more general framework for reasoning about metrics has been given by
\citet{DBLP:conf/lics/BacciMPP18}, who presented an algebraic
axiomatization of Markov processes in quantitative
equational logic. Their 
framework supports reasoning about
various metrics, including the Kantorovich metric.

\section{Conclusion}
\label{sec:conc}
We defined a pre-expectation calculus to compute upper bounds for Kantorovich
metrics, and applied it to prove convergence of reinforcement learning and card
shuffling algorithms, algorithmic stability of SGD, and uniformity of limit
distributions. Our calculus provides theoretical foundations for reasoning about
quantitative relational properties of probabilistic programs.

There are several natural directions for future work. One possible extension is
to lift the requirement that programs terminate with equal probability on pairs
of executions, possibly by leveraging alternative notions of the Kantorovich
metric that accommodate distributions of different weight \citep{Piccoli2016}.
Other directions include developing a relational version of quantitative
separation logic~\citep{Batz19}, and use it for proving relational
properties of probabilistic heap-manipulating programs.


We also explored methods for proving lower bounds of convergence speed. In
general, we are not aware of many works that prove lower bounds using program
logics, with some notable exceptions~\cite{DBLP:journals/pacmpl/HarkKGK20}.
Developing more tools and techniques for reasoning about these fascinating
properties is an interesting avenue for future work.


\bibliographystyle{ACM-Reference-Format}
\bibliography{header,main}

\clearpage
\appendix


\section{Background: Real Analysis}

The following are standard convergence results in real analysis, see for
instance~\cite{williams_1991}. In all of them we consider a sequence of
relational expectations $\Exp_n \colon \Mem \times \Mem \to \RRe$ and a
distribution $\mu \colon \Dist(\Mem\times\Mem)$.

\begin{lem}[Fatou's Lemma]  \label{lem:fatou}
  Let $\Exp_n$ be a monotone increasing sequence of relational expectations.
  Then,
  \[ \EE_\mu[\lim_{n \to \infty} \Exp_n] \leq \lim_{n \to \infty} \EE_\mu[\Exp_n] . \]
\end{lem}

Fatou's Lemma also holds when $\Exp_n$ is not a monotone sequence (replacing the
limit by a limit inferior), but the monotone version suffices for our purposes.

Now we present a result that will be useful in showing convergence of couplings.
A similar result can be found in the monograph~\citet[Theorem 5.19]{Villani08}.

\begin{theorem}[Convergence of couplings]
\label{thm:conv-couplings}
Let $\nu_i$ and $\rho_i$ denote two sequences of sub-distributions with
countable support over $X$, converging pointwise to $\nu$ and $\rho$
respectively. Let $\mu_i \in \Gamma(\nu_i, \rho_i)$ be a sequence of couplings
of $\nu_i$ and $\rho_i$. Then there exists a subsequence $\mu_i'$ of $\mu_i$
that converges to a coupling $\mu \in \Gamma(\nu,\rho)$.
\end{theorem}

\begin{proof}
The proof proceeds in two steps. First we show that there exists a convergent
subsequence of couplings. By the Bolzano-Weierstrass theorem, $[0, 1]$ is
\emph{sequentially compact}, i.e., every sequence in $[0,1]$ has a subsequence
that converges in $[0,1]$. Moreover, countable products preserve sequential
compactness. Since every $\nu_i$ and $\rho_i$ have countable support, so does
every $\mu_i$, so we can consider the sequence $\{\mu_i\}_{i\in\NN}$ as a
sequence over $[0,1]^\mathcal{S}$ where $\mathcal{S} = \cup_i \supp(\mu_i)$.
Since this is a sequentially compact space, we can extract a subsequence
$\{\mu_i'\}_{i\in\NN}$ that converges pointwise to some distribution, call it
$\mu \in \Dist(X \times X)$; let $\{\nu_i'\}_{i\in\NN}$ and
$\{\rho_i'\}_{i\in\NN}$ be the corresponding subsequences of
$\{\nu_i\}_{i\in\NN}$ and $\{\rho_i\}_{i\in\NN}$ such that $\mu'_i \in
\Gamma(\nu_i', \rho_i')$. Since they are subsequences of convergent sequences,
these were convergent, $\{ \nu_i' \}_{i \in \NN}$ converges to $\nu$ and $\{
\rho_i' \}_{i \in \NN}$ converges to $\rho$.

The main task is showing that $\mu$ is indeed a coupling of $\nu$ and $\rho$,
i.e., $\mu \in \Gamma(\nu, \rho)$. We consider the first marginal condition.
Let $\epsilon > 0$ be any positive number. Since $\rho$ is a distribution over a
countable set, there exists a finite set $S(\epsilon)$ such that $\sum_{x_2
\notin S(\epsilon)} \rho(x_2) < \epsilon$. We first show that:
\begin{equation} \label{eq:tails-uniform-limit}
  \lim_{i \to \infty} \sum_{x_2 \notin S(\epsilon/2)} \rho_i(x_2) = 0 .
\end{equation}
Since $\{ \rho_i' \}_i$ converges pointwise to $\rho$, it also converges in
$L_1$. So, there exists a finite number $N(\epsilon)$ such that for all $i >
N(\epsilon)$, we have:
\[
  \sum_{x_2 \in X} | \, \rho_i'(x_2) - \rho(x_2) \, | < \epsilon .
\]
Thus for all $i > N(\epsilon/2)$, we have:
\begin{align*}
  \sum_{x_2 \notin S(\epsilon/2)} \rho_i'(x_2) 
  &= \sum_{x_2 \notin S(\epsilon/2)} (\rho_i'(x_2) - \rho(x_2)) + \sum_{x_2 \notin S(\epsilon/2)} \rho(x_2)
  \\
  &\leq \sum_{x_2 \in X} | \, \rho_i'(x_2) - \rho(x_2) \, | + \epsilon/2
  \\
  &\leq \epsilon .
\end{align*}
Since $\epsilon$ was arbitrary, this establishes \cref{eq:tails-uniform-limit}.
Now, let $x_1 \in X$ be any element. We can compute:
\begin{align*}
  | \, \nu(x_1) - (\pi_1(\mu))(x_1) \, |
  &= | \, \nu(x_1) - \sum_{x_2 \in X} \lim_{i \to \infty} \mu_i(x_1, x_2) \, |
  \notag \\
  &\leq \, | \nu(x_1) - \sum_{x_2 \in S(\epsilon)} \lim_{i \to \infty} \mu_i(x_1, x_2) \, |
  + \sum_{x_2 \notin S(\epsilon)} \lim_{i \to \infty} \mu_i(x_1, x_2)
  \tag{triangle ineq.} \\
  &= \, | \nu(x_1) - \lim_{i \to \infty} \sum_{x_2 \in S(\epsilon)} \mu_i(x_1, x_2) \, |
  + \sum_{x_2 \notin S(\epsilon)} \lim_{i \to \infty} \mu_i(x_1, x_2)
  \tag{$S(\epsilon)$ finite} \\
  &\leq | \nu(x_1) - \lim_{i \to \infty} \sum_{x_2 \in S(\epsilon)} \mu_i(x_1, x_2) \, |
  + \sum_{x_2 \notin S(\epsilon)} \lim_{i \to \infty} \rho_i(x_2)
  \tag{$\pi_2(\mu_i) = \rho_i$} \\
  &\leq | \, \nu(x_1) - \lim_{i \to \infty} \sum_{x_2 \in S(\epsilon)} \mu_i(x_1, x_2) \, |
  + \sum_{x_2 \notin S(\epsilon)} \rho(x_2)
  \tag{limit $\rho$} \\
  &\leq | \, \nu(x_1) - \lim_{i \to \infty} \sum_{x_2 \in S(\epsilon)} \mu_i(x_1, x_2) \, | + \epsilon
  \tag{def. $S(\epsilon)$} \\
  &= | \, \nu(x_1) - \lim_{i \to \infty} \nu_i(x_1) + \lim_{i \to \infty} \sum_{x_2 \notin S(\epsilon)} \mu_i(x_1, x_2) \, | + \epsilon
  \tag{$\pi_1(\mu_i) = \nu_i$} \\
  &= | \, \lim_{i \to \infty} \sum_{x_2 \notin S(\epsilon)} \mu_i(x_1, x_2) \, | + \epsilon
  \tag{limit $\nu$} \\
  &\leq \lim_{i \to \infty} \sum_{x_2 \notin S(\epsilon)} \rho_i(x_2) + \epsilon
  = \epsilon \tag{by \cref{eq:tails-uniform-limit}} .
\end{align*}
Since $\epsilon > 0$ and $x_1 \in X$ are arbitrary, this shows the first
marginal condition $\nu = \pi_1(\mu)$. The second marginal condition follows
similarly, and so $\mu \in \Gamma(\nu, \rho)$ as desired.
\end{proof}

\section{Program semantics}

A state $s\in\Mem$ is a map from a finite set of variable names $\Var$ to a set of values
$\Val$. Given an expression $e$, we abuse the notation $s(e)$ to denote the natural lifting
of $s$ to a map from expressions to values. Similarly, given an expression $d$ denoting
a distribution, we abuse the notation $s(d)$ to denote the lifting of $s$ to a map from
distributions to distribution over values.
Given $s\in \Mem$, $x\in \Var$ and $v \in \Val$, we write $s \{ v / x \}$ to denote
the unique state such that $s \{ v / x \} (y) = v$ if $y = x$ and $s \{ v / x \}(y) = s(y)$
otherwise.

The semantics $\denot{c}$ of a command $c$ is a map from an input state in $\Mem$ to 
an output distribution in $\Dist(\Mem)$. This semantics is standard, and is defined
by induction on the structure of the command:
\begin{align*}
  \denot{\Skip} s &\triangleq \delta(s) \\
  \denot{\Assn{x}{e}} s &\triangleq \delta(s \{ s(e) / x \}) \\
  \denot{\Rand{x}{d}} s &\triangleq \EE_{v \sim s(d)}[s \{ v / x \}] \\ 
  \denot{c; c'} s &\triangleq \EE_{s' \sim \denot{c} s}[\denot{c'}s']  \\
  \denot{\Cond{e}{c}{c'}} s &\triangleq 
  	\ind{s(e)} \cdot \denot{c} s
	+ \ind{\neg s(e)} \cdot \denot{c'} \delta(s) \\
  \denot{\WWhile{e}{c}} s &\triangleq \lim_{n \to \infty} \denot{c_n}s
  \quad \text{where}\ c_0 \triangleq \Abort \text{ and } c_{i + 1} \triangleq \Condt{e}{c; c_i}
\end{align*}
We use a dummy $\Abort$ command that denotes the constant zero sub-distribution
to help define the semantics for loops. The limit exists and is a
sub-distribution because for any initial state $s$, the sub-distributions
$\denot{c_i}s$ are monotone increasing in $i$ under the pointwise order on
sub-distributions, i.e., $(\denot{c_i}s)(s') \leq (\denot{c_j}s)(s')$ for all
states $s, s' \in \Mem$ and all $i \leq j$, and $(\denot{c_i}s)(s')$ is
bounded above by $1$.

\section{Section 2: Omitted Proofs}

\begin{proof}[Theorem~\ref{thm:tv}]
For the direction ``$\leq$'', it suffices to show that $
  \bigl| \mu_1(S) - \mu_2(S) \bigr| \leq \EE_{\mu}[ \Exp]
$
for every
$\mu \in \Gamma(\mu_1, \mu_2)$ and every $S\subseteq X$.
By the property of marginals and monotonicity of probabilities, we have:
\begin{align*}
  \bigl| \mu_1(S) - \mu_2(S) \bigr| & = \bigl| \mu(S\times X) - \mu(X \times S) \bigr|
                                     = \bigl| \mu(S \times (X \setminus S)) - \mu((X \setminus S) \times S) \bigr| \\
                                    & \leq \max(\mu(S \times (X \setminus S)),\mu((X \setminus S) \times S))
                                     \leq \mu (\{(x_1,x_2) \mid x_1\neq x_2\})
                                     = \EE_{\mu}[ \Exp]~.
\end{align*}
For the other direction, we construct a so-called optimal coupling.  For every
$x\in X$, let $\mu_0(x)=\min (\mu_1(x),\mu_2(x))$.  The optimal coupling
for $(\mu_1,\mu_2)$ is defined by:
\[
  \mu(x_1,x_2) = \begin{cases}
    \mu_0(x_1) &: x_1=x_2 \\
    \frac{(\mu_1(x_1)-\mu_0(x_1)) \cdot (\mu_2(x_2)-\mu_0(x_2))}{TV(\mu_1,\mu_2)} &: x_1\neq x_2~,
\end{cases}
\]
where $0/0 \coloneqq 0$. One can check that $\mu$ is a
coupling for $(\mu_1,\mu_2)$ and that $\mu([x_1\neq x_2])=TV(\mu_1,\mu_2)$.
\end{proof}%

\begin{proof}[Theorem~\ref{thm:abs-diff}]
  It suffices to show $\bigl| \EE_{\mu_1}[ f_1 ] - \EE_{\mu_2}[ f_2 ] \bigr|
    \leq \EE_{\mu}[ \Exp ]$ for every $\mu \in \Gamma(\mu_1,\mu_2)$, which follows from the marginal properties of couplings,
  linearity of expectation, and the fact that ${|\EE_\mu[f]|\leq
  \EE_\mu[|f|]}$:
  \begin{align*}
           &\phantom{=} \ \bigl| \EE_{\mu_1}[ f_1 ] - \EE_{\mu_2}[ f_2 ] \bigr|
           = \bigl| \EE_{\mu}[\lambda (x_1,x_2).~f_1(x_1)] - \EE_{\mu}[\lambda (x_1,x_2).~f_2(x_2)] \bigr| \\
           &= \bigl| \EE_{\mu}[\lambda (x_1,x_2).~f_1(x_1) - f_2(x_2)] \bigr| 
           \leq \EE_{\mu}[\lambda (x_1,x_2).~\bigl|f_1(x_1) - f_2(x_2) | \bigr] 
           = \EE_{\mu}[\Exp]~. \qedhere
  \end{align*}
\end{proof}%

\begin{proof}[Theorem~\ref{thm:scaled-tv}]
	The proof follows from Theorem~\ref{thm:tv} and from the
  observations that $\Kant{(\rho \cdot \Exp)}=\rho\cdot \Kant{\Exp}$
  and that $\Exp\leq\Exp'$ implies $\Kant{\Exp}\leq\Kant{\Exp'}$, taking the
  pointwise order in both cases.
\end{proof}%

\section{Soundness and Continuity: Omitted Proofs}

The syntactic relational pre-expectation transformer is a monotonic operator.

\begin{lem}[Monotonicity of $\gcp{c}{-}$] \label{lem:monotonic}
  Let $\Exp$ be a relational expectation and let $c$ be a program. Then
  $\gcp{c}{-}$ and $\Phi_{\Exp, c}(-)$ are monotonic, i.e.\ for any two
  relational expectations $\Exp_1, \Exp_2$ such that $\Exp_1 \leq \Exp_2$, we have
  $\gcp{c}{\Exp_1} \leq \gcp{c}{\Exp_2}$ and $\Phi_{\Exp, c}(\Exp_1) \leq
  \Phi_{\Exp, c}(\Exp_2)$.
\end{lem}
\begin{proof}
  The latter result is a corollary from the former. By definition,
  \[ \Phi_{\Exp,c,e}(\Exp_1) = \ind{e\sidel \land e\sider} \cdot \gcp{c}{\Exp_1}
        + \ind{\neg e\sidel \land \neg e\sider} \cdot \Exp
+ \ind{e\sidel \neq e\sider} \cdot \infty \]
  and
  \[ \Phi_{\Exp,c,e}(\Exp_2) = \ind{e\sidel \land e\sider} \cdot \gcp{c}{\Exp_2}
        + \ind{\neg e\sidel \land \neg e\sider} \cdot \Exp
+ \ind{e\sidel \neq e\sider} \cdot \infty. \]
  So given $\gcp{c}{\Exp_1} \leq \gcp{c}{\Exp_2}$ we can conclude 
  $\Phi_{\Exp, c, e}(\Exp_1) \leq \Phi_{\Exp, c, e}(\Exp_2)$.

  The former result is proven by induction on $c$:
	\begin{itemize}
		\item $\Skip$. Then
			\[ \gcp{\Skip}{\Exp_1} = \Exp_1 \leq \Exp_2 = \gcp{\Skip}{\Exp_2} \]

		\item $\Assn{x}{e}$. Then
			\[ \gcp{\Assn{x}{e}}{\Exp_1} = \Exp_1\{ e\sidel, e\sider / x\sidel, x\sider \} \]
			and
			\[ \gcp{\Assn{x}{e}}{\Exp_2} = \Exp_2\{ e\sidel, e\sider / x\sidel, x\sider \} \]
			Consider a pair of states $s_1,s_2$ then:
			\begin{align*} 
				\Exp_1\{ e\sidel, e\sider / x\sidel, x\sider \}(s_1, s_2) &= 
				\Exp_1(s_1\{ s_1(e) / x\})(s_2\{ s_2(e) / x\}) \\
				&\leq\Exp_2(s_1\{ s_1(e) / x\})(s_2\{ s_2(e) / x\}) \\
				&=\Exp_2\{ e\sidel, e\sider / x\sidel, x\sider \}(s_1, s_2)
			\end{align*}

		\item $\Rand{x}{d}$. Then,
			\[ \gcp{\Rand{x}{d}}{\Exp_1} = \inf_{\mu \in \Gamma(\mu_1, \mu_2)} \EE_{\mu} [ \Exp_1 ] \]
			and 
			\[ \gcp{\Rand{x}{d}}{\Exp_2} = \inf_{\mu \in \Gamma(\mu_1, \mu_2)} \EE_{\mu} [ \Exp_2 ] \]
			Let $\mu\in\Gamma(\mu_1, \mu_2)$ be an arbitrary coupling. By monotonicity of the expectation,
			then $\EE_{\mu} [ \Exp_1 ] \leq \EE_{\mu} [ \Exp_2 ]$, and therefore the infimum for $\Exp_1$ is
			less or equal than the one for $\Exp_2$.

		\item $c; c'$. By the induction hypothesis,
			\[ \gcp{c; c'}{\Exp_1} = \gcp{c}{\gcp{c'}{\Exp_1}} \leq \gcp{c}{\gcp{c'}{\Exp_2}} = \gcp{c; c'}{\Exp_2}  \]
			Note that the inequality needs two applications of the I.H., one to show that 
			$\gcp{c'}{\Exp_1} \leq \gcp{c'}{\Exp_2}$ and another one to show 
			$\gcp{c}{\gcp{c'}{\Exp_1}} \leq \gcp{c}{\gcp{c'}{\Exp_2}}$.

		\item $\Cond{e}{c}{c'}$. By the induction hypothesis (applied at $c$ and $c'$),  	
			\begin{align*}
				\gcp{\Cond{e}{c}{c'} }{\Exp_1}
				&= \ind{e\sidel \land e\sider} \cdot \gcp{c}{\Exp_1}
  				+ \ind{\neg e\sidel \land \neg e\sider} \cdot \gcp{c'}{\Exp_1}
				+ \ind{e\sidel \neq e\sider} \cdot \infty \\
				&\leq \ind{e\sidel \land e\sider} \cdot \gcp{c}{\Exp_2}
  				+ \ind{\neg e\sidel \land \neg e\sider} \cdot \gcp{c'}{\Exp_2}
				+ \ind{e\sidel \neq e\sider} \cdot \infty \\
				&= \gcp{\Cond{e}{c}{c'} }{\Exp_2}
			\end{align*}

		\item $\WWhile{e}{c}$. Then, 
		\begin{align*} \gcp{\WWhile{e}{c}}{\Exp_1} &= {\sf lfp} X.
 		 \ind{e\sidel \land e\sider} \cdot \gcp{c}{X}
		  + \ind{\neg e\sidel \land \neg e\sider} \cdot \Exp_1
	  	  + \ind{e\sidel \neq e\sider} \cdot \infty  \\
		\gcp{\WWhile{e}{c}}{\Exp_2} &= {\sf lfp} X.
 		 \ind{e\sidel \land e\sider} \cdot \gcp{c}{X}
		  + \ind{\neg e\sidel \land \neg e\sider} \cdot \Exp_2
	  	  + \ind{e\sidel \neq e\sider} \cdot \infty
		\end{align*}
		Existence of the least fixed points is guaranteed by
		monotonicity of the functionals, which follows from the
		inductive hypothesis applied to $c$.  Suppose $X_2$ is the
		least fixpoint of the second expression. We will show that it
		is a pre-fixpoint of the first expression. 		
		\begin{align*}
		&\ind{e\sidel \land e\sider} \cdot \gcp{c}{X_2}
		  + \ind{\neg e\sidel \land \neg e\sider} \cdot \Exp_1
	  	  + \ind{e\sidel \neq e\sider} \cdot \infty \\
		  & \leq \ind{e\sidel \land e\sider} \cdot \gcp{c}{X_2}
		  + \ind{\neg e\sidel \land \neg e\sider} \cdot \Exp_2
	  	  + \ind{e\sidel \neq e\sider} \cdot \infty \\
		&= X_2
	  \end{align*}
		By Knaster-Tarski, the least fixed point of a monotonically
		increasing operator is the greatest lower bound of the set of
		pre-fixpoints. From this we conclude
		$\gcp{\WWhile{e}{c}}{\Exp_1} \leq X_2$. 
    \qedhere
	\end{itemize}
\end{proof}

We need a lemma about the existence of a coupling realizing the minimum
Kantorovich distance. 

\begin{lem} \label{lem:realize-min}
  Let $\mu_1, \mu_2 \in \Dist(\Mem)$ be two subdistributions of finite support
  with the same weight, and let $\Exp \colon \Mem \times \Mem \to \RRe$ be a
  relational expectation. There exists a coupling $\mu \in \Gamma(\mu_1, \mu_2)$
  realizing the minimum Kantorovich distance:
  \[
    \EE_{\mu} [ \Exp ]
    = \inf_{\mu \in \Gamma(\mu_1, \mu_2)} \EE_\mu [ \Exp ]
    = \Kant{\Exp}(\mu_1, \mu_2)~.
  \]
\end{lem}

This is an extremely simple case of standard existence results
in the theory of optimal transport (see, e.g., Theorem 4.1 in Villani's
monograph~\cite{Villani08}). We include a proof to keep the exposition
self-contained.

\begin{proof}
  Let $d^* = \inf_{\mu \in \Gamma(\mu_1, \mu_2)} \EE_\mu [ \Exp ]$ be the
  infimum distance. If $d^* = \infty$ then the product coupling realizes the
  distance. Otherwise, suppose that the infimum $d^*$ is finite. By the
  definition of infimum, there exists a sequence of couplings $\mu^{(1)},
  \mu^{(2)}, \dots \in \Gamma(\mu_1, \mu_2)$ such that
  \[
    \lim_{k \to \infty} \EE_{\mu^{(k)}} [ \Exp ] = d^* .
  \]
  Without loss of generality, we may assume that for each $k$ the distance
  $\EE_{\mu^{(k)}} [ \Exp ]$ is finite as well.  Let $S = \cup_k
  \supp(\mu^{(k)})$ be the union of the supports of all $\mu^{(k)}$.  Since
  $\mu_1, \mu_2$ have countable support, $S$ is countable.  Since all the
  expected distances are finite, in fact all pairs of states $(s_1, s_2) \in S$
  have $\Exp(s_1, s_2) < \infty$. By Theorem~\ref{thm:conv-couplings} we can
  find a subsequence of $\mu^{(k)}$ that is converging pointwise; define:
  \[
    \mu(s_1, s_2) = \lim_{k \to \infty} \mu^{(k)} (s_1, s_2)
  \]
  for every $s_1, s_2 \in \Mem$, where the limit is taken over the subsequence
  (so it exists). Then $\mu$ is indeed a coupling in $\Gamma(\mu_1, \mu_2)$.  To
  show that $\mu$ realizes the infimum distance, we derive:
  \begin{align*}
    \EE_{\mu} [ \Exp ]
    &= \sum_{(s_1, s_2) \in S} \Exp(s_1, s_2) \cdot \mu(s_1, s_2) \\
    &= \sum_{(s_1, s_2) \in S} \Exp(s_1, s_2) \cdot \lim_{k \to \infty} \mu^{(k)}(s_1, s_2) \\
    &\leq \sum_{(s_1, s_2) \in S} \lim_{k \to \infty} \Exp(s_1, s_2) \cdot \mu^{(k)}(s_1, s_2) \\
    &\leq \lim_{k \to \infty} \sum_{(s_1, s_2) \in S} \Exp(s_1, s_2) \cdot \mu^{(k)}(s_1, s_2) \\
    &= \lim_{k \to \infty} \EE_{\mu^{(k)}} [ \Exp ] \\
    &= d^* .
  \end{align*}
  The first inequality is because $\Exp$ may take value infinity; the second
  inequality is by Fatou's lemma.
\end{proof}

Continuity proceeds in two steps. We first need a lemma about continuity of the
Kantorovich distance. While it seems challenging to establish this lemma for
distributions with infinite support, we establish it for distributions with
finite support. 

\begin{lem} \label{lem:fin-continuity}
  Let $\mu_1, \mu_2 \in \Dist(\Mem)$ be two distributions with finite support,
  and let $\Exp_n \colon \Mem \times \Mem \to \RRe$ be a monotonically
  increasing chain of relational expectations converging pointwise to $\Exp
  \colon \Mem \times \Mem \to \RRe$. Then:
  \[
    \inf_{\mu \in \Gamma(\mu_1, \mu_2)} \EE_\mu [ \Exp ]
    = \inf_{\mu \in \Gamma(\mu_1, \mu_2)} \EE_\mu [ \lim_{n \to \infty} \Exp_n ]
    = \lim_{n\to\infty} \inf_{\mu \in \Gamma(\mu_1, \mu_2)} \EE_\mu[\Exp_n] .
  \]
\end{lem}
\begin{proof}
  If $\mu_1, \mu_2$ have different weights, then both infimums are infinity and
  we are done. It is not hard to show that
  \[
    \lim_{n \to \infty} \inf_{\mu \in \Gamma(\mu_1, \mu_2)} \EE_\mu[\Exp_n]
    \leq \inf_{\mu \in \Gamma(\mu_1, \mu_2)} \EE_\mu [ \Exp ] ,
  \]
  since $\Exp_n \leq \Exp$ and the coupling realizing the infimum (which exists by
  Lemma~\ref{lem:realize-min}) is a valid coupling in each of the limit terms. 

  Showing the other direction is more involved. Define the finite relations
  \begin{align*}
    R_{< \infty} &= \{ (s_1, s_2) \mid \Exp(s_1, s_2) < \infty \} \cap (\supp(\mu_1) \times \supp(\mu_2)) \\
    R_{\infty} &= (\supp(\mu_1) \times \supp(\mu_2)) \setminus R_{< \infty} .
  \end{align*}
  We first consider the case where
  \[
    \inf_{\mu \in \Gamma(\mu_1, \mu_2)} \EE_\mu [ \Exp ] = \infty .
  \]
  This means that every coupling must put weight on $R_{< \infty}$.
  To see this fact, note that the following infimum is realized by some
  coupling $\mu^*$:
  \[
    \inf_{\mu \in \Gamma(\mu_1, \mu_2)} \EE_\mu [ R_{ \infty} ] .
  \]
  If $\mu^*(R_{ \infty}) = 0$, then $\mu^*$ does not place any mass
  on points where $\Exp$ is infinity. Since $\mu^*$ has finite support, this
  means that $\EE_{\mu^*} [ \Exp ]$ would be finite, a contradiction. So, we
  have:
  \[
    \inf_{\mu \in \Gamma(\mu_1, \mu_2)} \EE_\mu [ R_{\infty} ]
    = \inf_{\mu \in \Gamma(\mu_1, \mu_2)} \EE_\mu [ R_{ \infty} ]
    \geq \rho
    > 0 .
  \]
  for some constant $\rho$. Now, let $M$ be any real number greater than $\rho$,
  and take $N$ large enough so that for every $(s_1, s_2) \in R_{\infty}$, we
  have $\Exp_n(s_1, s_2) > M / \rho$ for all $n > N$. Such an $N$ must exist
  since $R_{\infty}$ is finite, and $\Exp_n(s_1, s_2)$ is tending to infinity
  for $(s_1, s_2) \in R_{\infty}$. We now have
  \[
    \inf_{\mu \in \Gamma(\mu_1, \mu_2)} \EE_\mu [ \Exp_n ]
    \geq \inf_{\mu \in \Gamma(\mu_1, \mu_2)} \EE_\mu [ \ind{R_\infty} \cdot \Exp_n ]
    \geq (M / \rho) \cdot \rho 
    \geq M
  \]
  for all $n > N$. Since this is true for $M$ arbitrarily large, we must have
  \[
    \lim_{n \to \infty} \inf_{\mu \in \Gamma(\mu_1, \mu_2)} \EE_\mu [ \Exp_n ]
    = \infty
    \geq \inf_{\mu \in \Gamma(\mu_1, \mu_2)} \EE_\mu [ \Exp ]
  \]
  as claimed.

  Otherwise, suppose that the infimum is equal to $w^* < \infty$. Let $M =
  \sup_{(s_1, s_2) \in R_{< \infty}} \Exp(s_1, s_2)$ be the largest finite value
  assigned by $\Exp$.  Since $\Exp_n(s_1, s_2)$ tends to infinity for all $(s_1,
  s_2) \in R_\infty$ and $R_\infty$ is finite, we may take a subsequence
  $\Exp'_n$ such that $\Exp'_n(s_1, s_2) \geq n$ for all $(s_1, s_2) \in
  R_\infty$. Let $\nu'_i$ be a coupling realizing the infimum
  \[
    \inf_{\mu \in \Gamma(\mu_1, \mu_2)} \EE_{\mu} [ \Exp'_i ] .
  \]
  Since this infimum is less than $w^*$, we have $\nu'_i(s_1, s_2) < w^* / n$
  for every $(s_1, s_2) \in R_\infty$.  Since each $\nu'_i$ has finite support
  and takes values in $[0, 1]$, by the Bolzano-Weierstrass theorem there exists
  a subsequence $\nu''_i$ converging pointwise to $\nu^*$; we write $\Exp''_i$
  for the corresponding expectations. Note that $\nu''_i \in \Gamma(\mu_1,
  \mu_2)$ is a coupling, and $\nu''_i(R_\infty) = 0$.

  Now let $\epsilon > 0$. Let $N$ be such that for all $n > N$ and $(s_1, s_2)
  \in R_{< \infty}$, we have $|\nu''_n(s_1, s_2) - \nu^*(s_1, s_2)| < \epsilon /
  M$; such an $N$ exists since the distributions have finite support.  Then
  since $\Exp''_n(s_1, s_2) \leq \Exp(s_1, s_2) \leq M$ for all $(s_1, s_2 \in
  R_{< \infty}$, and $\nu^*(s_1, s_2) = 0$ for all $(s_1, s_2) \in R_\infty$, we
  have
  \[
    \EE_{\nu^*} [ \Exp''_n ] 
    < \inf_{\mu \in \Gamma(\mu_1, \mu_2)} \EE_{\mu} [ \Exp''_n ] + \epsilon
  \]
  for all $n > N$. The monotone convergence theorem implies:
  \[
    \EE_{\nu^*} [ \Exp ]
    = \lim_{n \to \infty} \EE_{\nu^*} [ \Exp''_n ]
    \leq \lim_{n \to \infty} \inf_{\mu \in \Gamma(\mu_1, \mu_2)} \EE_{\mu} [ \Exp''_n ] + \epsilon .
  \]
  On the other hand, we have the bound
  \[
    \inf_{\mu \in \Gamma(\mu_1, \mu_2)} \EE_{\mu} [ \Exp ]
    \leq \EE_{\nu^*} [ \Exp ] .
  \]
  Since both bounds hold for all $\epsilon$, we can conclude:
  \[
    \inf_{\mu \in \Gamma(\mu_1, \mu_2)} \EE_{\mu} [ \Exp ]
    \leq \lim_{n \to \infty} \inf_{\mu \in \Gamma(\mu_1, \mu_2)} \EE_{\mu} [ \Exp''_n ]
    = \lim_{n \to \infty} \inf_{\mu \in \Gamma(\mu_1, \mu_2)} \EE_{\mu} [ \Exp_n ] .
    \qedhere
  \]
\end{proof}

Now, we can prove continuity of relational pre-expectations, provided that
programs sample from distributions with \emph{finite} support. Note that such
programs can still produce distributions with infinite support, for instance by
sampling in a loop.

\begin{theorem*}[Continuity]
  Let $c$ be a program where all primitive distributions have finite support,
  and let $\Exp_n \colon \Mem \times \Mem \to \RRe$ be a monotonically
  increasing chain of relational expectations converging pointwise to $\Exp
  \colon \Mem \times \Mem \to \RRe$. Then,
  \[
	  \gcp{c}{\Exp} = \sup_{n\in\NN} \gcp{c}{\Exp_n} .
  \]
\end{theorem*}
\begin{proof}[Proof of Theorem~\ref{thm:fin-continuity}]
	By induction on the structure of the program.
	\begin{itemize}
		\item $\Skip$. Then,
			\[ \gcp{\Skip}{\Exp} = \Exp = \sup_{n\in\NN} \Exp_n = \sup_{n\in\NN} \gcp{\Skip}{\Exp_n} \]

		\item $\Assn{x}{e}$. Then,
			\begin{align}
        \gcp{\Assn{x}{e}}{\Exp} 
				&= \Exp\{ e\sidel, e\sider / x\sidel, x\sider \}
        \notag \\
        &= \sup_{n\in\NN} \Exp_n\{ e\sidel, e\sider / x\sidel, x\sider \}
        \tag{subst. continuous} \\
				&= \sup_{n\in\NN} \gcp{\Assn{x}{e}}{\Exp_n} \notag
			\end{align}

    \item $\Rand{x}{d}$. Let $s_1, s_2$ be any two states. By assumption,
      $\denot{d}s_1$ and $\denot{d}s_2$ have finite support, so $\mu_1 =
      \denot{\Rand{x}{d}} s_1$ and $\mu_2 = \denot{\Rand{x}{d}} s_2$ also have
      finite support. By Lemma~\ref{lem:fin-continuity} applied to $\mu_1, \mu_2$,
      we have
      \[
	      \gcp{\Rand{x}{d}}{\Exp}(s_1, s_2) = \inf_{\mu\in\Gamma(\mu_1,\mu_2)}\EE_\mu[\Exp] 
	      = \lim_{n\in\NN} \inf_{\mu\in\Gamma(\mu_1,\mu_2)}\EE_\mu[\Exp_n] = \lim_{n\in\NN} \gcp{\Rand{x}{d}}{\Exp_n}(s_1, s_2) .
      \]
			By monotonicity, the $\sup$ and the $\lim$ coincide.
		\item $c; c'$. Then,
			\begin{align} 
				\gcp{c; c'}{\Exp} 
          &= \gcp{c}{\gcp{c'}{\Exp}}
          \notag \\
          &= \gcp{c}{\sup_{n\in\NN} \gcp{c'}{\Exp_n}}
          \tag{induction} \\
          &= \sup_{n\in\NN} \gcp{c}{\gcp{c'}{\Exp_n}}
          \tag{induction} \\
          &= \sup_{n\in\NN} \gcp{c;c'}{\Exp_n}
          \tag{definition}
			\end{align}

		\item $\Cond{e}{c}{c'}$. Then,
			\begin{align*}
        &\hspace{-2em}\gcp{\Cond{e}{c}{c'}}{\Exp}
        \notag \\
        &= \ind{e\sidel \land e\sider} \cdot \gcp{c}{\Exp}
        + \ind{\neg e\sidel \land \neg e\sider} \cdot \gcp{c'}{\Exp}
        + \ind{e\sidel \neq e\sider} \cdot \infty
        \notag \\
        &= \ind{e\sidel \land e\sider} \cdot \sup_{n\in\NN} \gcp{c}{\Exp_n}
        + \ind{\neg e\sidel \land \neg e\sider} \cdot \sup_{n\in\NN} \gcp{c'}{\Exp_n}
        + \ind{e\sidel \neq e\sider} \cdot \infty
        \tag{induction} \\
        &= \sup_{n\in\NN}(\ind{e\sidel \land e\sider} \cdot \gcp{c}{\Exp_n}
        + \ind{\neg e\sidel \land \neg e\sider} \cdot \gcp{c'}{\Exp_n}
        + \ind{e\sidel \neq e\sider} \cdot \infty)
        \notag \\
        &= \sup_{n\in\NN} \gcp{\Cond{e}{c}{c'} }{\Exp_n}
        \tag{definition} 
			\end{align*}

		\item $\WWhile{e}{c}$. Then, 
			\begin{align*} \gcp{\WWhile{e}{c}}{\Exp} &= {\sf lfp} X. \Phi_{c,\Exp}(X) \\
				\text{where}\ \Phi_{c,\Exp}(X) &\triangleq
 		 \ind{e\sidel \land e\sider} \cdot \gcp{c}{X}
		  + \ind{\neg e\sidel \land \neg e\sider} \cdot \Exp
	  	  + \ind{e\sidel \neq e\sider} \cdot \infty
		\end{align*}
		We claim that:	
    \begin{align}
      \gcp{\WWhile{e}{c}}{\Exp} 
      &= {\sf lfp} X. \Phi_{c, \sup_{n\in\NN} \Exp_n}(X)
      \notag \\
      &= {\sf lfp} X. \sup_{n\in\NN} \Phi_{c, \Exp_n}(X)
      \tag{1} \\
      &= \sup_{n\in\NN} {\sf lfp} X. \Phi_{c, \Exp_n}(X)
      \tag{2} \\
      &= \sup_{n\in\NN} \gcp{\WWhile{e}{c}}{\Exp_n} 
      \tag{definition} 
    \end{align}
		Equality (1) follows from:
		\begin{align*}
		{\sf lfp} X. \Phi_{c, \sup_{n\in\NN} \Exp_n}(X)
		&= {\sf lfp} X. \ind{e\sidel \land e\sider} \cdot \gcp{c}{X }
		+ \ind{\neg e\sidel \land \neg e\sider} \cdot \sup_{n\in\NN} \Exp_n
	  	+ \ind{e\sidel \neq e\sider} \cdot \infty \\
		&= {\sf lfp} X. \sup_{n\in\NN} (\ind{e\sidel \land e\sider} \cdot \gcp{c}{X}
		+ \ind{\neg e\sidel \land \neg e\sider} \cdot \Exp_n
		+ \ind{e\sidel \neq e\sider} \cdot \infty)
		\end{align*}
	To show (2) we note that---by the Knaster-Tarski fixpoint theorem and
	the fact that $\Phi_{c,\Exp_n}(X)$ is monotonic---${\sf lfp}$ is itself
	a supremum (over the ordinals), namely 
        \[ {\sf lfp} X. \sup_{n \in \NN} \Phi_{c,\Exp_n}(X) ~=~ \sup_{m} \sup_{n \in \NN} \Phi_{c, \Exp_n}^{m}(0). \]
        Hence, the two suprema can be swapped. \qedhere
	\end{itemize}
\end{proof}

We are now ready to show soundness (Theorem~\ref{thm:sound}).

\begin{theorem*}[Soundness]
  Let $c$ be a program, and suppose that $\Exp \colon \Mem \times \Mem \to \RRe$ is a
  relational expectation. Then
  \[
    \sgcp{c}{\Exp} \leq \gcp{c}{\Exp} ~.
  \]
  Equivalently, if $\gcp{c}{\Exp}(s_1, s_2)$ is finite for input states $s_1,
  s_2 \in \Mem$ then there exists a coupling $\mu_{s_1, s_2} \in
  \Gamma(\denot{c}s_1, \denot{c}s_2)$ such that
  \[
    \EE_{\mu_{s_1, s_2}} [ \Exp] \leq \gcp{c}{\Exp}(s_1, s_2) ~.
  \]
\end{theorem*}
\begin{proof}
  Given $v \in X$, we write $\dunit{v}$ for the point distribution centered at
  $v$, and given $\mu \in \Dist(X)$ and $f \colon X \to \Dist(X)$, we write
  $\dbind{\mu}{f}$ for the distribution bind. Throughout, let $(s_1, s_2) \in
  \Mem \times \Mem$ be two initial states. We prove the second, equivalent
  formulation by induction on the structure of $c$. Suppose that
  $\gcp{c}{\Exp}(s_1, s_2)$ is finite.
	\begin{itemize}
		\item 
			$\Skip$. Take the coupling $\dunit{s_1, s_2}$. 
			Then
			\begin{align*}
				\EE_{\delta(s_1,\, s_2)} [ \Exp] ~{}={}~ \Exp(s_1,\, s_2) ~{}={}~ \gcp{\Skip}{\Exp}(s_1, s_2)~.
			\end{align*}

		\item 
			$\Assn{x}{e}$. 
			Analogous to $\Skip$, but taking the coupling $\dunit{s_1', s_2'}$, where $s_i' =
      s_i[x \mapsto \denot{e}s_i]$.
      
		\item 
			$\Rand{x}{d}$. 
      Let $F \colon \Mem \to \Dist(\Mem)$ be defined as $F =
      \denot{\Rand{x}{d}}$. $F(s_1)$ and $F(s_2)$ must have equal weights for
      $\gcp{\Rand{x}{d}}{\Exp}(s_1, s_2)$ to be finite and evidently $F(s_1)$
      and $F(s_2)$ both have countable support, so Lemma~\ref{lem:realize-min}
      implies that there exists a coupling $\mu \in
      \Gamma(\denot{\Rand{x}{d}}s_1, \denot{\Rand{x}{d}}s_2)$ such that
      \[
        \EE_\mu [ \Exp ] = \Kant{\Exp}(\denot{\Rand{x}{d}}s_1, \denot{\Rand{x}{d}}s_2)
        = \gcp{\Rand{x}{d}}{\Exp}~.
      \]
    \item $c; c'$. By induction, there exists a coupling $\mu_{s_1, s_2}
      \in \Gamma(\denot{c}s_1, \denot{c}s_2)$ such that
      \[
        \EE_{\mu_{s_1, s_2}} [ \gcp{c'}{\Exp} ] \leq \gcp{c; c'}{\Exp}(s_1, s_2) < \infty .
      \]
      As a consequence, $\gcp{c'}{\Exp}(s_1', s_2')$ must be finite for all
      pairs $(s_1', s_2') \in \supp(\mu_{s_1, s_2}) \triangleq S$. Again by
      induction, for all $(s_1', s_2') \in S$ there exists a coupling
      $\mu'_{s_1', s_2'} \in \Gamma(\denot{c'}s_1', \denot{c'}s_2')$ such that
      \[
        \EE_{\mu'_{s_1', s_2'}} [ \Exp ] \leq \gcp{c'}{\Exp}(s_1',s_2') < \infty .
      \]
      Define the following joint distribution:
      \[
	      \mu^*_{s_1,s_2}(x_1,x_2) = \dbind{(y_1,y_2)\sim\mu_{s_1, s_2}}{\mu'_{y_1, y_2}(x_1,x_2)} .
      \]
      By a routine calculation, it is not hard to show that $\mu^*$ is indeed a
      coupling in $\Gamma(\denot{c; c'}s_1, \denot{c; c'}s_2)$. Let's for instance
      compute the first marginal (the second marginal is analogous):
      \begin{align*} 
	      \pi_1(\mu^*_{s_1,s_2})(x_1) &= \sum_{x_2 \in \Mem}\mu^*_{s_1,s_2}(x_1,x_2) \\
                                    &= \sum_{x_2 \in \Mem}\dbind{(y_1, y_2) \sim \mu_{s_1, s_2}}{\mu'_{y_1, y_2}(x_1,x_2)} \\
                                    &= \dbind{(y_1, y_2) \sim \mu_{s_1, s_2}}{\sum_{x_2 \in \Mem} \mu'_{y_1, y_2}(x_1,x_2)} \\
                                    &= \dbind{(y_1, y_2) \sim \mu_{s_1, s_2}}{(\denot{c'} y_1)(x_1)} \\
                                    &= \sum_{y_1 \in\Mem}\sum_{y_2 \in\Mem}{(\denot{c'} y_1)(x_1)\cdot \mu_{s_1,s_2}(y_1,y_2)} \\
                                    &= \left( \sum_{y_1 \in\Mem}(\denot{c'} y_1)(x_1)\right) \cdot \left( \sum_{y_2 \in\Mem} \mu_{s_1,s_2}(y_1,y_2) \right) \\
                                    &= \sum_{y_1 \in\Mem}{(\denot{c'} y_1)(x_1)\cdot (\denot{c} s_1)(y_1) } \\
                                    &= (\denot{c;c'}s_1)(x_1)
      \end{align*}
      Combining inequalities and applying monotonicity of expectations yields
      \[
        \dbind{\mu^*_{s_1, s_2}}{\Exp} = \dbind{\mu_{s_1,s_2}}{\dbind{\mu'_{-, -}}{\Exp}}
        \leq \dbind{\mu_{s_1,s_2}}{\gcp{c'}{\Exp}}
        \leq \gcp{c; c'}{\Exp}(s_1, s_2) .
      \]
	
	\item 
		$\Cond{e}{c}{c'}$. Note that $e$ cannot be different between $s_1$ and $s_2$,
      otherwise 
      \[\gcp{\Cond{e}{c}{c'}}{\Exp}(s_1, s_2)\] 
      is infinite. Thus, there
      are two possible cases: either $e$ is true in both $s_1, s_2$, or $e$ is
      false in both $s_1, s_2$. In the first case, we have:
      \[
        \denot{\Cond{e}{c}{c'}} s_i = \denot{c} s_i .
      \]
      By induction, there exists a coupling $\mu_t(s_1, s_2) \in
      \Gamma(\denot{c}s_1, \denot{c}s_2)$ such that
      \[
        \dbind{\mu_t(s_1, s_2)}{\Exp} \leq \gcp{c}{\Exp}(s_1, s_2)
        = \gcp{\Cond{e}{c}{c'}}{\Exp}(s_1, s_2)
      \]
      since the right-hand side is finite by assumption.
      Similarly, if $e$ is false in both $s_1$ and $s_2$, by induction there
      exists a coupling $\mu_f(s_1, s_2) \in \Gamma(\denot{c'}s_1,
      \denot{c'}s_2)$ such that
      \[
        \dbind{\mu_f(s_1, s_2)}{\Exp} \leq \gcp{c'}{\Exp}(s_1, s_2)
        = \gcp{\Cond{e}{c}{c'}}{\Exp}(s_1, s_2)
      \]
      since the right-hand side is finite by assumption.
      Thus, we can define the desired coupling by case analysis:
      \[
        \mu(s_1, s_2) \triangleq \begin{cases}
          \mu_t(s_1, s_2) &: \denot{e}s_1 = \denot{e}s_2 = \mathit{tt} \\
          \mu_f(s_1, s_2) &: \denot{e}s_1 = \denot{e}s_2 = \mathit{ff} \\
          &: \denot{e}s_1 \neq \denot{e}s_2 \quad \text{(impossible)} 
        \end{cases}
      \]
	\item 
		$\WWhile{e}{c}$. By induction on $c$, for any pair of states $s_1',
      s_2'$ and any expectation $\Exp_c$ such that $\gcp{c}{\Exp_c}(s_1', s_2')$
      is finite, there exists a coupling $\nu_{s_1, s_2} \in
      \Gamma(\denot{c}s_1, \denot{c}s_2)$ such that
      \[
        \EE_{\nu_{s_1, s_2}} [ \Exp_c ] \leq \Exp_c(s_1', s_2') .
      \]
      
      Now, let's consider the loop. We define the following loop approximants:
      \begin{align*}
        c_0 &\triangleq \WWhile{\mathit{tt}}{\Skip} \\
        c_{i+1} &\triangleq \Cond{e}{c;c_i}{\Skip}
      \end{align*}
      Each approximant executes at most $i$ iterations of the loop; the zero-th
      approximant returns the zero distribution and does not execute any
      iterations of the loop body. By definition, the relational
      pre-expectation of $\Exp$ with respect to $c_0$ is:
      \[
	      \gcp{c_0}{\Exp} = \mathsf{lfp}\ X. \Phi_{\Exp, \Skip} (X) ,
      \]
      where the characteristic functional of a loop $\WWhile{e}{c}$ is defined as:
      \[ \Phi_{\Exp, c} (X) \triangleq \ind{e\sidel \land e\sider} \cdot \gcp{c}{X}
        + \ind{\neg e\sidel \land \neg e\sider} \cdot \Exp
      + \ind{e\sidel \neq e\sider} \cdot \]
      It is easy to show that the constant zero relational expectation is a fixed
      point for the loop $c_0$, and evidently it must be the least fixed point. So,
      $\gcp{c_0}{\Exp} = 0$. Let
      \begin{align*}
        \Exp_0 &\triangleq \gcp{c_0}{\Exp} = 0  \\
        \Exp_{i+1} &\triangleq \gcp{c_{i+1}}{\Exp} = 
        \ind{e\sidel \land e\sider} \cdot \gcp{c}{\Exp_i}
        + \ind{\neg e\sidel \land \neg e\sider} \cdot \Exp
        + \cdot \ind{e\sidel \neq e\sider} \cdot \infty 
      \end{align*} 
      By induction and definition of relational pre-expectation, $\Exp_{i} =
      \Phi^{i}_{\Exp,c}(0)$. Furthermore, by monotonicity
      (Lemma~\ref{lem:monotonic}) $\Phi^{i}_{\Exp,c}(0)$ is a monotone
      increasing sequence in $i$.

      We now need two small lemmas.

      \begin{lem} \label{claim:below-lfp}
        For every $j \in \NN$, program $c$, and relational expectation $\Exp$,
        we have:
        \[
          \Phi^{j}_{\Exp,c}(0) \leq \mathsf{lfp} X. \Phi_{\Exp,c}(X)
          = \gcp{\WWhile{e}{c}}{\Exp} .
        \]
      \end{lem}
      \begin{proof}
        By induction on $j$. The base case $j = 0$ is clear, and the inductive
        step follows by monotonicity (Lemma~\ref{lem:monotonic}):
        \[
          \Phi^{j+1}_{\Exp,c}(0) = \Phi_{\Exp,c}(\Phi^i_{\Exp,c}(0))
          \leq \Phi_{\Exp,c}(\mathsf{lfp} X. \Phi_{\Exp,c}(X))
          = \mathsf{lfp} X. \Phi_{\Exp,c}(X) .
          \qedhere
        \]
      \end{proof}

      Now, let $(s_1, s_2)$ be two given input states such that
      $\gcp{\WWhile{e}{c}}{\Exp}(s_1, s_2) < \infty$. As an immediate
      consequence, $\Phi^i_{\Exp,c}(0)(s_1, s_2)$ must be finite for all $i$, so
      $\Exp_i(s_1, s_2)$ are all finite.
      
      \begin{lem} \label{claim:couple-approx}
        For all $j \in \NN$ and $(s_1', s_2') \in \Mem \times \Mem$, if
        $\Exp_j(s_1', s_2') < \infty$ then there exists a coupling $\mu_{j,
        s_1', s_2'} \in \Gamma(\denot{c_j}s_1', \denot{c_j}s_2')$ such that
        \[
          \EE_{\mu_{j, s_1', s_2'}}[ \Exp ] \leq \Exp_j(s_1', s_2') < \infty .
        \]
      \end{lem}
      \begin{proof}
        By induction on $j$. The base case $j = 0$ is clear, taking
        the null coupling that assigns probability zero to every pair of states.
        For the inductive step, we have
        \[
          \Exp_{j + 1}
          = \ind{e\sidel \land e\sider} \cdot \gcp{c}{\Exp_j}
          + \ind{\neg e\sidel \land \neg e\sider} \cdot \Exp
          + \ind{e\sidel \neq e\sider} \cdot \infty .
        \]
        Note that $e$ must be equal in $s_1'$ and $s_2'$, since $\Exp_{j +
        1}(s_1', s_2')$ is finite. There are two cases. If $e$ is false in $s_1'$
        and $s_2'$, then $\denot{c_{j + 1}}s_1' = \dunit{s_1'}$ and $\denot{c_{j +
        1}}s_2' = \dunit{s_2'}$. We can define the coupling $\mu_{s_1', s_2'} =
        \dunit{s_1', s_2'} \in \Gamma(\denot{c_{j + 1}}s_1', \denot{c_{j +
        1}}s_2')$ and we are done, since
        \[
          \EE_{\mu_{s_1', s_2'}} [ \Exp ] = \Exp(s_1', s_2') = \Exp_{j + 1}(s_1', s_2') .
        \]
        Otherwise, suppose that $e$ is true in $s_1'$ and $s_2'$. Since $\Exp_{j +
        1}(s_1', s_2') < \infty$, we must have $\gcp{c}{\Exp_j}(s_1', s_2') <
        \infty$ as well. Hence by the (outer) induction on the structure of the
        program, there exists a coupling $\nu_{s_1', s_2'} \in
        \Gamma(\denot{c}s_1', \denot{c}s_2')$ such that
        \[
          \EE_{\nu_{s_1', s_2'}} [ \Exp_j ] \leq \Exp_j(s_1', s_2') < \infty .
        \]
        As a result, for all states $(t_1, t_2) \in \supp(\nu_{s_1', s_2'})$, we
        must have $\Exp_j(t_1, t_2)$ finite as well. By the (inner) induction
        hypothesis on $j$, there is a coupling $\mu_{j, t_1, t_2} \in
        \Gamma(\denot{c_j}t_1, \denot{c_j}t_2)$ such that 
        \[
          \EE_{\mu_{j, t_1, t_2}}[ \Exp ] \leq \Exp_j(t_1, t_2) < \infty .
        \]
        Now, we can define the coupling for the $(j + 1)$-th approximants:
        \[
          \mu_{j + 1, s_1', s_2'} \triangleq \dbind{\nu_{s_1', s_2'}}{\mu_{j, -, -}}
        \]
        We first check the distance condition. By definition, we have:
        \begin{align*}
          \EE_{\mu_{j + 1, s_1', s_2'}} [ \Exp ]
        &= \EE_{(t_1, t_2) \sim \nu_{s_1', s_2'}} [ \EE_{\mu_{j, t_1, t_2}} [ \Exp ] ] \\
        &\leq \EE_{(t_1, t_2) \sim \nu_{s_1', s_2'}} [ \Exp_j(t_1, t_2) ] \\
        &\leq \Exp_j(s_1', s_2') \\
        &\leq \Exp_{j + 1}(s_1', s_2')
        \end{align*}
        The marginal condition is not hard to show, using the marginal properties
        of $\nu_{s_1', s_2'}$ and $\mu_{j, t_1, t_2}$ combined with the
        definition of approximants: since $e$ is true in $s_1'$ and $s_2'$, we
        have $\denot{c_{j + 1}}s_1' = \denot{c; c_j}s_1'$ and $\denot{c_{j + 1}}
        s_2' = \denot{c; c_j}s_2'$. The proof follows the case for sequential
        composition.
      \end{proof}

      Since $\Exp_i(s_1, s_2) < \infty$ by assumption, we may apply
      Lemma~\ref{claim:couple-approx} with input states $s_1, s_2$ and
      expectations $\Exp_i$ to produce a sequence of couplings $\mu_{i, s_1,
      s_2} \in \Gamma(\denot{c_i}s_1, \denot{c_i}s_2)$ such that
      \[
        \EE_{\mu_{i, s_1, s_2}} [ \Exp ]
        \leq \Exp_i(s_1, s_2)
        = \gcp{c_i}{\Exp}
        = \Phi^i_{\Exp, c}(0) 
        < \infty .
      \]
      By Theorem~\ref{thm:conv-couplings}, from the sequence $\mu_{i,s_1,s_2}$
      we can extract a subsequence $\mu'_{i,s_1, s_2}$ (and a corresponding
      subsequence $c_i'$ of $c_i$) that converges monotonically to a coupling
      satisfying
      \[
        \tilde{\mu}_{s_1,s_2}
        \in \Gamma(\lim_{i \to \infty} \denot{c_i'}s_1, \lim_{i \to \infty} \denot{c_i'} s_2)
        = \Gamma(\denot{\WWhile{e}{c}} s_1, \denot{\WWhile{e}{c}} s_2) ,
      \]
      by the definition of semantics for loops. All that remains to show is:
      \[
        \EE_{(s_1',s_2')\sim \tilde{\mu}_{s_1,s_2}}[\Exp(s_1',s_2')] \leq \gcp{\WWhile{e}{c}}{\Exp}(s_1,s_2) .
      \]
      We can compute:
      \begin{align}
        \EE_{(s_1',s_2')\sim \tilde{\mu}_{s_1,s_2}}[\Exp(s_1',s_2')] 
    &= \sum_{(s_1',s_2')\in \Mem\times\Mem} \Exp(s_1',s_2') \cdot \lim_{i \to \infty} \mu'_{i,s_1,s_2}(s_1',s_2')
    \notag \\
    &\leq \sum_{(s_1',s_2')\in \Mem\times\Mem} \lim_{i \to \infty} \Exp(s_1',s_2') \cdot \mu'_{i,s_1,s_2} (s_1',s_2')
    \tag{$\Exp$ may be $\infty$} \\
    &\leq \lim_{i \to \infty} \sum_{(s_1',s_2')\in \Mem\times\Mem} \Exp(s_1',s_2') \cdot \mu'_{i,s_1,s_2} (s_1',s_2')
    \tag{by Fatou's lemma} \\
    &\leq \lim_{i \to \infty} (\gcp{c'_i}{\Exp}(s_1,s_2))
    \tag{by construction} \\
    &= (\lim_{i \to \infty} \gcp{c'_i}{\Exp})(s_1,s_2)
    \tag{definition} \\
    &= \lim_{i \to \infty} (\Phi^i_{\Exp,c}(0))(s_1,s_2)
    \tag{subsequence} \\
    &\leq ({\sf lfp} X. \Phi_{\Exp,c}(X))(s_1,s_2)
    \tag{Lemma~\ref{claim:below-lfp}} \\
    &= \gcp{\WWhile{e}{c}}{\Exp}(s_1,s_2) .
    \tag{definition}
      \end{align}
      \qedhere
  \end{itemize}
\end{proof}

\section{Embedding Relational Hoare Logics} \label{sec:embeddings}

\begin{figure*}
\begin{mathpar}
  \inferrule*[Left=Assn]
  { ~ }
  { \vdash \ehl
    {\Assn{x_1}{e_1}}{\Assn{x_2}{e_2}}
    {Q[e_1\sidel, e_2\sider / x_1\sidel, x_2\sider]; \Exp[e_1\sidel, e_2\sider / x_1\sidel, x_2\sider]}
    {Q; \Exp}
  {\text{id}}}
  \and  
  \inferrule*[Left=Rand*]
  { h : D \to D \text{ bijection} }
  { \vdash \ehl
    {\Rand{x_1}{[D]}}
    {\Rand{x_2}{[D]}}
    {\forall v \in D.\, Q[v, h(v) / x_1\sidel, x_2\sider]; \EE_{v \sim [D]} [ \Exp[v, h(v) / x\sidel, x\sider] ] }
    {Q; \Exp} 
    {\text{id}}}
  \and
  \inferrule*[Left=Seq]
  { \vdash \ehl{c_1}{c_2}{P; \Exp}{Q; \Exp'}{f} \\
  \vdash \ehl{c_1'}{c_2'}{Q; \Exp'}{R; \Exp''}{f'} }
  { \vdash \ehl{c_1; c_1'}{c_2; c_2'}{P; \Exp}{R; \Exp''}{f' \circ f} }
  \and
  \inferrule*[Left=Cond]
  { \vdash \ehl{c_1}{c_2}{P \land e_1\sidel ; \Exp}{Q; \Exp'}{f}
  \\ \vdash \ehl{c_1'}{c_2'}{P \land \neg e_1\sidel ; \Exp}{Q; \Exp'}{f}
  \\ \models P \to e_1\sidel = e_2\sider }
  { \vdash \ehl{\Cond{e_1}{c_1}{c_1'}}{\Cond{e_2}{c_2}{c_2'}}
    {P ; \Exp}{Q; \Exp'}{f} }
  \and
  \inferrule*[Left=While]
  { \vdash \ehl{c_1}{c_2}
    {P \land v\sidel = k; \Exp_k}{P \land v\sidel = k - 1; \Exp_{k - 1}}{f_k}
    \\ \models P \to e_1\sidel = e_2\sidel \land (v\sidel \leq 0 \leftrightarrow \neg e_1\sidel ) }
  { \vdash \ehl{\WWhile{e_1}{c_1}}{\WWhile{e_2}{c_2}}
    {P \land v\sidel = n; \Exp_n}{P \land v \sidel = 0; \Exp_0}
    {f_1 \circ \cdots \circ f_n} }
  \and
  \inferrule*[Left=Conseq]
  { \vdash \ehl{c_1}{c_2}{P; \Exp}{Q; \Exp'}{f}
    \\ \models P' \to (P \land f(\Exp) \leq f'(\Exp''))
    \land Q \to Q' \land (\Exp''' \leq \Exp') }
  { \vdash \ehl{c_1}{c_2}{P'; \Exp''}{Q'; \Exp'''}{f'} }
  \and
  \inferrule*[Left=Case]
  { \vdash \ehl{c_1}{c_2}{P \land R; \Exp}{Q; \Exp'}{f}
  \\ \vdash \ehl{c_1}{c_2}{P \land \neg R; \Exp}{Q; \Exp'}{f} }
  { \vdash \ehl{c_1}{c_2}{P; \Exp}{Q; \Exp'}{f} }
  \and
  \inferrule*[Left=Frame-D]
  { \vdash \ehl{c_1}{c_2}{P; \Exp}{Q; \Exp'}{f}
  \\ f(z) = k \cdot z \text{ where } k \geq 1
  \\ FV(\Exp'') \cap MV(c_1, c_2) = \emptyset }
  { \vdash \ehl{c_1}{c_2}{P; \Exp + \Exp''}{Q; \Exp' + \Exp''}{f} }
\end{mathpar}
\caption{\eprhl\ proof pules}\label{fig:eprhl}
\end{figure*}

\begin{proof}[Proof of Theorem~\ref{thm:eprhl}]
  We adopt the convention that $f(\infty) \triangleq \infty$, even if $f$ is
  constant. The proof is by induction on the derivation.
  \begin{description}
    \item[Case \textsc{Assn}:]
      By definition, we have:
      \[
        \gcp{\Assn{x}{e}}{\Exp + [\neg Q] \cdot \infty}
        = \text{id} (\Exp \{ e\sidel, e\sider / x\sidel, x\sider \})
        + [ \neg Q \{ e\sidel, e\sider / x\sidel, x\sider \} ] \cdot \infty
      \]
    \item[Case \textsc{Rand*}:]
      By the proof rule for sampling (Proposition~\ref{prop:sampling-sound})
      taking the coupling function given by the bijection coupling $M(s_1, s_2)
      = M_h$, we have:
      \begin{align*}
        \gcp{\Rand{x}{[D]}}{\Exp + [\neg Q] \cdot \infty}
        &\leq \mathbb{E}_{v \sim \denot{[D]}} [ \Exp \{ v, h(v) / x\sidel, x\sider \}
        + [ \neg Q \{ v, h(v) / x\sidel, x\sider \} ] \cdot \infty ] \\
        &\leq \mathbb{E}_{v \sim \denot{[D]}} [ \Exp \{ v, h(v) / x\sidel, x\sider \} ] 
        + [ \forall v \in D.\ \neg Q \{ v, h(v) / x\sidel, x\sider \} ] \cdot \infty
      \end{align*}
      where the second inequality is because each element $v \in D$ has
      positive probability under $\denot{[D]}$, so if $\neg Q \{ v, h(v) /
      x\sidel, x\sider \}$ for some $v \in D$ then both sides are infinite.
    \item[Case \textsc{Seq}:]
      By induction, we have:
      \begin{align}
        \gcp{c'}{\Exp'' + [\neg R] \cdot \infty}
        &\leq f'(\Exp') + [\neg Q] \cdot \infty \tag{induction} \\
        \gcp{c}{\Exp' + [\neg Q] \cdot \infty}
        &\leq f(\Exp) + [\neg P] \cdot \infty \tag{induction}
      \end{align}
      Then, we can conclude:
      \begin{align}
        \gcp{c ; c'}{\Exp'' + [\neg R] \cdot \infty}
        &= \gcp{c}{\gcp{c'}{\Exp'' + [\neg R] \cdot \infty}} \tag{definition} \\
        &= \gcp{c}{f'(\Exp') + [\neg Q] \cdot \infty} \tag{monotonicity} \\
        &\leq f'(\gcp{c}{\Exp' + [\neg Q] \cdot \infty}) \tag{affine} \\
        &\leq f'(f(\Exp) + [\neg P] \cdot \infty) \tag{monotonicity} \\
        &= f' \circ f(\Exp) + [\neg P] \cdot \infty \notag
      \end{align}
    \item[Case \textsc{Cond}:]
      By induction, we have:
      \begin{align}
        \gcp{c}{\Exp' + [\neg Q] \cdot \infty}
        &\leq f(\Exp') + [\neg (P \land e\sidel)] \cdot \infty
        \tag{induction} \\
        \gcp{c'}{\Exp' + [\neg Q] \cdot \infty}
        &\leq f(\Exp') + [\neg (P \land \neg e\sidel)] \cdot \infty
        \tag{induction}
      \end{align}
      Then, we have:
      \begin{align}
        &\gcp{\Cond{e}{c}{c'}}{\Exp' + [\neg Q] \cdot \infty}
        \notag \\
        &= [e\sidel \land e\sider] \cdot \gcp{c}{\Exp' + [\neg Q] \cdot \infty}
        + [\neg e\sidel \land \neg e\sider] \cdot \gcp{c'}{\Exp' + [\neg Q] \cdot \infty}
        + [\neg (e\sidel = e\sider)] \cdot \infty
        \tag{definition} \\
        &\leq [e\sidel \land e\sider] \cdot (f(\Exp') + [\neg (P \land e\sidel)] \cdot \infty)
        + [\neg e\sidel \land \neg e\sider] \cdot (f(\Exp') + [\neg (P \land \neg e\sidel)] \cdot \infty)
        + [\neg (e\sidel = e\sider)] \cdot \infty
        \tag{induction} \\
        &\leq [e\sidel \land e\sider] \cdot (f(\Exp') + [\neg (P \land e\sidel)] \cdot \infty)
        + [\neg e\sidel \land \neg e\sider] \cdot (f(\Exp') + [\neg (P \land \neg e\sidel)] \cdot \infty)
        + [\neg P] \cdot \infty
        \tag{$\models P \to e\sidel = e\sider$} \\
        &= ([e\sidel \land e\sider] + [\neg e\sidel \land \neg e\sider]) \cdot f(\Exp') + [\neg P] \cdot \infty
        \tag{boolean alg.} \\
        &= f(\Exp') + [\neg P] \cdot \infty
        \tag{non-negative}
      \end{align}
    \item[Case \textsc{While}:]
      Let $n$ be any natural number. For any natural number $0 < m \leq n$, we
      write $f^{(m)} = f_1 \circ \cdots \circ f_{m}$ and we define $f^{(0)} =
      \text{id}$. We will show:
      \[
        \gcp{\WWhile{e}{c}}{\Exp_0 + [\neg (P \land v\sidel = 0)] \cdot \infty}
        \leq f^{(n)}(\Exp_n) + [\neg (P \land v\sidel = n)] \cdot \infty
      \]
      We take the following loop invariant:
      \[
        \Inv_n \triangleq \sum_{j = 0}^n ([ P \land v\sidel = j ] \cdot f^{(j)}(\Exp_j))
        + [\neg (P \land v\sidel \leq n)] \cdot \infty
      \]
      We claim that:
      \[
        [e\sidel \land e\sider] \cdot \gcp{c}{\Inv_n}
        + [\neg e\sidel \land \neg e\sider] \cdot (\Exp_0 + [\neg (P \land v\sidel = 0)] \cdot \infty)
        + [e\sidel \neq e\sider] \cdot \infty
        \leq \Inv_n .
      \]
      Both sides are infinite if $e\sidel \neq e\sider$, or $\neg (P \land
      v\sidel \leq n)$, or $\neg e\sidel \land \neg e\sider \land v \neq 0$. So,
      it suffices to prove:
      \[
        [P \land e\sidel \land e\sider \land v\sidel \leq n] \cdot \gcp{c}{\Inv_n}
        + [P \land \neg e\sidel \land \neg e\sider] \cdot \Exp_0
        \leq [P \land v\sidel \leq n] \cdot \Inv_n
        = \sum_{j = 0}^n [ P \land v\sidel = j ] \cdot f^{(j)}(\Exp_j) .
      \]
      If $P \land \neg e\sidel \land \neg e\sider$ holds, then $v\sidel = 0$
      holds by assumption. So by definition of $\Inv_n$:
      \[
        [P \land \neg e\sidel \land \neg e\sider] \cdot \Exp_0
        = [P \land v\sidel = 0] \cdot f^{(0)}(\Exp_0)
        \leq \Inv_n .
      \]
      If $P \land e\sidel \land e\sider \land v\sidel \leq n$ holds, then
      suppose that $v\sidel = k$ with $0 < k \leq n$. We have:
      \begin{align}
        &[P \land e\sidel \land e\sider \land v\sidel = k] \cdot \gcp{c}{\Inv_n}
        \notag \\
        &\leq [P \land e\sidel \land e\sider \land v\sidel = k] \cdot
        \gcp{c}{\sum_{j = 0}^n [ P \land v\sidel = j ] \cdot f^{(j)}(\Exp_j) + [\neg (P \land v\sidel = k - 1)] \cdot \infty}
        \tag{boolean alg.} \\
        &= [P \land e\sidel \land e\sider \land v\sidel = k] \cdot
        \gcp{c}{[ P \land v\sidel = k - 1 ] \cdot f^{(k - 1)}(\Exp_{k - 1}) + [\neg (P \land v\sidel = k - 1)] \cdot \infty}
        \tag{boolean alg.} \\
        &\leq [P \land e\sidel \land e\sider \land v\sidel = k] \cdot
        f^{(k - 1)}(\gcp{c}{\Exp_{k - 1} + [\neg (P \land v\sidel = k - 1)] \cdot \infty})
        \tag{affine} \\
        &\leq [P \land e\sidel \land e\sider \land v\sidel = k] \cdot
        f^{(k - 1)}(f_{k}(\Exp_k) + [\neg (P \land v\sidel = k)] \cdot \infty)
        \tag{induction} \\
        &= [P \land e\sidel \land e\sider \land v\sidel = k] \cdot f^{(k)}(\Exp_k)
        \tag{boolean alg.} \\
        &\leq [P \land e\sidel \land e\sider \land v\sidel = k] \cdot \Inv_n
        \tag{definition}
      \end{align}
      Above, we have used the induction hypothesis:
      \begin{align*}
        \gcp{c}{\Exp_{k - 1} + [\neg (P \land v\sidel = k - 1)] \cdot \infty}
        &\leq f_k(\Exp_k) + [\neg (P \land v\sidel = k)] \cdot \infty .
      \end{align*}
      By the loop rule, we can conclude:
      \[
        \gcp{\WWhile{e}{c}}{\Exp + [\neg (P \land v\sidel = 0)] \cdot \infty}
        \leq \Inv_n \leq f^{(n)}(\Exp_n) + [\neg (P \land v\sidel = n)] \cdot \infty .
      \]
    \item[Case \textsc{Conseq}:]
      By induction, we have:
      \begin{align}
        \gcp{c}{\Exp''' + [\neg Q'] \cdot \infty}
        &= \gcp{c}{[Q] \cdot \Exp''' + [\neg Q'] \cdot \infty}
        \tag{boolean alg.} \\
        &\leq \gcp{c}{\Exp' + [\neg Q'] \cdot \infty}
        \tag{monotonicity} \\
        &\leq f(\Exp) + [\neg P] \cdot \infty
        \tag{induction} \\
        &\leq f(\Exp) + [\neg P'] \cdot \infty
        \tag{assumption} \\
        &= [P'] \cdot f(\Exp) + [\neg P'] \cdot \infty
        \tag{boolean alg.} \\
        &\leq [P'] \cdot f'(\Exp'') + [\neg P'] \cdot \infty
        \tag{assumption} \\
        &\leq f'(\Exp'') + [\neg P'] \cdot \infty
        \notag
      \end{align}
    \item[Case \textsc{Case}:]
      By induction, we have:
      \begin{align}
        \gcp{c}{\Exp' + [\neg Q] \cdot \infty}
        &\leq f(\Exp') + [\neg (P \land R)] \cdot \infty
        \tag{induction} \\
        \gcp{c}{\Exp' + [\neg Q] \cdot \infty}
        &\leq f(\Exp') + [\neg (P \land \neg R)] \cdot \infty
        \tag{induction}
      \end{align}
      Thus, we have:
      \begin{align}
        \gcp{c}{\Exp' + [\neg Q] \cdot \infty}
        &\leq f(\Exp') + \min([\neg (P \land R)], [\neg (P \land \neg R)]) \cdot \infty
        \tag{induction} \\
        &\leq f(\Exp') + [\neg P] \cdot \infty
        \tag{boolean alg.}
      \end{align}
    \item[Case \textsc{Frame-D}:]
      We have:
      \begin{align}
        \gcp{c}{\Exp' + \Exp'' + [\neg Q] \cdot \infty}
        &\leq \gcp{c}{\Exp' + [\neg Q] \cdot \infty} + \Exp''
        \tag{frame} \\
        &\leq f(\Exp) + [\neg P] \cdot \infty + \Exp''
        \tag{induction} \\
        &\leq f(\Exp + \Exp'') + [\neg P] \cdot \infty
        \tag{$f$ expanding}
      \end{align}
  \end{description}
\end{proof}

\section{Convergence of TD(0): Omitted details}
\label{sec:td0}

We start by analyzing the inner loop~$w_{in}$. We first show that
\[
    \gcp{w_{in}}{\| W\sidel - W\sider \|_\infty} \leq \cI_{in}
\]
for the invariant $\cI_{in}$:
\begin{multline*}
	\cI_{in} \triangleq \ind{i\sidel \neq i\sider} \cdot \infty \\
                      + \ind{i\sidel = i\sider} \cdot \max_{l < |S|} (\ind{l < i\sidel} \cdot | W\sidel[l] - W\sider[l] |
                      + \ind{i\sidel\leq l}\cdot k\cdot \| V\sidel - V\sider \|_\infty ).
\end{multline*}
Let $c_{in}$ be the body, and $c_{samp}$ be the three sampling statements.
Applying \textsc{Inv}, it suffices to show:
\[
  \ind{i\sidel < |\cS| \land i\sider < |\cS|} \cdot \gcp{c_{in}}{\Inv_{in}}
  + \ind{i\sidel \geq |\cS| \land i\sider \geq |\cS|} \cdot \| W\sidel - W\sider \|_\infty
  + \ind{i\sidel \neq i\sider} \cdot \infty \leq \Inv_{in}
\]
We describe how to bound the key part of the invariant,
$\gcp{c_{in}}{\Inv_{in}}$ in the first term; the other cases are simpler.
Computing the pre-expectation for the last two instructions gives us
\[
  \gcp{c_{in}}{\Inv_{in}} \leq \gcp{c_{samp}}{\ind{i\sidel = i\sider} \cdot \mathcal{J}} ,
\]
where $\mathcal{J}$ is the following relational expectation:
\[\small \max_{l < |S|} 
    \left( \begin{array}{l}
\ind{l < i\sidel} \cdot | W\sidel[l] - W\sider[l] | \\
+ \ind{l = i\sidel}\cdot | \, (1 - \alpha) \cdot (V[i]\sidel - V[i]\sider)
  + \alpha \cdot (r\sidel - r\sider + \gamma \cdot (V[j]\sidel - V[j]\sider) ) \, | \\
+ \ind{i\sidel + 1 \leq l}\cdot k\cdot \| V\sidel - V\sider \|_\infty
    \end{array}\right).
\]
By taking an appropriate coupling, we will show that
$\gcp{c_{samp}}{\mathcal{J}}$ is at most $\Inv_{in}$. For sampling $j$, we
take the coupling function where if $a\sidel = a\sider$, then we take the
identity coupling ensuring $j\sidel = j\sider$, otherwise we take the product
coupling. We take the same coupling for sampling $r$. Finally for sampling $a$,
we take the identity coupling ensuring $a\sidel = a\sider$. When $i\sidel =
i\sider$, $j\sidel = j\sider$, and $r\sidel = r\sider$. By applying rule
\textsc{Samp}, we can upper-bound $\gcp{c_{in}}{\Inv_{in}}$ by $\ind{i\sidel =
i\sider = i}$ times:
\begin{equation} \label{eq:td0:big}
  \rho(v_a, v_r, v_j, i)
  \cdot \max_{l < |S|} 
  \left( \begin{array}{l}
      \ind{l < i} \cdot | W\sidel[l] - W\sider[l] | \\
      + \ind{l = i}\cdot | \, (1 - \alpha) \cdot (V[i]\sidel - V[i]\sider)
      + \alpha \cdot (\gamma \cdot (V\sidel[v_j] - V\sider[v_j]) ) \, | \\
      + \ind{i + 1 \leq l}\cdot k\cdot \| V\sidel - V\sider \|_\infty
  \end{array}\right).
\end{equation}
taking a sum over triples $(v_a, v_r, v_j) \in \cA \times \cR \times \cS$ and
$\rho(v_a, v_r, v_j, i)$ is the probability of drawing $v_a, v_r, v_j$ in current
state $i$; note that for any fixed $i < |\cS|$, the coefficients sum to $1$.

Now for any $l < |\cS|$ and any input states, at most one of the three summands
in the max is non-zero. We can bound the first and last summands:
\begin{align}
  \ind{l < i} \cdot | W\sidel[l] - W\sider[l] |
  &\leq \Inv_{in}
  \tag{since $l < i$} \\
  \ind{i + 1 \leq l}\cdot k\cdot \| V\sidel - V\sider \|_\infty
  &\leq \Inv_{in}
  \tag{since $i \leq l$}
\end{align}
For the second summand, we have:
\begin{align} 
  &\ind{l = i} \cdot | \, (1 - \alpha) \cdot (V\sidel[i] - V\sider[i]) + \alpha\gamma \cdot (V\sidel[v_j] - V\sider[v_j]) \, |
  \notag \\
  &\leq \ind{l = i}  \cdot | \, (1 - \alpha) \cdot \|V\sidel - V\sider\|_\infty + \alpha\gamma \cdot \|V\sidel - V\sider\|_\infty  \, |
  \notag \\
  &\leq \ind{l = i} \cdot k \cdot \|V\sidel - V\sider\|_\infty
  \notag \\
  &\leq \Inv_{in} \tag{since $i \leq l$}~,
\end{align}
Putting everything together, we have:
\begin{align*}
  \gcp{c_{in}}{\Inv_{in}}
  &\leq \ind{i\sidel = i\sider = i} \cdot \sum_{v_a, v_r, v_j} (\text{\cref{eq:td0:big}}) \\
  &\leq \ind{i\sidel = i\sider = i} \cdot \sum_{v_a, v_r, v_j} \rho(v_a, v_r, v_j, i) \cdot \Inv_{in} \\
  &= \Inv_{in} ,
\end{align*}
establishing the inner invariant.

\section{Verifying robustness of Projected Gradient Descent (PGD)}
\label{sec:pgd}

\begin{figure*}
\[
\begin{array}{l}
  \mathbf{pgd}(w_0, \alpha, T): \\
  \quad \Assn{w}{w_0}; \\
  \quad \Assn{t}{1}; \\
  \quad \WWhile{t < T}{} \\
  \quad\quad \Assn{g}{\nabla\ell(z, -)(w)}; \\
  \quad\quad \Assn{w}{\Pi_\Omega(w - \alpha_t \cdot g)}; \\
  \quad\quad \Assn{t}{t + 1}; 
\end{array}
\]
\caption{Projected Gradient Descent (PGD)} \label{fig:pgd}
\end{figure*}

This example is inspired by an analysis by Miller and Hardt~\cite{MillerH18}.
Let $\Omega \subseteq \RR^d$ be a compact and convex set of feasible parameters,
and let $\Pi_\Omega : \RR^d \to \Omega$ be the Euclidean projection sending a
point from $\RR^d$ to the closest point in $\Omega$ under the Euclidean
distance.  Given a loss function $\ell$, initial parameters $w_0$, and a
sequence of step sizes $\{ \alpha_t \}_t$, the program \textbf{pgd} in
Figure~\ref{fig:pgd} runs projected gradient descent for $T$ iterations.

Consider running this algorithm with two different loss functions
$\ell\sidel$ and $\ell\sider$, satisfying the following conditions:
\begin{enumerate}
  \item Gradients are close. For any parameter $w \in \RR^d$,
    \[
      \| \nabla \ell\sidel(z, -)(w) - \nabla \ell\sider(z, -)(w) \| \leq \gamma .
    \]
  \item Gradient of loss function is Lipschitz. For any two parameters $w, w'
    \in \RR^d$,
    \[
      \| \nabla \ell\sidel(z, -)(w) - \nabla\ell \sidel(z, -)(w') \| \leq \beta \| w - w' \| .
    \]
\end{enumerate}
Taking the step sizes $\alpha_t \leq \alpha / t$, we can bound the distance
between final weights $\| w\sidel - w\sider \|$ from running projected gradient
descent on the loss functions $\ell\sidel$ and $\ell\sider$ by showing the
following bound on the relational pre-expectation, which matches the analysis of \citet{MillerH18}:
\[
  \gcp{\mathbf{pgd}(w_0, \alpha, T)}{\| w\sidel - w\sider \|} \leq \alpha \gamma T^{\alpha \beta + 1}
\]
Intuitively, this property means that small changes to the loss function in PGD
do not lead to large changes in the learned parameters.

To start the proof, we take the following loop invariant:
\begin{align*}
  \Inv &\triangleq \ind{t\sidel \neq t\sider} \cdot \infty \\
       &+ \ind{t\sidel = t\sider} \cdot \|w\sidel - w\sider\| \prod_{j = t\sidel}^{T} (1 + \alpha_j \beta) \\
       &+ \ind{t\sidel = t\sider} \cdot \sum_{s = t\sidel}^{T} \alpha_s \gamma \prod_{j = s + 1}^{T} \exp(1 + \alpha_j \beta)
\end{align*}
To apply the loop rule, we need to check
\[
  \ind{(t < T)\sidel \land (t < T)\sider} \gcp{c}{\Inv}
  + \ind{(t \geq T)\sidel \land (t \geq T)\sider} \Exp
  + \ind{(t < T)\sidel \neq (t < T)\sider} \cdot \infty
  \leq \Inv .
\]
The main case is when both loop guards are true and when both loop counters are
equal. Taking the relational pre-expectation for the loop body in this case, we
have:
\begin{align*}
&\gcp{c}{\Inv}= \\ &= \| \Pi_\Omega(w - \alpha_t \cdot \nabla \ell(z, -)(w)) \sidel - \Pi_\Omega(w - \alpha_t \cdot \nabla \ell(z, -)(w)) \sider \|
	\hspace{-0.5em}  \prod_{j = t\sidel + 1}^{T} \!\!(1 + \alpha_j \beta)
                +\hspace{-0.5em} \sum_{s = t\sidel + 1}^{T} \!\!\alpha_s \gamma \prod_{j = s + 1}^{T} (1 + \alpha_j \beta) \\
                &\leq (\| w\sidel - w\sider \| + \alpha_t \|\nabla \ell(z, -)(w)\sidel - \nabla \ell(z, -)(w)\sider \|)
                \prod_{j = t\sidel + 1}^{T} (1 + \alpha_j \beta)
                + \sum_{s = t\sidel + 1}^{T} \alpha_s \gamma \prod_{j = s + 1}^{T} (1 + \alpha_j \beta) \\
                &\leq (\| w\sidel - w\sider \| + \alpha_t \beta \| w\sidel - w\sider \| + \alpha_t \gamma)
                \prod_{j = t\sidel + 1}^{T} (1 + \alpha_j \beta)
                + \sum_{s = t\sidel + 1}^{T} \alpha_s \gamma \prod_{j = s + 1}^{T} (1 + \alpha_j \beta) \\
                &= \Inv
\end{align*}
Pushing the invariant past the initial assignment instructions and taking the
same step sizes $\alpha_t \leq \alpha / t$ as Miller and Hardt~\cite{MillerH18},
we conclude:
\begin{align*}
  \gcp{\mathbf{pgd}(w_0, \alpha, T)}{\| w \sidel - w\sider \|}
  &\leq \sum_{s = 1}^{T} \alpha_s \gamma \prod_{j = s + 1}^{T} (1 + \alpha_j \beta) \\
  &\leq \sum_{s = 1}^{T} \alpha_s \gamma \prod_{j = s + 1}^{T} \exp(\alpha_j \beta) \\
  &\leq \sum_{s = 1}^{T} \frac{\alpha \gamma}{s} \prod_{j = s + 1}^{T} \exp \left( \frac{\alpha \beta}{j} \right) \\
  &= \sum_{s = 1}^{T} \frac{\alpha \gamma}{s} \exp \left( \alpha \beta\sum_{j = s + 1}^{T} \frac{1}{j} \right) \\
  &\leq \sum_{s = 1}^{T} \frac{\alpha \gamma}{s} \exp \left( \alpha \beta \log(T/s) \right) \\
  &\leq \alpha \gamma T^{\alpha \beta} \sum_{s = 1}^{T} \frac{1}{s^{\alpha \beta + 1}} \\
  &\leq \alpha \gamma T^{\alpha \beta + 1} .
\end{align*}

\section{Random-to-top: Omitted details}

\paragraph*{Axioms}
We assume a few axioms about the ${\sf shiftR}$ operation.  Let $a_1, a_2$ be
two decks and $J$ such that $\forall i. (0\leq i \leq J) \Rightarrow a_1[i] =
a_2[i]$. Then,
\begin{itemize}
	\item If $j \leq J$ and $a_i' =  {\sf shiftR}(a_i, j)$, then 
$\forall i. (0\leq i \leq J) \Rightarrow a_1'[i] = a_2'[i]$
	\item If $j_1,j_2 > J$, $a_i' =  {\sf shiftR}(a_i, j_i)$, and $a_1[j_1] = a_2[j_2]$, then
$\forall i. (0\leq i \leq J+1) \Rightarrow a_1'[i] = a_2'[i]$
\end{itemize}
Additionally, if $a$ is a permutation of $[N]$, then, for all $i<N$, so is ${\sf shiftR}(a,i)$.

\paragraph*{Establishing the invariant}
Let $C \triangleq (N-1)/N$. Recall the loop invariant:
\[ \Inv \triangleq \ind{k\sidel \neq k\sider} \cdot \infty + 
\ind{k\sidel = k\sider} \cdot d_M \cdot C^{\max(0, K-k\sidel)}  \]
We check that it satisfies the loop rule:
\[
  \ind{k\sidel < K \land k\sider < K} \cdot \gcp{c}{\Inv}
  + \ind{k\sidel \geq K \land k\sider \geq K} \cdot F
  + \ind{(k\sidel < K) \neq (k\sider < K)} \cdot \infty 
  \leq \Inv ,
\]
If $\ind{k\sidel \neq k\sider} \infty$ then the right-hand side of the
inequality is $\infty$, and it is satisfied.  Otherwise, if $\ind{k\sidel \geq K
\land k\sider \geq K}$ then we need to check that, indeed,
\[ 
	F \leq d_M \cdot C^{\max(0, K-k\sidel)} = d_M
\]
Finally, if $\ind{k\sidel < K \land k\sider < K}$, we compute the pre-expectation
of the loop body with respect to $\Inv$. Let $\Inv'$ be the pre-expectation of the
loop body without the sampling, i.e.,
\begin{align*}
	\Inv' &\triangleq \ind{k\sidel+1 \neq k\sider+1} \cdot \infty \\ 
	&+ \ind{k\sidel+1 = k\sider+1}
  \cdot (1/N) \cdot \left(N - \max_{i} \left(\forall j<i. a'\sidel[j] = a'\sider[j]\right) \cdot
  C^{\max(0, K-k\sidel-1)}\right)
\end{align*}
where $a'\sidel = {\sf shiftR}(a\sidel, y\sidel)$ and $a'\sider = {\sf
shiftR}(a\sider, y\sider)$.  In the following, let $l'$ denote $\max_{i} \left(\forall
j<i a'\sidel[j] = a'\sider[j]\right)$.  We pick a coupling induced by a bijection
$\pi$ such that, for all $z$, $a\sidel[z\sidel] = a\sider[\pi(z\sidel)]$. The
pre-expectation induced by this assignment is:
\begin{align*}
	\Inv'' & \triangleq \ind{k\sidel+1 \neq k\sider+1} \cdot \infty \\
	       & + \ind{k\sidel+1 = k\sider+1}
  \cdot (1/N) \cdot \left(N - \max_{i} (\forall j<i. a''\sidel[j] = a''\sider[j]) \cdot
C^{\max(0, K-k\sidel-1)}\right)
\end{align*}
where $a''\sidel = {\sf shiftR}(a\sidel, y)$ and $a'\sider = {\sf shiftR}(a\sider, \pi(y))$.

Now we have to compute the expected value of $\Inv''$ when we sample $y$
uniformly from $\unif{[N]}$.  There are two cases. If $y < l'$, then $\pi(y) = y$, and
$a\sidel[y] = a\sider[y]$, and
\[
\max_{i} (\forall j<i. a''\sidel[j] = a''\sider[j]) = \max_{i}(\forall j<i. a\sidel[j] = a\sider[j])
\]
where we use the first axiom of \textsf{shiftR}.  The probability of this
happening is precisely $l'/N = 1 - d_M$. In the other case, by the second axiom
of \textsf{shiftR}
\[
        \max_{i} (\forall j<i. a''\sidel[j] = a''\sider[j]) + 1 \leq \max_{i} (\forall j<i. a\sidel[j] = a\sider[j])
\]
This case happens with probability $d_M$. The inequality arises from the fact
that we may have matches below $l'$. From the expression above we derive:
\begin{align*}
	(1/ N)\left(N - \max_{i} (\forall j<i. a''\sidel[j] = a''\sider[j])\right) 
	& \leq (1/ N)\left(N - \max_{i} (\forall j<i. a\sidel[j] = a\sider[j]) - 1\right) \\
 	&= d_M - 1/N
\end{align*}
Using this inequality, we can bound the pre-expectation of the loop invariant
(simplifying under the assumptions $\ind{k\sidel \geq K \land k\sider \geq K}$
and $\ind{k\sidel = k\sider}$):
\begin{align*}
  \EE_{\Rand{y}{\unif{[N]}}}[\Inv''] &\leq
	  (1-d_M) \cdot d_M \cdot C^{K-k\sidel-1} + d_M \cdot (d_M-1/N) \cdot C^{K-k\sidel-1} \\
	  &= C^{K-\sidel - 1} \cdot d_M \cdot ( (1-d_M) + (d_M-1/N)) \\
	  &= C^{K-\sidel - 1} \cdot d_M \cdot C \\
	  &= C^{K-\sidel} \cdot d_M = \Inv
\end{align*}
This finishes the proof of the premise of the loop rule. Note that we did not
explicitly compute the pre-expectation of the loop invariant, we just found an
upper bound which is enough to apply the loop rule.

\section{Uniform riffle: Omitted details}

\paragraph*{Axioms}
We use some axioms about permutations, filtering, and concatenation.
\begin{itemize}
  \item Let $\mathsf{perm}(a_1, a_2)$ be the predicate that $a_1$ and $a_2$ are
    permutations of $C$. Then if we split a deck into two pieces and concatenate
    them, the result is a permutation of the original. Formally, for any
    bit-vector $b$ we have:
    \[
      \mathsf{perm}(a, \mathsf{cat}(a(\bar{b}), a(b)))
    \]
  \item Let $a_1, a_2$ be permutations, $b_1, b_2$ be bitstrings, and $a_1',
    a_2'$ be
    \[
      a_i' = \mathsf{cat}(a_i(\bar{b_i}), a_i(b_i)) .
    \]
    Then if $b_1, b_2$ match cards in $a_1, a_2$, i.e., $b_1 \circ a_1^{-1} =
    b_2 \circ a_2^{-1}$, then we can bound the size of blocks in the block
    decomposition of $a_1', a_2'$ as:
    \[
      \forall c \in [C].\, |BD(a_1', a_2')(c)|
      \leq \bar{b}(a_1^{-1}(c)) (\bar{b}(BD(a_1, a_2)(c))) + b(a_1^{-1}(c)) (b(BD(a_1, a_2)(c)))
    \]
    where we write $b(P)$ and $\bar{b}(P)$ to mean the total number of ones in
    $b$ and $\bar{b}$ at the positions $P$.
  \item Summing the previous bound over all cards gives:
    \[
      \sum_{c \in C} |BD(a_1', a_2')(c)|
      \leq \sum_{[c] \in BD(a_1, a_2)} \bar{b}(BD(a_1, a_2)(c))^2 + b(BD(a_1, a_2)(c))^2
    \]
    where the right-hand side sums over the equivalence classes of
    cards/positions induced by the block decomposition.
\end{itemize}

\paragraph*{Defining the distance}

Defining the distance between decks requires some care.
Consider the following distance based on positions:
\[
  d_P(\mathit{deck}_1, \mathit{deck}_2) \triangleq (1/N^2) \sum_{c \in C} |\mathit{deck}_1^{-1}(c) - \mathit{deck}_2^{-1}(c)|
\]
This distance measures the total difference between the positions of each card
in $\mathit{deck}_1$ and its counterpart in $\mathit{deck}_2$, normalized to be
in $[0, 1]$; and $d_P = 0$ holds only when
$\mathit{deck}_1 = \mathit{deck}_2$.  However, it is not easy to directly show
that this distance is monotonically decreasing in expectation---indeed, some
terms in the sum may actually increase. Instead, we define an upper bound
$d_c$ on $|\mathit{deck}_1^{-1}(c) - \mathit{deck}_2^{-1}(c)|$ for every card.
The sum $d_M \triangleq 1/N^2 \sum_{c \in C} d_c$ will be an upper bound of
$d_P$, and $d_M$ decreases monotonically to zero.

We will define $d_c$ in terms of a few concepts from the theory of permutations.
Given two decks $\mathit{deck}_1, \mathit{deck}_2$ and a permutation $\pi$ on
positions taking $\mathit{deck}_1$ to $\mathit{deck}_2$, there is a unique
\emph{cyclical decomposition} of $\pi$, i.e., we can partition the positions
into $P_1, \dots, P_k$ such that $\pi$ moves positions in $P_i$ as a single
cycle. We define a \emph{block decomposition} of $\pi$ to be a partition of
the positions $B_1, \dots, B_j$ such that each block is contiguous, and $\pi$
acts as a permutation on each $B_i$. A block decomposition is \emph{minimal} if
no block can be further decomposed; it is not hard to show that a minimal block
decomposition must be unique. When $\mathit{deck}_1, \mathit{deck}_2$ are
permutations, we write $BD(\mathit{deck}_1, \mathit{deck}_2)$ for the block
decomposition induced by two decks $\mathit{deck}_1$ and $\mathit{deck}_2$.
Finally, to define the distance, for every card $c \in C$ we let:
\[
  d_c \triangleq |BD(\mathit{deck}_1, \mathit{deck}_2)(c)| - 1
\]
where $|BD(\mathit{deck}_1, \mathit{deck}_2)(c)|$ is the size of the block
containing card $c$ in $\mathit{deck}_1$ and $\mathit{deck}_2$; both positions
must be in the same block. The size of each block is at least $1$, and if the
distance $d_c$ is zero then $c$ must be at the same position in
$\mathit{deck}_1$ and $\mathit{deck}_2$. It is not hard to show that the size of
the $c$'s block is at least the difference in $c$'s position across
$\mathit{deck}_1$ and $\mathit{deck}_2$:
\[
  |\mathit{deck}_1^{-1}(c) - \mathit{deck}_2^{-1}(c)| \leq d_c
\]
so $d_c = 0$ implies that $c$ is at the same position in $\mathit{deck}_1$ and
$a_2$.  (However, the reverse implication may not hold.) As a result, we can
upper bound our target distance
\[
  d_P \leq \frac{1}{N^2} \sum_{c \in C} d_c = d_M .
\]
Now, we turn to the loop. Let $\Phi$ be the binary invariant
\[
  \Phi \triangleq
    \mathsf{perm}(deck\sidel, deck\sider)
    \wedge k\sidel = k\sider
    \wedge (b \circ deck^{-1})\sidel = (b\circ deck^{-1})\sider
\]
and take the following invariant expectation:
\[
  \Inv = \ind{\neg\Phi} \cdot \infty + \ind{\Phi} \cdot d_M \cdot (1/2)^{(K - k\sidel)_+}
\]
We want to verify that:
\[
\begin{array}{rl}
  &\ind{(k < K)\sidel \land (k < K)\sider} \cdot \gcp{\mathit{bd}}{\Inv} \\
  + &\ind{(k \geq K)\sidel \land (k \geq K)\sider} \cdot d_P \\
  + &\ind{(k < K)\sidel \neq (k < K)\sider} \cdot \infty \qquad \qquad \leq \quad \Inv,
\end{array}
\]
where $\mathit{bd}$ is the loop body.
The cases $\ind{(k \geq K)\sidel \land (k \geq K)\sider}$ and $\ind{(k <
K)\sidel \neq (k < K)\sider}$ are almost immediate. The main case is when
$\ind{(k < K)\sidel \land (k < K)\sider}$. Focusing on the case where $\Phi$
holds (otherwise there is nothing to show), this boils down to:
\[
\begin{array}{l}
  \EE_b [d_M(\mathsf{cat}(deck(\bar{b}), deck(b))\sidel, \mathsf{cat}(deck(\bar{b}), deck(b))\sider)]
  \leq \frac{1}{2} d_M ,
\end{array}
\]
i.e., each iteration of the loop halves the invariant, where the expected value
is taken over $b\sidel \sim \{ 0, 1 \}^N$ and $b\sider$ is coupled so that $(b
\circ deck^{-1}) \sidel = (b \circ deck^{-1})\sider$.  Above, we write $d_M(x_1,
x_2)$ as shorthand for $d_M[x_1, x_2 / \mathit{deck}\sidel,
\mathit{deck}\sider]$.

The inequality follows from the permutation axioms, and from the mean and
variance of the binomial distribution---for $\mathit{deck}_1, \mathit{deck}_2$
fixed, $\bar{b}(BD(\mathit{deck}_1, \mathit{deck}_2))$ and
$b(BD(\mathit{deck}_1, \mathit{deck}_2))$ each follow the binomial distribution
with $|BD(\mathit{deck}_1, \mathit{deck}_2)(c)|$ trials and parameter $1/2$.
This completes the proof for the body of the loop. Finally, we push the
invariant past the initialization of the procedure, and we have the bound:
\[
  \gcp{\mathbf{riffle}(deck, N, K)}{d_P} \leq \ind{\neg\Phi} + \ind{\Phi} \cdot d_M \cdot (1/2)^K
                                         \leq \ind{\neg\Phi} + \ind{\Phi} \cdot (1/2)^K .
\]
since the initial distance $d_M$ is at most $1$. Given that $d_P$ assigns
different decks a distance of at least $1/N^2$, Theorem~\ref{thm:scaled-tv}
implies that the TV distance between the deck distributions is at most
\[
	v(K,N) = \max_{d_1,d_2 \in [N]} TV(\denot{\mathbf{riffle}}(d_1,N,K), \denot{\mathbf{riffle}}(d_2,N,K) ) 
	\leq N^2 \left( \frac{1}{2} \right)^K ,
\]
so the distributions converge to one another and to the uniform distribution
exponentially quick. If we take $K \geq \log_2 (N^2\rho) $, $v(K)$ is
asymptotically bounded by $O(1/\rho)$ for large $N$. When setting $\rho = N$,
we establish the following guarantee.

\begin{theorem}
	Let $K = 3 \log N$, and $\mathit{Perm}([N])$ be the set of permutations over $N$. 
	For any initial permutation of $\mathit{deck}$,
	\[TV(\mathbf{riffle}(\mathit{deck},N,K), \mathbf{Unif}\{\mathit{Perm}([N])\}) \in \mathcal{O}(1/N) \] 
\end{theorem}

\paragraph*{Establishing the invariant}
Recall that we need to show:
\[
  \EE_b [d_M(\mathsf{cat}(deck(\bar{b}), deck(b))\sidel, \mathsf{cat}(deck(\bar{b}), deck(b))\sider)]
  \leq \frac{1}{2} d_M(deck\sidel, deck\sider) ,
\]
i.e., each iteration of the loop halves the invariant, where the expected value
is taken over $b\sidel \sim \{ 0, 1 \}^N$ and $b\sider$ is coupled so that $(b
\circ deck^{-1}) \sidel = (b \circ deck^{-1})\sider$.  Writing $a_1, a_2 =
deck\sidel, \sider$, and $a_1', a_2' = \mathsf{cat}(deck(\bar{b}),
deck(b))\sidel, \sider$, and $b_1, b_2 = b\sidel, \sider$, the permutation
axioms give:
\begin{align*}
  \EE_{b} \left[ d_M(a_1', a_2') \right]
  &=  \frac{1}{N^2} \sum_{c \in C} \EE_{b} [ |BD(a_1', a_2')(c)| - 1] \\ 
  &\leq \frac{1}{N^2}
  \sum_{[c] \in BD(a_1, a_2)} \EE_{b} [ \bar{b}(BD(a_1, a_2)(c))^2 ] + \EE_b[
  b(BD(a_1, a_2)(c))^2 ] - |BD(a_1, a_2)(c)| \\
  &= \frac{1}{2 N^2}\sum_{[c] \in BD(a_1, a_2)} |BD(a_1, a_2)(c)|^2 - |BD(a_1, a_2)(c)| \\
  &= \frac{1}{2 N^2}\sum_{c \in C } (|BD(a_1, a_2)(c)| - 1) \\
  &= \frac{1}{2} d_M(a_1, a_2) .
\end{align*}

\section{Asynchronous rules: Omitted details} \label{sec:async-extra}

We prove soundness of the asynchronous rules.

\begin{proof}[Proof of \cref{thm:async}]
  We start with the rule for conditionals. Let $c$ be a program that is almost
  surely terminating, let $\Exp$ be a relational pre-expectation, and let $s_1,
  s_2 \in \Mem$ be two states. If $s_1(e) = s_2(e)$, then the bound follows from
  soundness of synchronous case (\cref{thm:sound}):
  \begin{align*}
    \sgcp{\Condt{e}{c}}{\Exp}(s_1, s_2)
    &\leq \gcp{\Condt{e}{c}}{\Exp}(s_1, s_2) \\
    &= (\ind{e\sidel \land e\sider} \cdot \gcp{c}{\Exp} + \ind{\neg e\sidel \land \neg e\sider}\cdot \Exp)(s_1, s_2) .
  \end{align*}
  Otherwise if $e$ is true in $s_1$ and false in $s_2$, then:
  \[
    \sgcp{\Condt{e}{c}}{\Exp}(s_1, s_2)
    = \inf_{\mu \in \Gamma(\denot{c}s_1, \dunit{s_2})} \EE_{\mu}[ \Exp ] 
    \leq \wpel{c}{\Exp}(s_1, s_2)
  \]
  by \cref{lem:wpe:sound}. The case where $e$ is false in $s_1$ and true in
  $s_2$ is almost identical.

  Next, we consider the asynchronous rule for loops. Let $\WWhile{e}{c}$ be
  almost surely terminating. We define a sequence of loop approximants:
  \begin{align*}
    c_0 &\triangleq \Skip \\
    c_{i+1} &\triangleq (\Condt{e}{c}) ; c_i
  \end{align*}
  When the loop is almost surely terminating, we have the following equivalence:
  \[
    \denot{\WWhile{e}{c}}s = \lim_{i \to \infty} (\denot{c_i}s)
  \]
  for any input state $s$, and the limit of distributions exists.

  Our overall argument proceeds much like the proof for the synchronous case. We
  first show that the least-fixed point of a characteristic function of the loop
  is an upper bound on pre-expectation. Then, we argue that the asynchronous
  loop rule shows that $\Inv$ is a fixed point with respect to the
  characteristic function, so it must also be an upper bound.
  We work with the following characteristic function:
  \begin{multline*}
    \Psi_{\Exp, c}(\Exp')
    \triangleq \ind{e\sidel \land e\sider} \cdot \gcp{c}{\Exp'}
    + \ind{\neg e\sidel \land \neg e\sider} \cdot \Exp \\
    + \ind{e\sidel \land \neg e\sider} \cdot \wpel{c}{\Exp'}
    + \ind{\neg e\sidel \land e\sider} \cdot \wper{c}{\Exp'} .
  \end{multline*}
  By \cref{lem:monotonic} and monotonicity of the weakest pre-expectation
  operator, the operator $\Psi_{\Exp, c}$ is monotone. Thus, the least
  fixed-point exists:
  \[
    \mathcal{L}_{\Exp, c} = {\sf lfp} X.\, \Psi_{\Exp, c}(X) .
  \]
  We can inductively define:
  \begin{align*}
    \Exp_0 &\triangleq \ind{\neg e\sidel \land \neg \sider} \cdot \Exp \\
    \Exp_{i + 1} &\triangleq \ind{e\sidel \land e\sider} \cdot \gcp{c}{\Exp_i}
                 + \ind{\neg e\sidel \land \neg e\sider} \cdot \Exp \\
                 &+ \ind{e\sidel \land \neg e\sider} \cdot \wpel{c}{\Exp_i}
                 + \ind{\neg e\sidel \land e\sider} \cdot \wper{c}{\Exp_i} .
  \end{align*}
  By definition $\Exp_i = \Psi^i_{\Exp, c}(\Exp_0)$, and by monotonicity
  $\Exp_i$ is a monotone increasing sequence. Furthermore, for any expectation
  $\Exp'$ we have $\Exp_0 \leq \Psi_{\Exp, c}(\Exp')$, hence $\Exp_0 \leq
  \mathcal{L}_{\Exp, c}$. By monotonicity of $\Psi_{\Exp, c}$, we have $\Exp_i
  \leq \mathcal{L}_{\Exp, c}$ for every $i$.

  We now prove an analogue of \cref{claim:couple-approx}.
  
  \begin{lem} \label{claim:couple-approx-async}
    For all $j \in \NN$ and $(s_1', s_2') \in \Mem \times \Mem$, there exists
    $\mu_{j, s_1', s_2'} \in \Gamma(\denot{c_j}s_1', \denot{c_j}s_2')$ such that
    \[
      \EE_{\mu_{j, s_1', s_2'}}[ \Exp ]
      \leq \Exp_j(s_1', s_2') + (\rho_j(s_1') + \rho_j(s_2')) \cdot M_j(\Exp, s_1', s_2')
    \]
    where $\rho_j(s)$ is the probability of $e$ being true in $\denot{c_j}s$,
    and:
    \[
      M_j(\Exp, s_1', s_2') =
      \max \{ \Exp(t_1, t_2) \mid t_1 \in \supp(\denot{c_j}s_1'), t_2 \in \supp(\denot{c_j}s_2') \}
      .
    \]
  \end{lem}
  \begin{proof}
    By induction on $j$. The base case $j = 0$ is clear, taking the coupling
    $\dunit{s_1', s_2'}$. For the inductive step, we have
    \begin{multline}
      \Exp_{j + 1} \triangleq \ind{e\sidel \land e\sider} \cdot \gcp{c}{\Exp_j}
      + \ind{\neg e\sidel \land \neg e\sider} \cdot \Exp \\
      + \ind{e\sidel \land \neg e\sider} \cdot \wpel{c}{\Exp_j}
      + \ind{\neg e\sidel \land e\sider} \cdot \wper{c}{\Exp_j} .
    \end{multline}
    There are four cases.

    \begin{description}
      \item[Case: $s_1'(e) = \mathit{ff}$ and $s_2'(e) = \mathit{ff}$.]
        In this case, $\denot{c_{j + 1}}s_1' = \dunit{s_1'}$ and $\denot{c_{j +
        1}}s_2' = \dunit{s_2'}$. We can define the coupling $\mu_{s_1', s_2'} =
        \dunit{s_1', s_2'} \in \Gamma(\denot{c_{j + 1}}s_1', \denot{c_{j +
        1}}s_2')$ and we are done, since
        \[
          \EE_{\mu_{s_1', s_2'}} [ \Exp ] = \Exp(s_1', s_2') = \Exp_{j + 1}(s_1', s_2') .
        \]
      \item[Case: $s_1'(e) = \mathit{tt}$ and $s_2'(e) = \mathit{tt}$.]
        In this case, $\denot{c_{j + 1}}s_1' = \denot{c; c_j}s_1'$ and
        $\denot{c_{j + 1}}s_2' = \denot{c; c_j}s_2'$. If $\Exp_{j + 1}(s_1',
        s_2')$ is infinite we are done, so suppose that it is finite. By
        \cref{thm:sound}, there exists a coupling $\nu_{s_1',
        s_2'} \in \Gamma(\denot{c}s_1', \denot{c}s_2')$ such that
        \[
          \EE_{\nu_{s_1', s_2'}} [ \Exp_j ] \leq \Exp_j(s_1', s_2') .
        \]
        By the induction hypothesis, there is a coupling $\mu_{j, t_1, t_2} \in
        \Gamma(\denot{c_j}t_1, \denot{c_j}t_2)$ such that 
        \[
          \EE_{\mu_{j, t_1, t_2}}[ \Exp ]
          \leq \Exp_j(t_1, t_2) + (\rho_j(t_1) + \rho_j(t_2)) \cdot M_j(\Exp, t_1, t_2) .
        \]
        Now, we can define the coupling for the $(j + 1)$-th approximants:
        \[
          \mu_{j + 1, s_1', s_2'} \triangleq \dbind{\nu_{s_1', s_2'}}{\mu_{j, -, -}}
        \]
        We first check the distance condition. By definition, we have:
        \begin{align*}
          \EE_{\mu_{j + 1, s_1', s_2'}} [ \Exp ]
    &= \EE_{(t_1, t_2) \sim \nu_{s_1', s_2'}} [ \EE_{\mu_{j, t_1, t_2}} [ \Exp ] ] \\
    &\leq \EE_{(t_1, t_2) \sim \nu_{s_1', s_2'}} [ \Exp_j(t_1, t_2) + (\rho_j(t_1) + \rho_j(t_2)) \cdot M_j(\Exp, t_1, t_2) ] \\
    &\leq \Exp_j(s_1', s_2') + (\rho_{j + 1}(s_1') + \rho_{j + 1}(s_2')) \cdot M_{j + 1}(\Exp, s_1', s_2') \\
    &\leq \Exp_{j + 1}(s_1', s_2') + (\rho_{j + 1}(s_1') + \rho_{j + 1}(s_2')) \cdot M_{j + 1}(\Exp, s_1', s_2') .
        \end{align*}
        The marginal condition is not hard to show, using the marginal properties
        of $\nu_{s_1', s_2'}$ and $\mu_{j, t_1, t_2}$ combined with the
        definition of approximants: since $e$ is true in $s_1'$ and $s_2'$, we
        have $\denot{c_{j + 1}}s_1' = \denot{c; c_j}s_1'$ and $\denot{c_{j + 1}}
        s_2' = \denot{c; c_j}s_2'$. The proof follows the case for sequential
        composition.
      \item[Case: $s_1'(e) = \mathit{tt}$ and $s_2'(e) = \mathit{ff}$.]
        In this case, $\denot{c_{j + 1}}s_1' = \denot{c; c_j}s_1'$ and
        $\denot{c_{j + 1}}s_2' = \denot{\Skip}s_2'$. By \cref{lem:wpe:sound},
        there exists a coupling $\nu_{s_1', s_2'} \in \Gamma(\denot{c}s_1',
        \denot{\Skip}s_2')$ such that
        \[
          \EE_{\nu_{s_1', s_2'}} [ \Exp_j ]
          \leq \wpel{c}{\Exp}(s_1', s_2')
          = \Exp_j(s_1', s_2') .
        \]
        By the induction hypothesis, there is a coupling $\mu_{j, t_1, t_2} \in
        \Gamma(\denot{c_j}t_1, \denot{c_j}t_2)$ such that 
        \[
          \EE_{\mu_{j, t_1, t_2}}[ \Exp ]
          \leq \Exp_j(t_1, t_2) + (\rho_j(t_1) + \rho_j(t_2)) \cdot M_j(\Exp, t_1, t_2) .
        \]
        Now, we can define the coupling for the $(j + 1)$-th approximants:
        \[
          \mu_{j + 1, s_1', s_2'} \triangleq \dbind{\nu_{s_1', s_2'}}{\mu_{j, -, -}}
        \]
        The distance and marginal conditions follow as in the previous case.
      \item[Case: $s_1'(e) = \mathit{ff}$ and $s_2'(e) = \mathit{tt}$.]
        In this case, $\denot{c_{j + 1}}s_1' = \denot{\Skip}s_1'$ and
        $\denot{c_{j + 1}}s_2' = \denot{c; c_j}s_2'$. By \cref{lem:wpe:sound},
        there exists a coupling $\nu_{s_1', s_2'} \in \Gamma(\denot{\Skip}s_1',
        \denot{c}s_2')$ such that
        \[
          \EE_{\nu_{s_1', s_2'}} [ \Exp_j ]
          \leq \wper{c}{\Exp}(s_1', s_2')
          = \Exp_j(s_1', s_2') .
        \]
        By the induction hypothesis, there is a coupling $\mu_{j, t_1, t_2} \in
        \Gamma(\denot{c_j}t_1, \denot{c_j}t_2)$ such that 
        \[
          \EE_{\mu_{j, t_1, t_2}}[ \Exp ]
          \leq \Exp_j(t_1, t_2) + (\rho_j(t_1) + \rho_j(t_2)) \cdot M_j(\Exp, t_1, t_2) .
        \]
        Now, we can define the coupling for the $(j + 1)$-th approximants:
        \[
          \mu_{j + 1, s_1', s_2'} \triangleq \dbind{\nu_{s_1', s_2'}}{\mu_{j, -, -}}
        \]
        The distance and marginal conditions follow as in the previous case.
        \qedhere
    \end{description}
  \end{proof}

  Thus, we may apply Lemma~\ref{claim:couple-approx-async} with input states
  $s_1, s_2$ and expectations $\Exp_i$ to produce a sequence of couplings
  $\mu_{i, s_1, s_2} \in \Gamma(\denot{c_i}s_1, \denot{c_i}s_2)$ such that
  \begin{align*}
    \EE_{\mu_{i, s_1, s_2}} [ \Exp ]
    &\leq \Exp_i(s_1, s_2) + (\rho_i(s_1) + \rho_i(s_2)) \cdot M_i(\Exp, s_1, s_2)
    \\
    &= \Psi^i_{\Exp, c}(\Exp_0)(s_1, s_2) + (\rho_i(s_1) + \rho_i(s_2)) \cdot M_i(\Exp, s_1, s_2) .
  \end{align*}
  By Theorem~\ref{thm:conv-couplings}, we can extract a subsequence
  $\mu'_{i,s_1, s_2}$ (with a corresponding subsequence $c_i'$ of $c_i$) from
  the sequence $\mu_{i, s_1, s_2}$ that converges monotonically to a coupling
  satisfying
  \[
    \tilde{\mu}_{s_1,s_2}
    \in \Gamma(\lim_{i \to \infty} \denot{c_i'}s_1, \lim_{i \to \infty} \denot{c_i'} s_2)
    = \Gamma(\denot{\WWhile{e}{c}} s_1, \denot{\WWhile{e}{c}} s_2) ,
  \]
  where the equality holds because the loop is almost surely terminating. All
  that remains to show is:
  \[
    \EE_{(s_1',s_2')\sim \tilde{\mu}_{s_1,s_2}}[\Exp(s_1',s_2')] \leq \gcp{\WWhile{e}{c}}{\Exp}(s_1,s_2) .
  \]
  We can compute:
  \begin{align}
    \EE_{(s_1',s_2')\sim \tilde{\mu}_{s_1,s_2}}[\Exp(s_1',s_2')] 
    &= \sum_{(s_1',s_2')\in \Mem\times\Mem} \Exp(s_1',s_2') \cdot \lim_{i \to \infty} \mu'_{i,s_1,s_2}(s_1',s_2')
    \notag \\
    &\leq \sum_{(s_1',s_2')\in \Mem\times\Mem} \lim_{i \to \infty} \Exp(s_1',s_2') \cdot \mu'_{i,s_1,s_2} (s_1',s_2')
    \tag{$\Exp$ may be $\infty$} \\
    &\leq \lim_{i \to \infty} \sum_{(s_1',s_2')\in \Mem\times\Mem} \Exp(s_1',s_2') \cdot \mu'_{i,s_1,s_2} (s_1',s_2')
    \tag{by Fatou's lemma} \\
    &= (\lim_{i \to \infty} \Psi^i_{\Exp,c}(\Exp_0))(s_1,s_2)
    + \lim_{i \to \infty} (\rho_i'(s_1) + \rho_i'(s_2)) \cdot M_i'(\Exp, s_1, s_2)
    \tag{subsequence} \\
    &\leq \mathcal{L}_{\Exp, c}(s_1,s_2) .
    \tag{bounded assumption} 
  \end{align}
  Finally, the premise of the asynchronous loop rule implies that $\Psi_{\Exp,
  c}(\Inv) \leq \Inv$, i.e., $\Inv$ is a pre-fixed-point of $\Psi_{\Exp, c}$.
  Since $\mathcal{L}_{\Exp, c}$ is the least fixed point, we have:
  \[
    \sgcp{\WWhile{e}{c}}{\Exp}(s_1, s_2)
    \leq \EE_{(s_1',s_2')\sim \tilde{\mu}_{s_1,s_2}}[\Exp(s_1',s_2')] 
    \leq \mathcal{L}_{\Exp, c}(s_1,s_2)
    \leq \Inv(s_1, s_2) .
    \qedhere
  \]
\end{proof}

\end{document}